\pdfoutput=1

\documentclass[
 twocolumn, 
  nofootinbib,showpacs,superscriptaddress,preprintnumbers,aps,reprint]{revtex4-2}

\usepackage[utf8]{inputenc}

\usepackage{graphicx, color}

\usepackage[svgnames,psnames]{xcolor}

\usepackage[colorlinks,citecolor=DarkGreen,linkcolor=FireBrick,linktocpage]{hyperref}

\usepackage{dcolumn}
\usepackage{bm}
\usepackage{amsmath}
\usepackage{amsthm}
\usepackage{ascmac}
\usepackage{xparse} 
\usepackage{mathtools} 
\usepackage{blkarray} 
\usepackage{amscd,verbatim}
\usepackage{amssymb}
\usepackage[all]{xy}
\usepackage{tikz}
\usepackage{tikz-cd}
\usepackage{pgf,pgfplots}
\usepackage{enumitem}
\usepackage{caption}
\usepackage{subcaption}
\usepackage{floatrow}
\usepackage{slashed}
\usepackage{comment}
\usepackage{footmisc}

\usepackage{listofitems} 
\usetikzlibrary{arrows.meta} 
\usepackage[outline]{contour} 
\contourlength{1.4pt}

\usetikzlibrary{decorations.pathmorphing}

\tikzset{>=latex} 

\colorlet{myred}{red!80!black}
\colorlet{myblue}{blue!80!black}
\colorlet{mygreen}{green!60!black}
\colorlet{myorange}{orange!70!red!60!black}
\colorlet{mydarkred}{red!30!black}
\colorlet{mydarkblue}{blue!40!black}
\colorlet{DarkGreen}{green!30!black}
\colorlet{FireBrick}{red!178!green!34!blue!34}

\tikzstyle{node}=[thick,circle,draw=myblue,minimum size=22,inner sep=0.5,outer sep=0.6]
\tikzstyle{node in}=[node,green!20!black,draw=mygreen!30!black,fill=mygreen!25]
\tikzstyle{node hidden}=[node,blue!20!black,draw=myblue!30!black,fill=myblue!20]
\tikzstyle{node convol}=[node,orange!20!black,draw=myorange!30!black,fill=myorange!20]
\tikzstyle{node out}=[node,red!20!black,draw=myred!30!black,fill=myred!20]
\tikzstyle{connect}=[thick,mydarkblue]
\tikzstyle{connect arrow}=[-{Latex[length=4,width=3.5]},thick,mydarkblue,shorten <=0.5,shorten >=1]
\tikzset{
  node 1/.style={node in},
  node 2/.style={node hidden},
  node 3/.style={node out},
}

\newcommand\p{\partial}

\usepackage[normalem]{ulem}  

\renewcommand\sout{\bgroup \color{red} \ULdepth=-.5ex \ULset}

\def\changed#1{{#1}}

\newcommand\I{ {\rm I} }
\newcommand\II{ {\rm I\hspace{-1.1pt}I} }

\theoremstyle{plain}
\newtheorem{thm}{Theorem}[section]
\newtheorem{prop}[thm]{Proposition}

\theoremstyle{definition}

\begin{document}
\preprint{KUNS-3006}

\title{
Data-driven discovery of self-similarity using neural networks
}

\author{Ryota~Watanabe}
\thanks{watanabe@gauge.scphys.kyoto-u.ac.jp}
\affiliation{Department of Physics, Kyoto University, Kyoto 606-8502, Japan}

\author{Takanori~Ishii}
\thanks{ishiit@gauge.scphys.kyoto-u.ac.jp}
\affiliation{Department of Physics, Kyoto University, Kyoto 606-8502, Japan}

\author{Yuji~Hirono}
\thanks{Contact author. yuji.hirono@gmail.com}
\affiliation{Department of Physics, Osaka University, Toyonaka, Osaka 560-0043, Japan}

\author{Hirokazu~Maruoka}

\thanks{Contact author. \\ 
 hirokazu.maruoka@oist.jp, hmaruoka1987@gmail.com}
\affiliation{Advanced Statistical Dynamics, Yukawa Institute for Theoretical Physics, Kyoto University, Kitashirakawa Oiwake-Cho, Sakyo-ku, Kyoto 606-8502, Japan}
\affiliation{Nonlinear and Non-equilibrium Physics Unit, Okinawa Institute of Science and Technology (OIST), Tancha, Onna-son, Kunigami-gun Okinawa 904-0495,
Japan}

\begin{abstract}
Finding self-similarity is a key step for understanding the governing law behind complex physical phenomena.
Traditional methods for identifying self-similarity often rely on specific models, which can introduce significant bias. 
In this paper, we present a novel neural network-based approach that discovers self-similarity directly from observed data, without presupposing any models. 
The presence of self-similar solutions in a physical problem signals that the governing law contains a function whose arguments are given by power-law monomials of physical parameters, which are characterized by power-law exponents. 
The basic idea is to enforce such particular forms structurally in a neural network in a parametrized way. 
We train the neural network model using the observed data, and when the training is successful, we can extract the power exponents that characterize scale-transformation symmetries of the physical problem.
We demonstrate the effectiveness of our method with both synthetic and experimental data, validating its potential as a robust, model-independent tool for exploring self-similarity in complex systems.
\end{abstract}

\maketitle

\tableofcontents

\section{ Introduction }

The concept of {\it self-similarity} is of crucial importance in statistical physics, especially in understanding critical phenomena and phase transitions~\cite{goldenfeld2018lectures}.
Beyond its role in equilibrium physics, self-similarity emerges as a vital concept in deciphering the dynamics of various non-equilibrium phenomena~\cite{barenblatt_1996,barenblatt_2003}
such as 
soft matter physics~\cite{deGennesScaling}, 
complex fluids~\cite{Baumchen}, 
granular physics~\cite{Hatano,Saitoh}, and fluid dynamics~\cite{Barenblatt_2014}.
Self-similarity implies that certain properties of physical systems remain invariant under a scale transformation~\cite{Mandelbrot},
revealing an underlying universality across different scales. 
This feature not only facilitates a deeper understanding of the dynamical properties of materials but also aids the development of theoretical models that capture the essence of a physical phenomenon.
Solutions with this feature are called {\it self-similar solutions}.
Typically, self-similar solutions have the following structure:
\begin{equation}
\frac{y}{t^{q}} = f\left(\frac{x}{t^{p}} \right),
\label{eq:y-f-x}
\end{equation}
where $x$, $y$ and $t$ are physical parameters and $f$ is some smooth function.
Equation~\eqref{eq:y-f-x} is invariant under the scale transformation of the form: $t\to At,\, x \to A^p x,\, y \to A^q y$.
The essence of self-similar solutions lies in the appearance 
of {\it similarity parameters} such as $x/t^{p}$ and $y/t^{q}$, which are power-law monomials 
specified by their exponents, $p$ and $q$.
These exponents characterize
the scale-transformation symmetries 
that give rise to self-similar solutions.
Once the correct exponents are identified, their self-similar structure can be exploited to get a {\it data collapse}, in which all data points lie on a low-dimensional manifold by choosing the similarity parameters as the axes of the data plot~\cite{Stanley_1999,Cabella,Kimchi,Yokota,Okumura2020}.
The dimension of the manifold is determined by the number of the independent similarity parameters appearing in the self-similar solution.

For given experimental or simulated data, 
the identification of these power exponents, crucial for discovering the self-similar structure of a physical phenomenon, has traditionally been pursued through empirical methods.
While dimensional analysis serves as a common technique for determining the exponents, 
self-similar solutions found this way are
restricted to the cases of the {\it similarity of the first kind}, to which only specific problems belong.
However, there can be further similarity
that cannot be identified by dimensional analysis.
Such cases are referred to as the {\it similarity of the second kind}~\cite{barenblatt_1996,barenblatt_2003}, and many problems fall into this category, reflecting the nontrivial scale-invariance of the underlying physical laws.
To find the similarity of the second kind, theoretical techniques based on nonlinear eigenvalue problems~\cite{barenblatt_1996,barenblatt_2003} or renormalization method~\cite{goldenfeld1989intermediate} 
have been utilized. 
However, these methods often rely on pre-existing knowledge of the physical laws at play, potentially introducing significant bias through assumed models.
Therefore, there is a need for a more robust methodology that can discover self-similarity directly from data, independent of predefined models.

In this paper, we introduce a model-independent method designed to uncover the self-similarity based on observed data, using neural networks.
Deep neural networks~\cite{Goodfellow-et-al-2016,bishop2024deep}
provide a versatile way of parametrizing and optimizing functions of various forms. 
Indeed, universal approximation theorems \cite{cybenko1989approximation,HORNIK1989359} ensure that sufficiently complex neural networks are capable of approximating a wide variety of functions.
In recent years, neural networks have proven beneficial in tackling a variety of physical problems~\cite{wang2023scientific}.
For instance, they have been employed as variational wave functions in quantum many-body problems~\cite{doi:10.1126/science.aag2302}, 
used to enforce laws of physics during optimization processes~\cite{RAISSI2019686}, 
and utilized to detect phase transitions~\cite{doi:10.7566/JPSJ.86.063001}. 
In our approach, we configure neural networks to incorporate symmetry under scale transformations in a parametrized manner.
By optimizing the neural networks based on the provided data, the trained parameters encode the self-similarity inherent to the problem\footnote{We provide source codes implementing the method in Ref.~\cite{repository}}.
This model-independent approach not only pinpoints the scale-transformation symmetries essential for understanding the phenomenon but also significantly narrows down the spectrum of plausible theoretical models.
The approach we propose can be seen as an extension of the method proposed in Ref.~\cite{Somendra_M_Bhattacharjee_2001}, 
in which data collapse is detected through the minimization of a parametrized loss function. The strategy also utilizes the framework of crossover of scaling laws~\cite{Maruoka_2023}.

The identification self-similarity is equivalent to finding 
scale-transformation {\it symmetry} inherent in 
the phenomenon of interest.
Symmetry is a fundamental concept for describing natural phenomena, and its implementation or detection is becoming increasingly significant at the interface of physics and machine learning~\cite{wang2023physicsguided,otto2023unified}.
Several strategies have been developed based on the fact that neural networks for classification tasks can encode symmetries in their hidden layer~\cite{10.21468/SciPostPhys.11.1.014, Krippendorf_2021}. 
Moreover, there has been an approach employing Generative Adversarial Networks (GANs) for the discovery of symmetries. These networks are trained to identify transformations that preserve a given data distribution~\cite{pmlr-v202-yang23n, PhysRevD.105.096031}.
In the present approach, we integrate scale-transformation symmetries into the architecture of neural networks. In these networks, the generators of the symmetries are embedded as specific parameters. 
When the network is trained successfully, it becomes possible to extract the generators of scale transformations directly from the optimized parameters.

The rest of the paper is structured as follows. 
In Sec.~\ref{sec:self-similarity}, we review the self-similar solutions in physical phenomena, and provide a mathematical characterization of data collapse.
In Sec.~\ref{sec:method}, we describe the method for finding self-similarity using neural networks. 
We illustrate our method with synthetic and experimental data in Sec.~\ref{sec:example-1} and Sec.~\ref{sec:example-two-combinations}.
In Sec.~\ref{sec:conclusion}, we give a summary and discuss the tips for effective estimation and limitations.

\section{ Self-similarity and data collapse }\label{sec:self-similarity}

In this section, we give a brief review on the self-similarity of physical phenomena. 
Self-similar solutions appear as a result of scale-transformation symmetry in the physical phenomenon of interest. 
We give a mathematical characterization of 
the concept of {\it data collapse}, which is the reduction of the solution space utilizing the scale-transformation symmetry.

\subsection{ Dimension function }\label{sub_sec:self-similarity}

Let $\{z_i\}_{i=1,\ldots,N_z}$ denote a set of {\it physical parameters}\footnote{
\changed{
Note that ``physical parameters'' in our context refer 
to spacetime coordinates as well as numerical parameters that characterize the system with physical units. For example, in the one-dimensional diffusion equation, $\partial_t u(t,x) = D \partial_x^2 u(t,x)$, with the initial condition $u(0,x) = Q\delta (x)$, the set of physical parameters is $\{t, x, D, u, Q\}$ and all the physical parameters are expressed as a product of a numerical number and a physical unit.}
} describing a certain physical phenomenon.
We denote the physical unit of parameter $z_i$ 
by ${\rm unit}_i$ (such as ``cm'' and ``kg$\cdot$m/${\rm s}^2$'').
A physical parameter is expressed as a product of 
a numerical number and a physical unit. 
Namely, $z_i$ can be always written in the following form:
\begin{equation}
z_i = N [z_i, {\rm unit}_i] \, {\rm unit}_i ,
\label{eq:phys-para}
\end{equation}
where $N [z_i, {\rm unit}_i]$ denotes the numerical number 
of physical parameter $z_i$ measured in 
${\rm unit}_i$. 
For example, $N[10\,{\rm m}, {\rm m}] = 10$.

Let $\{U_{i}\}_{i=1,\ldots,N_{\rm unit}}$ be a chosen set of units, e.g., $\bm U = ({\rm m}, {\rm kg}, {\rm second})^\top$.
Suppose that we would like to use rescaled units,
\begin{equation}
U_{i} \mapsto U_{i} / L_{i} 
\quad \text{for each $i$},
\label{eq:unit-scale}
\end{equation}
where $L_i \in \mathbb R_{>0}$. 
The unit of $z_i$ is transformed
under Eq.~\eqref{eq:unit-scale} as
\begin{equation}
{\rm unit}_i \mapsto 
{\rm unit}'_i
= 
\frac{1}{\phi_i (\bm L)}
{\rm unit}_i .
\end{equation}
The function $\phi_i$ is called the {\it dimension function}.
One can show that $\phi_i(\bm L)$ is always 
written 
as a power-law monomial,
\begin{equation}
 \phi_i(\bm L) = \prod_{j=1}^{N_{\rm unit}} (L_j)^{\gamma_{ij}} .
 \label{eq:phi-power}
\end{equation}
\changed{where $\gamma_{ij}$ are power exponents.}

\changed{This follows from the following argument. 
Let us consider the scale transformation with scale factor ${\bm L} = \left( L_1, \cdots , L_{N_{\rm unit}} \right)$ by 
$\phi_i (\bm L)$,  
${\rm unit}_i \mapsto \frac{1}{\phi_i (\bm L)}{\rm unit}_i$.
This should be equivalent to the transformation with another scale factor ${\bm L'} = \left( L'_1, \cdots , L'_{N_{\rm unit}} \right)$ by $\phi_i (\bm L')$ 
followed by the transformation by 
$\phi_i \left(\frac{L_1}{L'_1}, \cdots , \frac{L_{N_{\rm unit}}}{L'_{N_{\rm unit}}} \right)$.
Thus, the function $\phi_i$ should satisfy the following consistency condition,
\begin{equation}
\frac{\phi_i (\bm L)}{\phi_i (\bm L')}
= \phi_i \left(\frac{L_1}{L'_1}, \cdots , \frac{L_{N_{\rm unit}}}{L'_{N_{\rm unit}}} \right).
\end{equation}
Solutions of this functional equation are power-law monomials like Eq.~\eqref{eq:phi-power} \cite{barenblatt_1996, barenblatt_2003}.\footnote{\changed{This argument will be repeated for an element $\varphi_i$, scale transformation acting on dimensionless parameters. How to solve the functional equation \eqref{eq:phi-power} is written in Ch.1 of Ref.~\cite{barenblatt_2003}}}}

%
Since a physical parameter $z_i$ is invariant 
under the scaling of units
(e.g., $1\,{\rm m} = 100\,{\rm cm}$),
\begin{equation}
\begin{split}
z_i 
&= 
\phi_i (\bm L)  
N [z_i, {\rm unit}_i] \, 
\frac{1}{\phi_i (\bm L)} 
{\rm unit}_i 
\\
&=
\phi_i (\bm L) 
N [z_i, {\rm unit}_i] \, 
{\rm unit}'_i     
\\
&= 
N [z_i, {\rm unit}'_i] \, 
{\rm unit}'_i .
\end{split}
\end{equation}
Thus, under the change of units $U_i \mapsto U_i / L_i$, the numerical value of a physical parameter $z_i$ is changed as
\begin{equation}
N [z_i, \frac{1}{\phi_i (\bm L)} {\rm unit}_i]
= 
\phi_i (\bm L) N [z_i, {\rm unit}_i] .
\end{equation}

\subsection{ Physical relations and symmetry groups }

Let us denote the space of possible physical parameters 
by $M$, which we take to be an $N_z$-dimensional manifold.
Suppose that physical parameters $\bm z$
satisfy relations of the form
\begin{equation}
\bm F(\bm z) = \bm 0 ,
\label{eq:f-0}
\end{equation}
where $\bm F: \mathbb R^{N_z} \to \mathbb R^{N_F}$ 
is a smooth function, and $N_F$ is the number of constraints satisfying  $N_F \le N_z$.
We will refer to $\bm z$ satisfying Eq.~\eqref{eq:f-0} 
a {\it solution}, and consider the space of solutions,
\begin{equation}
M_{\rm sol} \coloneqq \{ \bm z \in M 
\, | \, \bm F(\bm z) = \bm 0 \}. 
\end{equation}
We take $\bm F$ to be of maximal rank,
meaning that 
${\rm rank}\, \frac{\p F_k}{\p z^i} = N_F$ for 
$\bm z$ satisfying Eq.~\eqref{eq:f-0}.
Then, $M_{\rm sol}$ is an $(N_z - N_F)$-dimensional submanifold of $M$.

We can separate physical parameters into two groups; $\{z^{(\I)}_i\}_{i=1,\dots,N_{\rm unit}}$ are a chosen set of parameters with independent dimensions, with which the dimensions of other parameters can be expressed, 
and $\{ z^{( \II )}_i \}_{i=1,\ldots,N_\pi}$
are the other physical parameters,
where we defined $N_\pi \coloneqq N_z - N_{\rm unit}$.
The choice of the set of parameters with independent dimensions is arbitrary as long as 
parameters $z^{( \I )}_i$ are of independent dimensions.
One can introduce dimensionless parameters 
by rescaling 
$z^{( \II )}_i$
by a product of powers of 
$z^{( \I )}_i$,
\begin{equation}
\pi_i \coloneqq z_i^{(\II )} 
\prod_j (z^{(\I )}_j)^{\alpha_{ij}} .
\end{equation}
The function can be written in the form
\begin{equation}
F_k(\bm z) 
=
F_k(\bm z^{(\I )}, \bm z^{(\II )})
=
\prod_i (z^{(\I )}_i)^{\alpha_{ki}}
\mathcal F_k (\bm z^{(\I )}, \bm \pi),
\end{equation}
where $\mathcal F_k$ is dimensionless.

Now consider the rescaling of dimensionful parameters $\bm z^{(\I )} \mapsto \bm z'^{(\I )}$.
Since $\bm \pi$ and $\mathcal F$ are dimensionless, we have 
\begin{equation}
\bm{\mathcal F} (\bm z^{(\I )}, \bm \pi)
= 
\bm{\mathcal F} (\bm z'^{(\I )}, \bm \pi).
\end{equation}
This means that $\bm{\mathcal F} (\bm z^{(\I )}, \bm \pi)$ 
is in fact independent of $\bm z^{(\I )}$. 
Let us denote this function by $\bm \Phi$,
\begin{equation}
\bm \Phi (\bm \pi) 
\coloneqq 
\bm{\mathcal F} (\bm z^{(\I )},
\bm \pi).
\end{equation}
The function $\bm F$ is expressed as 
\begin{equation}
F_k(\bm z)    
= \prod_i (z^{(\I )}_i)^{\alpha_{ki}} \Phi_k (\bm \pi) .
\label{eq:f-phi}
\end{equation}

To summarize the discussion so far,
we have shown that a physical relation $\bm F(\bm z) = \bm 0$ 
can be stated in an equivalent form 
using $\bm \Phi(\bm \pi)$, 
\begin{equation}
\bm F(\bm z) = \bm 0 
\iff \bm \Phi (\bm \pi) = \bm 0 .
\end{equation}
This is nothing but the content of the Buckingham $\Pi$ theorem.

We consider transformations
of a physical state $\bm z \in M$.
Transformations  constitute a group action on $M$.
Namely, a transformation is a map $\cdot : G \times M \to M$, where $G$ is a group, under which a physical state is transformed as
\begin{equation}
M \ni \bm z \mapsto \bm z' = g \cdot \bm z \in M .
\end{equation}
We say that a transformation $g$
preserves the relation $\bm F (\bm z)=\bm 0$ if
the following holds:
\begin{equation}
\bm F (\bm z) = \bm 0 
\implies \bm F (g \cdot \bm z) = \bm 0.
\end{equation}
A set $G$ of transformations preserving $\bm F (\bm z)=\bm 0$
is called a {\it symmetry group}~\cite{MR836734}.
Namely, a symmetry group of a physical system 
is defined as a group of transformations $G$
acting on $M$ with the property that
$g \in G$ transforms solutions of the system to other solutions. 
In the following, we restrict our 
attention to connected Lie groups of symmetries,
and do not consider discrete symmetry groups.

\subsection{ Scale transformations and data collapse }

The basic idea behind data collapse is to obtain a reduced state space by exploiting scale-transformation symmetries of a system. 
We here mathematically formalize this idea.

Under a {\it scale transformation}, 
a state $\bm z \in M$ is transformed as
\begin{equation}
z_i \mapsto 
(\sigma \cdot \bm z)_i 
= \alpha_i z_i, \,\,\, \alpha_i \in \mathbb R_{>0}.
\end{equation}
Let $G^{\rm (s)} \subset G$ be the group of all scale transformations preserving $\bm F(\bm z) =\bm 0$.

As a particular kind of scale transformations, 
let us introduce the {\it scale transformation of units}\footnote{While it is called {\it scale transformation of units}, it does not mean that it actually changes the scale of units (e.g. 1~m $\to$ 100~cm). It is the scale transformations acting on {\it the physical parameters} (e.g. 1~m $\to$ 100~m) but their scale functions, which determine the factor by which the parameters change following their similarity, are equal to the dimension functions. In general, dimension functions are considered to act on the numerical values of physical quantities (see p.16 of Ref.~\cite{barenblatt_2003}). Here we intend to discuss the self-similarity of the physical parameters, the relation between the self-similarity of units and the self-similarity of physical parameters.}, under which physical parameters change
by a factor given by the dimension function,
\begin{equation}
z_i \mapsto z'_i = \phi_i (\bm L) z_i .
\label{eq:scale-transformation-of-units}
\end{equation}
Note here that numerical values are not transformed
and hence the value of $z_i$ is changed under this operation.
The group of scale transformations of units
will be denoted as 
$G_{\rm unit}^{\rm (s)}$,
\begin{equation}
G_{\rm unit}^{\rm (s)} 
\coloneqq
\{ \phi : z_i \mapsto \phi_i (\bm L) z_i \}.
\end{equation}
One can show the following:
\begin{prop}
A scale transformation of units is always 
a symmetry, i.e., $G^{\rm (s)}_{\rm unit} \subset G.$
\end{prop}
\begin{proof}
Under a scale transformation $\phi$ of units,
parameters $\bm \pi$ are invariant,
$\phi \cdot \bm \pi = \bm \pi$, because they are dimensionless.
Using the relation~\eqref{eq:f-phi},
\begin{equation}
\begin{split}
F_k (\phi \cdot \bm z) 
&= 
\prod_i 
( (L_i)^{\beta_{ki}} z^{(\I )}_i)^{\alpha_{ki}} 
\Phi_k (\phi \cdot \bm \pi) 
\\
&= \prod_i 
( (L_i)^{\beta_{ki}} z^{(\I )}_i)^{\alpha_{ki}} 
\Phi_k (\bm \pi) 
\\
&\propto F_k (\bm z)
\\
&= 0 . 
\end{split}
\end{equation}
Thus, the transformation $\phi$ preserves relation $\bm F$,
$\bm F(\phi \cdot \bm z) =\bm 0$, and $\phi \in G$.
\end{proof}
For example, let us consider the diffusion equation
\begin{equation}
\partial_t u(t,x) = D\partial_x^2u(t,x)\,, 
\end{equation}
with the following initial and boundary conditions 
\begin{equation}
u(0,x)=Q\delta(x), \quad 
u(t,\pm \infty)=0, \nonumber
\end{equation}
where $Q$ is the total mass, $u$ is the concentration of substance, $x$ is the location, $t$ is the time, $D$ is the diffusion coefficient, and $\delta(x)$ is the Dirac delta function.
The set of physical parameters for this system is $\bm z = (u, t, x, D, Q)^\top$.
We can construct two independent dimensionless parameters, $\pi_1 = \frac{u\sqrt{Dt}}{Q}$ and $\pi_2 = \frac{x}{\sqrt{Dt}}$.
The self-similar solution of diffusion equation is in the form of $\pi_1 = f\left(\pi_2 \right)$. The dimension function $\phi$, scale transformation of units, acting on $\bm z$ is given as
\begin{equation}    
\phi= \left(\frac{M}{L}, T, L, \frac{L^2}{T}, M \right)^{\top}.
\end{equation}
It is easy to see $\pi_1$ and $\pi_2$ are invariant under $\phi$. Therefore, the transformation $\phi$ preserves the relation $\pi_1 = f\left(\pi_2 \right)$. $\phi$ can be estimated by seeing the numerical values of $\bm z$ by rescaling the units of the length, mass and time as
\begin{equation}
    (U_{\rm length},U_{\rm mass},U_{\rm time})
    \mapsto
    \left(\frac{U_{\rm length}}{L}, \frac{U_{\rm mass}}{M}, \frac{U_{\rm time}}{T}\right)\,.
\end{equation}

The group 
$G_{\rm unit}^{\rm (s)}$
characterizes the self-similar solutions 
of the first kind. 
We can introduce a quotient space,
\begin{equation}
M_\pi \coloneqq M_{\rm sol} / G^{\rm (s)}_{\rm unit} ,
\end{equation}
which is the space of $G^{\rm (s)}_{\rm unit}$ orbits.
The dimensionless parameters $\bm \pi$ are the 
global coordinate of the quotient space, 
$\bm \pi \in M_\pi$. 

When the set of scale transformations preserving $\bm F(\bm z) = \bm 0$
is identical to the scale transformations of units,
i.e., $G^{\rm (s)} = G^{\rm (s)}_{\rm unit}$, 
the system only has {\it self-similar solutions 
of the first kind}\footnote{\changed{Note that similarity of the first kind corresponds to the case in which physical parameters and units posses the same structure of self-similarity as $G^{\rm (s)}_{\rm unit}$ is equivalent to the dimension functions that preserve the similarity under the scaling units. This is the reason why, in some case, the self-similarity can be exploited by simply applying dimensional analysis even though dimensional analysis should be nothing but the procedure to obtain the invariants of the scale transformation of the units, not the physical parameters. These are the cases when the problems belong to the similarity of the first kind.}}. 

In general, there can be scale transformations 
beyond $G^{\rm (s)}_{\rm unit}$. 
Such scale transformations lead to
{\it self-similar solutions of the second kind}\footnote{Similarity of the first kind and the second kind are a category of self-similarity. Barenblatt formulated this category by the convergence of their dimensionless functions $\Phi$. When a dimensionless function obtained by dimensional analysis $\Phi(\eta, \xi)$ converges to a finite limit as $\eta \to 0$ or $\infty$, it corresponds to {\it similarity of the first kind} or {\it complete similarity}. On the other hand, when $\Phi(\eta, \xi)$ does not satisfy the complete similarity but the convergence is recovered by introducing similarity parameters composing dimensionless parameter as $\xi / \eta^{\epsilon}$, it corresponds to {\it similarity of the second kind} or {\it incomplete similarity}. This formulation is identical with our formulation. 
}.
Possible self-similar solutions of the second kind 
are characterized by the following quotient group,
\begin{equation}
G^{\rm (s)}_{\rm second}
\coloneqq G^{\rm (s)} / G^{\rm (s)}_{\rm unit} .
\end{equation}
The group $G^{\rm (s)}_{\rm second}$ cannot be determined
by the dimensional analysis, and hence it is a consequence of a nontrivial physical law.
An element $\varphi \in G^{\rm (s)}_{\rm second}$ 
induces a nontrivial transformation on $\bm \pi$, 
\begin{equation}
\pi_i \mapsto (\varphi \cdot \pi)_i
= 
\varphi_i (\bm A)  \pi_i .
\end{equation}
where $\bm A$ are scale factors.
By repeating the same argument with the dimension function (see Sec.~\ref{sub_sec:self-similarity}),
one can show that $\varphi_i$ is always written as a power-law monomial,
\begin{equation}
\varphi_i (\bm A)
= 
\prod_{\alpha} (A_\alpha)^{\gamma_{i}^{(\alpha)}} .
\label{eq:varphi-action}
\end{equation}

Let us introduce the quotient space 
where all the scale transformations are modded out,
\begin{equation}
\bar M = M_{\rm sol} / G^{\rm (s)} . 
\end{equation}
On $\bar M$, all scale transformations act trivially.
We can introduce a coordinate on $\bar M$, $\bm Z \in \bar M$,
which is defined so that it is invariant under
all the scale transformation $\sigma \in G^{\rm (s)}$ 
preserving $\bm F (\bm z) = \bm 0$,
\begin{equation}
\bm Z (\sigma \cdot \bm z) 
= \bm Z (\bm z).
\end{equation}
Since the action of $G^{\rm (s)}_{\rm second}$ on $\bm \pi$ is given 
by Eq.~\eqref{eq:varphi-action}, 
$Z_i$ is given by products some powers of dimensionless parameters. 
The construction of these parameters 
can be done in a completely similar manner to the introduction of $\bm \pi$.
Namely, we can separate dimensionless parameters $\bm \pi$ 
into two groups, 
\begin{equation}
\bm \pi = 
\begin{pmatrix}
{\bm \pi}^{(\I )} \\
{\bm \pi}^{(\II )}
\end{pmatrix}, 
\end{equation}
where the number of the components 
of the first group ${\bm \pi}^{(\I )}$ 
is given by the dimension of $G^{\rm (s)}_{\rm second}$, 
\begin{equation}
N \coloneqq {\rm dim\,} G^{\rm (s)}_{\rm second},  
\end{equation}
\changed{and the component of ${\bm \pi}^{(\I )}$ must be functionally independent one another, which means that the component of ${\bm \pi}^{(\I )}$ cannot be rescaled by the power-law monomials composed of the other components of ${\bm \pi}^{(\I )}$ to obtain invariant parameters under $G^{\rm (s)}$\footnote{\changed{This argument is exactly the repartition of $z^{(\I )}$, the parameters with independent dimensions in Subsec. B. In a more general discussion, it corresponds to the functionally independent parameters on the scale transformation, with which the parameters cannot be rescaled by other parameters to obtain invariant parameters. The number of parameters that can be reduced is equivalent to the dimension of the orbit of the scale transformation acting on the parameters. See pp. 86-89 of Ref.~\cite{MR836734}.}}.}
One can introduce invariant parameters under $G^{\rm (s)}$ 
by rescaling $\pi^{(\II )}_i$ by a product of powers of 
$\pi^{(\I )}_i$,
\begin{equation}
Z_i 
= 
\pi^{(\II )}_i \prod_j (\pi^{(\I )}_j)^{r_{\mu j}} .
\end{equation}
The function $\Phi (\bm \pi)$
can be written in the form, 
\begin{equation}
\Phi_k (\bm \pi) 
= \prod_i (\pi^{(\I )}_i)^{r_{ki}} \bar\Phi_k (\bm Z).
\end{equation}
On the quotient space $\bar M$, the following $\bar G$ acts:
\begin{equation}
\bar G \coloneqq G /  G^{\rm (s)} . 
\end{equation}
The quotient $\bar G$ forms a group since $G^{\rm (s)}$
is a normal subgroup of $G$ as we show below:
\begin{prop}
The group of scale transformations $G^{\rm (s)}$ 
is a normal subgroup of $G$. 
\end{prop}

\begin{proof}
Since $\sigma \in G^{\rm (s)}$ acts trivially on $\bm Z$, we have
\begin{equation}
(g \sigma g^{-1}) \cdot \bm Z = \bm Z. 
\end{equation}
This implies that the action of $g\sigma g^{-1}$ 
on $\bm z \in M_{\rm sol}$ should be of the form
\begin{equation}
(g \sigma g^{-1}) \cdot \bm z = \sigma' \cdot \bm z ,
\label{eq:g-s-gi-sp}
\end{equation}
where $\sigma' \in G^{\rm (s)}$. 
Equation \eqref{eq:g-s-gi-sp} means that 
$G^{\rm (s)}$ is a normal subgroup of $G$.
\end{proof}

Let us summarize the discussion in this section
in the following diagram,
\begin{equation}
G 
\curvearrowright 
M_{\rm sol} 
\,\,\,\,
\longrightarrow
\,\,\,\,
G_\pi
\curvearrowright 
M_\pi 
\,\,\,\,
\longrightarrow
\,\,\,\,
\bar G
\curvearrowright 
\bar M
\end{equation}
where $\curvearrowright$ indicates a group action, 
and we have defined 
\begin{equation}
G_\pi \coloneqq G / G^{\rm (s)}_{\rm unit}.
\end{equation}
When we refer to achieving a {\it data collapse,} 
we are describing the process of
obtaining $\bar G \curvearrowright \bar M$ 
from $G \curvearrowright M_{\rm sol}$.
This means to fully exploit 
scaling symmetries existing in the phenomenon 
of interest to obtain a reduced state space.
The remaining symmetry group $\bar G = G / G^{\rm (s)}$
no longer contains scale transformations. 

In practice, this means that, 
when we plot the data points using newly introduced parameters 
$\bm Z$, all the data points converge to
a surface determined by $\bar {\bm \Phi} = \bm 0$, which we call a {\it scaling function}.

Note that there is a hierarchical structure (see Fig.~\ref{fig:F1b}). \changed{{\it Physical parameters} $\bm z$ are composed of a product of a numerical number and a physical units, which is not invariant under scale transformation of units}. 
While $\bm \pi$ that are composed of $\changed{\bm z}$
are invariant under the scale transformations of units $G^{\rm (s)}_{\rm unit}$,
$\bm Z$ are composed of $\bm \pi$ and 
are invariant under all the scale transformations.
We will refer to $\bm \pi$ and $\bm Z$
as the {\it similarity parameters of the first \changed{class}} and {\it similarity parameters of the second \changed{class}}, respectively \cite{Maruoka_2023}.

\begin{figure}
\includegraphics[width=8.5cm]{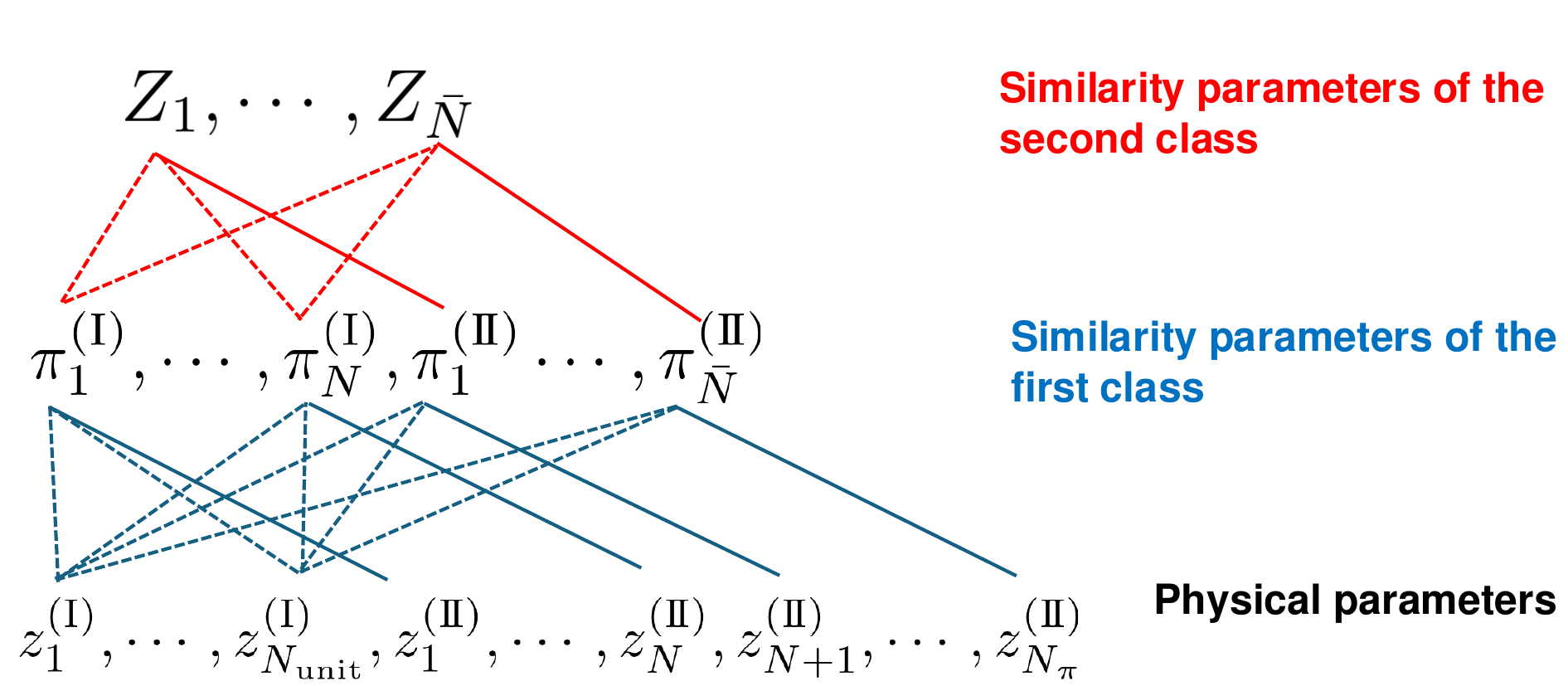}
\captionsetup{justification=raggedright,singlelinecheck=false}
\caption{Hierarchical structure of similarity parameters. \changed{Physical parameters $\bm z =  \left( z_1, \cdots, z_{N} \right)$ are parameters that are composed of a product of a numerical number and a physical unit. 
They are not invariant under the scale transformation of units.} $\bm \pi = \left( \pi_1, \cdots,\pi_{N_\pi} \right)$  
are referred to as similarity parameters of the first \changed{class}, that are composed of physical parameters and are invariant under the scale transformations of units. $\bm Z = \left(Z_1, \cdots, Z_{\bar N} \right)$ are the similarity parameters of the second \changed{class}, that are composed of similarity parameters of the first class and are invariant under the all the scale transformations where $\bar N \coloneqq N_\pi - N$. \changed{The dashed lines indicate the parameters with independent dimensions of their group (group $\I$) and the solid lines indicate the remaining parameters (group $\II$)}. 
} 
\label{fig:F1b}
\end{figure}

\subsection{ Infinitesimal group action }

Let us discuss infinitesimal actions of scale transformations. 
We here discuss 
the action of $G_\pi$ on $M_\pi$, 
$G_\pi \curvearrowright M_\pi$,
because we are interested in the similarity of the second kind.
There is a subgroup
$G^{\rm (s)}_{\rm second} \subset G_\pi$ of 
scale transformations,
that characterizes self-similar solutions of the second kind.
A generator of $G^{\rm (s)}_{\rm second}$ can be written as
\begin{equation}
\hat{\bm p}_\alpha
\coloneqq \sum_i p_{\alpha i} \pi_i \frac{\p}{\p \pi_i} ,
\label{eq:gen-scale}
\end{equation}
where $p_{\alpha i} \in \mathbb R$. 
Finite scale transformations can be generated by Eq.~\eqref{eq:gen-scale} as 
\begin{equation}
 \pi_i
 \mapsto 
 e^{\sum_\alpha \epsilon^\alpha \hat{\bm p}_\alpha } \pi_i 
 = 
 e^{\sum_{\alpha} \epsilon^\alpha p_{\alpha i} } \pi_i .
\end{equation}

We shall choose a set of vectors
$\{ \bm r_\mu \}_{\mu =1, \ldots, \bar N}$ (where $\bar N \coloneqq N_\pi - N$)
so that it is a basis of the orthogonal complement of 
${\rm span}
\{ 
\bm p_\alpha
\}_{\alpha=1,\ldots,N}
$, i.e., 
\begin{equation}
\left(
{\rm span}
\{ 
\bm p_\alpha
\}_{\alpha=1,\ldots, N}
\right)^\perp
= 
{\rm span}
\{ 
\bm r_\mu
\}_{\mu =1, \ldots, \bar N} .
\end{equation}
Specifying the generators 
is equivalent to specifying $\{\bm r_\mu \}_{\mu =1, \ldots, \bar N}$.

\section{ Finding self-similarity with neural networks }\label{sec:method}

\subsection{ Outline of the strategy }\label{sec:strategy}

As discussed in previous section, the identification of power exponents of similarity parameters is sufficient to exploit scale-invariance of the problems.
By the dimensional analysis, one can find the self-similar solutions of the first kind. 
When the problems belong to similarity of the second kind, the corresponding power-law exponents
reflect a nontrivial physical law.
Our primary goal is to obtain the power exponents of similarity parameters of the second \changed{class}
in a data-driven way.
Let us outline the strategy of our method in the following.

Firstly, we need to collect data relevant to the problem under study in a possibly wide range of parameters.
The data can come from either numerical simulations or actual experiments. Then we assume that we have sufficient amount of data plots of governing parameters, which are a minimum set of physical parameters to describe the problem.

Next, we perform dimensional analysis 
and introduce dimensionless parameters $\bm \pi$.
Then, we check if the system allows for self-similar solutions of the second kind, which means to identify the similarity parameters of the second \changed{class}.
Traditionally, this step involves making assumptions based on a physical law that is suitable for the specific problem, which can lead to biases.

To determine the power exponents of similarity parameters of the second \changed{class}, the framework of  crossover of scaling laws summarized in Ref.~\cite{Maruoka_2023} is useful. 
In many physical problems, there are certain region of parameters in which the dependencies on some parameters are weak and can be neglected, which results in a simple scaling law.
In such an ideal region there is a constant dimensionless number composed of the scaling law,
\begin{equation}
\frac{\pi}
{\pi_1^{q_1}\cdots\pi^{q_N}_{N}}
= {\rm const.}.
\label{eq:e5}
\end{equation} 
In such a case, the power exponents of similarity parameter can be easily identified by log-log plot. Or the scaling law can occasionally be obtained analytically as an idealized problem. Such a condition is generally {\it asymptotically} realized in a certain range of physical parameters, in which their scaling function $\Phi$ converges to a finite limit as similarity parameters goes zero or infinity, $\Phi \rightarrow {\rm const}$ as $\pi_k \ll 1$. 
The scaling law asymptotically valid in a certain scale range is called {\it intermediate asymptotics}\footnote{Recall that the ideal gas equation of state is intermediate asymptotics when the volume of molecule $b$ and molecular interaction $a$ are ignored in the van der Waals equation, $p = \frac{nRT}{V-nb}-\frac{an^2}{V^2} \rightarrow \frac{nRT}{V}~\left(\frac{an^2}{V^2} \ll p \ll \frac{RT}{b} \right)$.}.

Nonetheless, this idealized scaling law is disrupted when a similarity parameter exceeds its ideal region and begins to interfere. This disruption of the idealized scaling law might merely break the self-similarity, or it could give rise to a different scaling law.
Such a situation is frequently observed in experiments and simulations that deal with 
a broad range of parameters.
This phenomenon is referred to as the crossover of scaling laws.

Based on this framework, if we start from a scaling law in idealized range, 
where $\pi/\pi_1^{q_1}\cdots\pi_N^{q_N}={\rm const.}$ holds, and introducing a similarity parameter using the scaling law, we can identify at least one similarity parameter in the whole problem as
\begin{equation}
\Psi = \frac{\pi}{\pi_1^{q_1}\cdots\pi_N^{q_N}},
\label{eq:e6}
\end{equation}
which we call a {\it fixed similarity parameter}. 
Then the problem now is to identify other similarity parameters that breaks the asymptotics.
The breaking of the asymptotics is driven by a parameter which did not contribute the problem in the idealized situation.
Let us call such a parameter which starts to interfere as 
a {\it driving parameter}, $\pi_{D}$.
The interfering similarity parameter must include this parameter and it should be of the form 
\begin{equation}
Z = \frac{\pi_{D}}{\pi_1^{p_{1}}\cdots\pi_N^{p_{N}}},
\label{eq:e7}
\end{equation}
which we call an {\it interfering similarity parameter}.
For notational simplicity, we here describe the case when there is only one interfering similarity parameter. In general, there can be multiple ones, and the present method also applies to these situations. 
In Sec.~\ref{sec:example-two-combinations}, we analyze such examples.

The exponents $p_1,\cdots,p_N$ characterize the scaling law
that leads to self-similar solutions of the second kind.
Using neural networks, we estimate these exponents from the provided data.
The construction of these neural networks will be detailed in Sec.~\ref{sec:method-nn}.

The above procedure can be summarized as follows:
\begin{itemize}
    \item[1)] Collect the sufficient data of simulation or experimental data to describe the phenomenon of interest, 
    and identify governing parameters. 
    \item[2)] Apply dimensional analysis to governing parameters to obtain dimensionless parameters.
    \item[3)]
    Assuming that there can be nontrivial scale transformations beyond the scale transformations of units,
    find the asymptotic scaling law in an idealized region, $\pi = {\rm const.} \pi_1^{q_1}\cdots\pi_N^{q_N}$.
    \item[4)] Introduce a dimensionless number as fixed similarity parameter using scaling law in 3), 
    $\Psi = \frac{\pi}{\pi_1^{q_1} \cdots \pi_N^{q_N}}$.
    \item[5)] Identify the driving parameter $\pi_D$. 
    Then defining an interfering similarity parameter including $\pi_D$ as $ Z = \frac{\pi_{D}}{\pi_1^{p_1}\cdots\pi_N^{p_N}}$, assume the following self-similar solutions, $\Psi = \Phi\left[Z\left(p_1,\cdots, p_k \right)\right]$ with unknown power exponents $p_1,\cdots, p_k$.
    \item[6)] Determine the power exponents through the optimization of a neural network (see Sec.~\ref{sec:method-nn} for detail). 
\end{itemize}
\begin{figure}
\includegraphics[width=8.5cm]{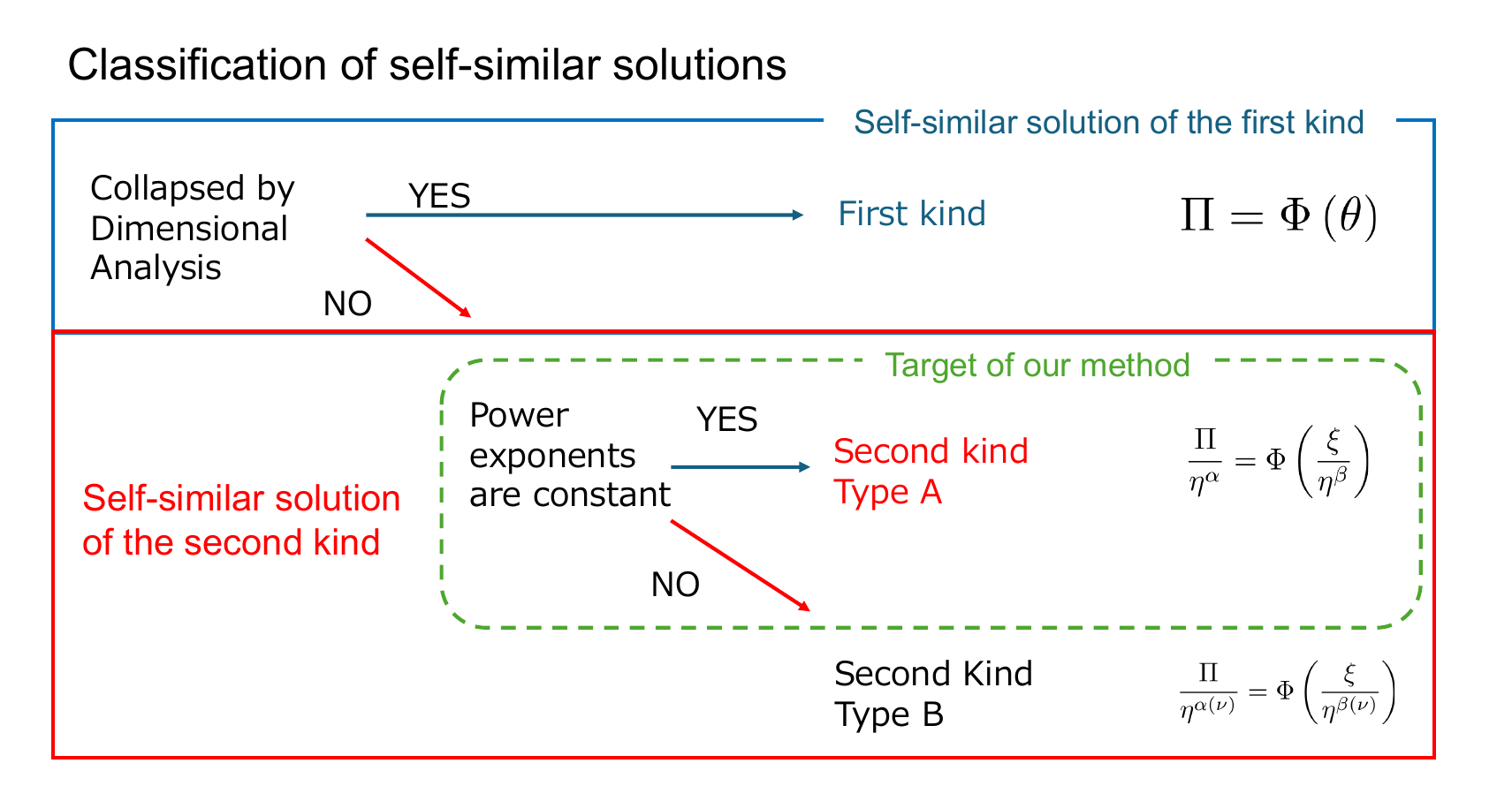}
\caption{Classification of self-similar solutions. The first branch asks whether data collapse is achieved by dimensional analysis. If this is the case, the problem is self-similar solution of the first kind. If not, the problem possesses self-similar solutions of the second kind. The latter case can be further divided into two types. Type A corresponds to the cases when power-law exponents of similarity parameters are constant, while in the case of Type B power-law exponents are functions of dimensionless parameters. The present neural network approach targets the determination of power-law exponents for problems under Type A.}
\label{fig:F1}
\end{figure}
If we succeed in finding the exponents leading to good data collapse, this means that the problem has similarity of the second kind.

In Fig.~\ref{fig:F1}, we summarize the classification of possible situations regarding data collapse.
First branch asks if dimensional analysis leads to data collapse. If this is the case, the problem belongs to similarity of the first kind.
When data collapse is not achieved by dimensional analysis, the physical problem has self-similarity of the second kind. These cases can be further classified into two subcategories depending on whether power-law exponents of similarity parameters are constant (Type A) or nontrivial functions of dimensionless parameters (Type B).
The present approach aims to find self-similar solutions of the second kind, Type A.
For physical problems under Type B, the approach remains effective provided there exists a parameter region where the power-law exponents can be regarded as constants.

\subsection{Neural network approach}\label{sec:method-nn}

Let us now give the details of the construction and optimization of neural networks for finding self-similarity.
For the simplicity of notations, we here describe the case when 
there is only one interfering similarity parameter\footnote{The construction can be readily extended to more general cases. We discuss examples of two interfering similarity parameters in
in Sec.~\ref{sec:example-two-combinations}.
}. 
Suppose that we have a system with $N_\pi$ dimensionless parameters
$\{\pi_i \}_{i=1,\ldots,N_\pi}$
and there is a physically significant relationship of the form
\begin{equation}
\prod_i (\pi_i)^{q_i}
= 
\Phi \left( \prod_i (\pi_i)^{p_i} \right),
\label{eq:F=G}
\end{equation}
where the vectors $\bm p$ and $\bm q$ are linearly independent.
The system has self-similar solutions associated with this relation.
Namely, we can rescale parameters,
\begin{equation}
\pi_i \mapsto e^{c_i} \pi_i,
\label{eq:scaling}
\end{equation}
in such a way that 
the both of the combinations 
$\prod_i (\pi_i)^{p_i}$ and $\prod_i (\pi_i)^{q_i}$ are invariant, 
and the rescaled parameters remain to be a solution.
To discuss scale transformations,
it is convenient to use the log of the original parameters, 
\begin{equation}
x_i \coloneqq \ln \pi_i 
\quad
\text{for each $i$.}
\end{equation}
Then, Eq.~\eqref{eq:F=G} can be rewritten as
\begin{equation} 
\bm q \cdot \bm x = \phi \left(\bm p \cdot \bm x \right), 
\label{eq:f-ax=g-bx}
\end{equation}
where 
$\phi(y)\coloneqq \ln \Phi(e^y)$. 
A scale transformation~\eqref{eq:scaling}
can be now represented as a shift of vector $\bm x$,
\begin{equation}
 \bm x \mapsto \bm x + \bm c.
 \label{eq:x-shift}
\end{equation}
Finding the scale invariance amounts to
finding a set of vectors $\{\bm c\}$
such that the combinations $\bm p \cdot \bm x$
and $\bm q \cdot \bm x$ are left invariant.
Namely, $\bm c$ has to be orthogonal to both $\bm p$ and $\bm q$,
$\bm c \cdot \bm p = \bm c \cdot \bm q = 0$.
We will denote the set of independent vectors with this property
by $\{ \bm c^{(\alpha)}\}_{\alpha=1,\ldots,N_\pi -2}$.
The scale invariance of a system is characterized 
by the following vector space $\mathcal C$,
\begin{equation}
\begin{split}    
\mathcal C
&\coloneqq 
{\rm span}
\{ \bm c^{(\alpha)} \}_{\alpha = 1, \ldots, N_\pi -2} 
\\
&= 
({\rm span}\{\bm p\})^\perp
\cap 
({\rm span}\{\bm q\})^\perp
\\
&= 
({\rm span}\{\bm p, \bm q \})^\perp ,
\end{split}
\end{equation}
where $V^\perp$ denotes the orthogonal complement of a vector space $V$.
Any element $\bm c \in \mathcal C$
gives rise to a scale transformation 
preserving the physical relation 
given by Eq.~\eqref{eq:F=G} (or equivalently Eq.~\eqref{eq:f-ax=g-bx}).

We note that the choice of $\{\bm p, \bm q\}$ 
that leads to a single $\mathcal C$ is not unique.
Obviously, the vectors $\bm p$ and $\bm q$ can be multiplied by any nonzero real numbers, $c,c' \in \mathbb R$, as
\begin{equation}
\bm p \mapsto c \bm p, \quad 
\bm q \mapsto c' \bm q. 
\end{equation}
Furthermore, we will have 
$
\mathcal C(\{\bm p, \bm q\})  
= 
\mathcal C(\{\bm p, \bm q' \})  
$ 
with 
\begin{equation}
\bm q' = \bm q + c \bm p ,
\end{equation}
where $c \in \mathbb R$. 
In general, any choice $\{\bm p', \bm q'\}$ will give rise to the same $\mathcal C = ({\rm span}\{\bm p, \bm q \})^\perp$  as long as 
${\rm span}\{\bm p', \bm q'\} = {\rm span}\{\bm p, \bm q\}$. 

\begin{figure*}[tb]
\centering
\includegraphics[width=14.5cm]{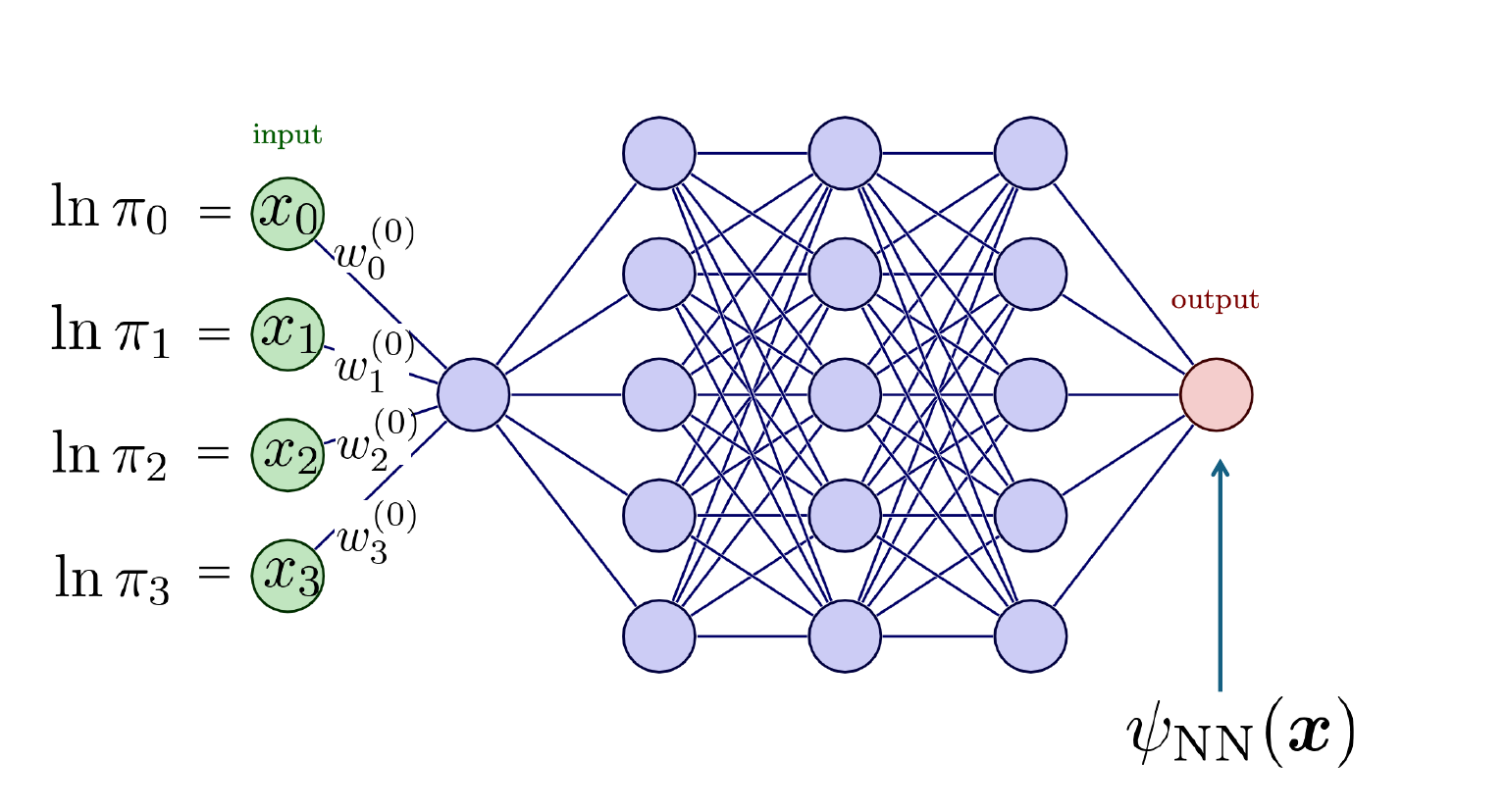}
\caption{Structure of the neural network for the case where the scaling function has single argument. In the case where the scaling function has two arguments, we should modify the structure so that the second layer has two nodes.}
\label{fig:nn-model}
\end{figure*}

The question we address here is whether
we can find the relation of the form~\eqref{eq:F=G} or equivalently~\eqref{eq:f-ax=g-bx} for given experimental or simulated data.
Let us consider the situation where $\bm q$ is already known, i.e., $\bm q$ represents a fixed similarity parameter.
The problem is to infer the vector $\bm p$ (up to the ambiguity) based on the data.
Denoting the fixed similarity parameter as
\begin{equation}
\Psi \coloneqq  \prod_i (\pi_i)^{q_i} ,
\end{equation}
the physically significant relation we would like to find is
\begin{equation}
\Psi = 
\Phi \left( \prod_i (\pi_i)^{p_i} \right) . 
\label{eq:physically_significant_relation}
\end{equation}
By taking the log of the parameters, we have 
\begin{equation}
\psi = \phi \left( \sum_i p_i x_i \right). 
\label{eq:z-f-ax}
\end{equation}
We parametrize the function \eqref{eq:z-f-ax} 
with a neural network (see Fig.~\ref{fig:nn-model}),
\begin{equation}
\psi_{\rm NN}(\bm x) 
\coloneqq 
\phi_{\rm NN} \left( \sum_i w^{(0)}_i x_i \right).
\label{eq:z-nn-def}
\end{equation}
Here, 
to enforce the functional dependence of Eq.~\eqref{eq:z-f-ax}, we chose the second layer to be of a single node, as shown in Fig.~\ref{fig:nn-model}.
Namely, the function $\psi_{\rm NN} (\bm x)$ 
is written by a function $\phi_{\rm NN}$ which 
has a single argument that takes a linear combination of $x_i$, 
and $w^{(0)}_i$ are parameters of the first linear layer.
We perform a supervised learning with training data $\{\psi_n, \bm x_n\}_{n=1,\ldots,N_{\rm data}}$, which can be experimental or synthetic data, by minimizing the following loss function, 
\begin{equation}
\ell_{\rm MSE} = \frac{1}{N_{\rm data}} 
\sum_{n=1}^{N_{\rm data} }
\left(\psi_n - \psi_{\rm NN}(\bm x_n) \right)^2. 
\label{eq:MSE_loss}
\end{equation}
To ensure that the network learns $\bm p$ rather 
than the already-known combination $\bm q$, we add the following regularization term,
\begin{equation}
\ell_{\rm reg} 
= 
\lambda_1 (|{\bm w}^{(0)}|^2-1)^2
+ 
\lambda_2 |\bm {\bm w}^{(0)} \cdot \bm q|^2.
\label{eq:regularization}
\end{equation}
The first term prevents $w^{(0)}_i$ from becoming too large.
The second term forces ${\bm w}^{(0)}$ to be orthogonal to ${\bm q}$ in order to avoid the trivial training, ${\bm w}^{(0)}=\bm q$ and $\psi_{\rm NN}$ being an identity.
If the neural network is successfully trained, we can expect that the first layer contains the information of $\bm p$.
Specifically, we will have 
\begin{equation}
\bm p \propto \bm w^{(0)} + c \bm q.
\label{eq:ambiguity}
\end{equation}
Here, the shift with $\bm q$ and the proportionality reflect the ambiguity of the neural network function $\psi_{\rm NN}$.
To see this, let us return to Eq.~\eqref{eq:physically_significant_relation}.
Without loss of generality, we can always replace this equation by the following form with a new function $F$:
\begin{equation}
    \Psi = F\left( \Psi^c \prod_i (\pi_i)^{p_i} \right)\,.
\label{eq:G}
\end{equation}
Specifically, $F$ can be constructed as follows.
By solving Eq.~\eqref{eq:physically_significant_relation} and multiplying by $\Psi^c$, one obtains
\begin{equation}
    \Psi^cF^{-1}(\Psi)=\Psi^c \prod_i (\pi_i)^{p_i}\,.
\end{equation}
The LHS is itself a function of $\Psi$. Denoting the inverse of this function as $F$ and solving the equation again, one can obtain the above expression.
This leads to the shift ambiguity of Eq.~\eqref{eq:ambiguity}.
The ambiguity of the proportionality in Eq.~\eqref{eq:ambiguity} is due to the fact that Eq.~\eqref{eq:G} can always be rewritten as
\begin{equation}
    \Psi = H\left( \left(\Psi^c \prod_i (\pi_i)^{p_i}\right)^\varepsilon \right)\,,
\end{equation}
where $H$ is a new function and $\varepsilon$ is any nonzero number.
These ambiguities need to be taken into account when we interpret the result of the training.

In general, the data can be noisy, which gives rise to uncertainties in parameter inference.
The uncertainties of the estimated parameters can be computed by the bootstrapping method, 
which is a statistical technique used for uncertainty quantification by resampling data with replacement. 
It generates multiple samples from a single dataset to estimate the distribution of a statistic 
such as mean or variance.
By repeatedly sampling, the bootstrap method creates a number of ``bootstrap samples,'' each of which is analyzed to produce an estimate. 
These estimates are then used to construct an empirical distribution of the statistic, allowing for the estimation of its variability.

In the following sections, we will illustrate the effectiveness of the procedure through applications to both synthetic and experimental data, as detailed in Sec.~\ref{sec:example-1}.
Moreover, our method can be readily extended to more general situations where the scaling function has multiple arguments. 
This extension is explored through specific examples in Sec.~\ref{sec:example-two-combinations}.

\section{ Example: dynamical impact of experiment of a solid sphere on a viscoelastic board }\label{sec:example-1}

To illustrate the present method, 
let us here discuss the experiment of dynamical impact of a solid sphere to a viscoelastic board~\cite{Maruoka_2023}.
We test our method using synthetic data as well as data obtained in experiments.
We provide source codes for the analyses in this and next sections in Ref.~\cite{repository}.

\subsection{ Setting }

The experiment investigates the response of a viscoelastic board when it is struck by a rigid sphere.
This scenario models the interaction between a rigid sphere and a Maxwell viscoelastic surface of finite thickness, based on the Winkler foundation model where only normal deformations are considered~\cite{Johnson_1985}. 
The observations from this experiment show that there is a crossover of scaling law between maximum deformation and impact velocity for different sizes of spheres. Elastic impact was observed for high-speed impacts with smaller spheres, while viscoelastic impact was observed for lower-speed impacts and larger spheres. This transition was governed by the inverse Deborah number, which includes the size of the sphere and the impact speed. The results were finally summarized by a self-similar solution of the second kind with two similarity parameters.

The setup of this experiment is illustrated in Fig.~\ref{fig:F2}. 
A viscoelastic board with thickness $h$ made of polydimethylsiloxane (PDMS) was placed on the floor.
To ensure the observed deformations stem solely from the material's bulk properties and not from any surface adhesion, 
the surface of the board is lightly coated with chalk.
The dynamical response of the board is modeled by the Maxwell foundation model,
which can be represented by a foundation consisting of Maxwell elements in which a dashpot with its viscous coefficient $\mu$ and an elastic spring with elastic modulus $E$ are serially connected.
A metallic sphere with radius $R$ and density $\rho$ is dropped 
from a set height to collide with the board at velocity $v_i$.
This impact results in the board deforming to a maximum deformation $\delta_m$.

\begin{figure}[t]
\includegraphics[width=0.9\columnwidth]{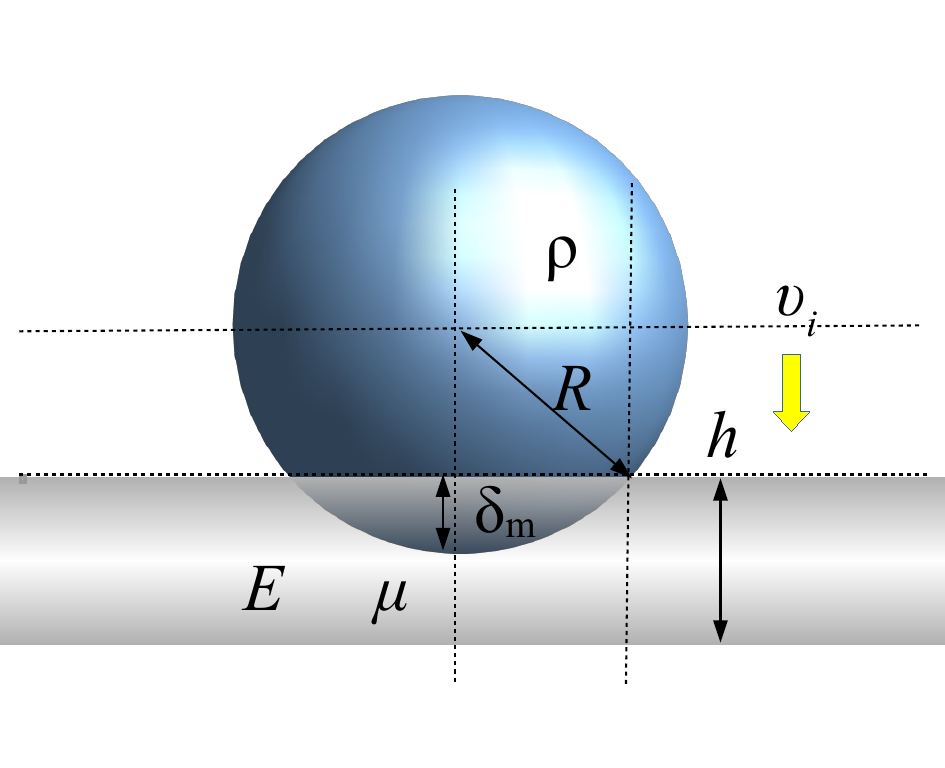}
\caption{The geometrical relation of physical parameters on the collision between a spherical impactor with radius $R$, density $\rho$, the impact velocity $v_i$, and a viscoelastic board with thickness $h$, elastic modulus $E$ and viscous coefficient $\mu$. $\delta_m$ is a maximum deformation after the impact. }
\label{fig:F2}
\end{figure}

In this problem, the maximum deformation can be determined by other parameters through a relation
\begin{equation}    
F\left(\delta_m, h, R, \rho, \mu, E, v_i \right) = 0.
\label{eq:f-rel}
\end{equation}
We introduce the dimensionless parameters as 
\begin{equation}    
\Pi \coloneqq \frac{\delta_m}{R} , \,\,\,\,
\kappa \coloneqq \frac{h}{R}, \,\,\,\,
\eta \coloneqq \frac{\rho v_i^2}{E}, \,\,\,\,
\theta \coloneqq \frac{\mu}{E^{1/2} \rho^{1/2} R}, 
\label{eq:dimless-vars_first}
\end{equation}
\changed{which were obtained by using the following physical parameters with independent dimensions, ${\bm z}^{(\I )} = \left(\rho, R, E \right)$}.
According to Ref.~\cite{Maruoka_2023}, using the dimensionless parameters, 
Eq.~\eqref{eq:f-rel} can be expressed by the following self-similar solution as 
\begin{equation}    
\Psi = \Phi \left( Z \right),
\label{eq:phi-truth}
\end{equation}
where 
\begin{equation}    
\Psi \coloneqq \frac{\Pi^3}{\kappa \eta} = \frac{\delta_m^3 E}{R^2 h \rho v_i^2} , \quad 
Z \coloneqq \frac{\Pi}{\theta \eta^{1/2}} = \frac{E \delta_m}{\mu v_i},
\label{eq:dimless-vars_second}
\end{equation}
and the function $\Phi$ is defined by 
\begin{equation}
\Phi (Z) \coloneqq \frac{2}{3}
\frac{Z}{1 - e^{-Z}} .
\label{eq:phi_EPJE}
\end{equation}
The problem belongs to the similarity of the second kind, which is the target of our approach.

Let $U_{\rm length},U_{\rm mass},U_{\rm time}$ be a chosen set of units for length, mass and time. Suppose that we are rescaling the units as
\begin{equation}
    (U_{\rm length},U_{\rm mass},U_{\rm time})
    \mapsto
    \left(\frac{U_{\rm length}}{L}, \frac{U_{\rm mass}}{M}, \frac{U_{\rm time}}{T}\right)\,.
    \label{eq:unit-rescale}
\end{equation}
In this case, the dimension function $\phi$ acting on the set of physical parameters
\begin{equation}    
\bm z = \left( \delta_m, h, R, \rho, \mu, E, v_i  \right)^\top\,
\end{equation}
is given as
\begin{equation}    
\phi= \left(L,\,L,\,L,\,\frac{M}{L^3},\,\frac{M}{LT},\,\frac{M}{LT^2}, \frac{L}{T} \right)^{\top}.
\end{equation}
An element $\varphi\in G_{\rm second}^{(\rm s)}$ acting on the dimensionless parameters $\bm{\pi} =(\Pi,\kappa,\eta,\theta)^{\top}$ is
\begin{equation}
\varphi = \left(A, A^3B^{-1}, B, AB^{-1/2} \right)^{\top}\,.
\end{equation}
One can easily check that this $\varphi$ leaves Eq.~\eqref{eq:dimless-vars_second} invariant.

In this problem, $\Psi = \frac{\Pi^3}{\kappa\eta}$ is a fixed similarity parameter of the second \changed{class}.
It is derived from the Chastel-Gondret-Mongruel solution~\cite{Chastel_2016,Chastel_2019,Mongruel_2020}, which is the scaling law in elastic impact regime. 
Here we assume that the CGM solution represents the intermediate asymptotics in ideal region where $Z \ll 1$.
That is, the similarity parameter $\Psi = \frac{\Pi^3}{\kappa\eta}$ approaches a constant as $Z \to 0$. Here we call $\Psi$ the CGM number.
On the other hand, the driving parameter which breaks the asymptotics is $\theta$. Therefore, the interfering similarity parameter should include $\theta$.

Suppose that we do not know the form of the function $\Phi$ and the combination 
of the argument of the RHS of Eq.~\eqref{eq:phi-truth}.
We would like to determine them through given data of parameters \eqref{eq:dimless-vars_first}. 
Let us parametrize the argument as 
\begin{equation}
\Psi = 
\Phi \left(
\Pi^{w^{(0)}_0} \kappa^{w^{(0)}_1} \eta^{w^{(0)}_2}\theta^{w^{(0)}_3}
\right).
\label{eq:1dim_parametrization}
\end{equation}
We introduce the log of the parameters,
\begin{equation}
\begin{split}    
&\psi \coloneq \ln \Psi, \quad 
x_0 \coloneqq \ln \Pi, \quad 
x_1 \coloneqq \ln \kappa, 
\\
&
x_2 \coloneqq \ln \eta, \quad 
x_3 \coloneqq \ln \theta,
\end{split}
\label{eq:log-variables}
\end{equation}
where the parameter $\psi$ can be written as 
\begin{equation}
\psi = 3 \ln \Pi - \ln \kappa - \ln \eta 
= 3 x_0 - x_1 - x_2    .
\end{equation}
The relation~\eqref{eq:1dim_parametrization} is now written as
\begin{equation}    
\psi = \ln \Phi 
\left( \exp \left[\sum_{i=0}^3 w^{(0)}_i x_i \right] \right)
\eqqcolon \phi \left(\sum_{i=0}^3 w^{(0)}_i x_i\right),
\end{equation}
where we have defined 
\begin{equation}
\phi (y) \coloneqq 
\ln \Phi(e^y). 
\end{equation}
In the following sections, we analyze this system by the neural network method described in the previous section.

\subsection{ Training with synthetic data } \label{sec:training-synthetic}
We verify that the neural network method works well with the data generated numerically according to Eq.~\eqref{eq:phi-truth}.
We prepare the training data in the following way:
\begin{itemize}
    \item[1)] Generate $N_{\rm data}$ sets of three random numbers uniformly from $[0, 1]$ and call them $(\Pi^{(i)}, \eta^{(i)}, \theta^{(i)})$ where $i=1,\cdots,N_{\rm data}$.
    \item[2)] Using $\Pi^{(i)}, \eta^{(i)}, \theta^{(i)}$ generated in Step 1, compute the RHS of Eq.~\eqref{eq:phi-truth} and call it $\Psi^{(i)}~(i=1,\cdots,N_{\rm data})$. Define $\kappa^{(i)}$ as $\kappa^{(i)}\coloneqq\Pi^{(i)3}/\Psi^{(i)} \eta^{(i)}$, so that $\Pi^{(i)}, \kappa^{(i)}, \eta^{(i)}, \theta^{(i)}$ and $\Psi^{(i)}$ satisfy the relation \eqref{eq:phi-truth}.
    \item[3)] For each of $\Pi^{(i)}, \kappa^{(i)}, \eta^{(i)}, \theta^{(i)}$ and $\Psi^{(i)}$, introduce noise as an artificial error as follows. Let $X$ be one of $\Pi^{(i)}, \kappa^{(i)}, \eta^{(i)}, \theta^{(i)}$ and $\Psi^{(i)}~(i=1,\cdots,N_{\rm data})$. The noise to $X$ is given as
    \begin{equation}
        \delta X = r \times \alpha_X \times X\,,
    \label{eq:noise}
    \end{equation}
    where $\alpha_X$ is a random variable sampled from the standard normal distribution $\mathcal{N}(0,1)$, and $r$ is an adjustable parameter, common to all $X$, that represents the strength of the noise,\footnote{
    For a given $X$, $\delta X$ is defined by sampling a random parameter $\alpha_X$ from the normal distribution $\mathcal{N}(0,1)$.
    If the sampling of $\alpha_X$ is repeated many times while $X$ is fixed, a set of data $\{X+\delta X\}$ will be obtained, which is scattered around $X$.
    The coefficient of variation of this data is nothing but $r$.
    In this sense, the parameter $r$ represents the strength of the noise.} and equals to coefficients of variation.
    \item[4)] Introduce the log of the parameters as
    \begin{align}
        &\psi^{(i)} \coloneq \ln(\Psi^{(i)}+\delta \Psi^{(i)})\,,\quad
        x_0^{(i)} \coloneqq \ln(\Pi^{(i)}+\delta\Pi^{(i)})\,,\notag 
        \\
        &x_1^{(i)} \coloneqq \ln(\kappa^{(i)}+\delta\kappa^{(i)})\,,\quad
        x_2^{(i)} \coloneqq \ln(\eta^{(i)}+\delta\eta^{(i)})\,, \\
        &
        x_3^{(i)} \coloneqq \ln(\theta^{(i)}+\delta\theta^{(i)})\,.\notag
    \end{align}
\end{itemize}
Using these $\psi^{(i)}$ and ${\bm x}^{(i)}$ as the training data, we perform the supervised learning.
If the training is done successfully, the parameters ${\bm w}^{(0)}$ of the first layer of the neural network should be close to
\begin{equation}
\begin{pmatrix}
1 & 0 & -1/2 & -1    
\end{pmatrix}^\top
\label{eq:true-value}
\end{equation}
up to the ambiguity of Eq.~\eqref{eq:ambiguity}.
Using the ambiguities, we can fix two of the components.
Specifically, we may fix the first and the last components as $w^{(0)}_0\to 1$ and $w^{(0)}_3\to -1$.
In practice, we perform this normalization as follows.
After training, we read off the parameters ${\bm w}^{(0)}$ of the first layer of the neural network.
Then, we normalize the last component $w_3^{(0)}$ of the ${\bm w}^{(0)}$ to $-1$ without loss of generality:
\begin{equation}
    {\bm w}^{(0)} \to \hat{{\bm w}}^{(0)} \coloneqq {\bm w}^{(0)}/(-w_3^{(0)})\,.
\label{eq:w_hat}
\end{equation}
After that, we normalize the first component of $\hat{{\bm w}}^{(0)}$ to unity by shifting a scalar times the vector
\begin{equation}
    {\bm q} = 
\begin{pmatrix}
3 & -1 & -1 & 0 
\end{pmatrix}^\top\,.
\label{eq:1dim_LHS_combination}
\end{equation}
Note that the LHS of Eq.~\eqref{eq:phi-truth} does not include $\theta$.
We denote the parameters after normalization by $\bm p$:
\begin{equation}
    {\bm p} \coloneqq \hat{{\bm w}}^{(0)} + \frac{1-\hat{w}_0^{(0)}}{3}{\bm q} = (1, p_1, p_2, -1)\,.
\end{equation}
After this normalization, we are interested in how close the remaining $p_1$ and $p_2$ are to the true values in Eq.~\eqref{eq:true-value}.

The structure of the neural network is 4-1-10-10-1,
where each number indicates the number nodes at each layer,
and 
we used the ReLU activation function. 
The number of data points is $N_{\rm data} = 500$.
As for the magnitudes of the regularization terms in Eq.~\eqref{eq:regularization}, we set $\lambda_1=\lambda_2=\lambda$ and tried training for $\lambda=0,\,0.5,\,1$.
For $\lambda=0$, the number of the epochs is set to $2,000$.
For $\lambda=0.5$ and $1$, more epochs are needed to obtain stable results, which in this case was set to $10,000$.

Figure~\ref{fig:3} shows the results of the training without the regularization terms, $\lambda=0$.
\footnote{The results shown hereafter are an example of several attempts of training. In all cases, if the value of the loss becomes sufficiently small, the training seems to be successful.}
The noise of the synthetic training data is increased in order from Fig.~\ref{fig:3a} to \ref{fig:3f}.
The plots demonstrate successful data collapse across all noise strengths.

Table \ref{tb:1} and Figs.~\ref{fig:4a}, \ref{fig:4b} show the obtained parameters and its relation with the strength of noise.
Here the reference values are obtained in a single training with all synthetic data and the statistical uncertainties are estimated by the bootstrap method.
This bootstrap method consists of repeating the process of resampling $N_{\rm data}$ data randomly from the original data set, allowing duplicates, and training with these resampled data.
The number of bootstrap resampling, $N_{\rm bs}$, is set to $N_{\rm bs} = 50$ in the present case.\footnote{During $N_{\rm bs}$ iterations of training with the bootstrap method, sometimes the same results are output repeatedly, which suggests that the parameter is trapped in a local minimum. 
Whether this behavior happens depends on the choice of optimization algorithm and the loss landscape of the problem. When one actually estimates uncertainty with the bootstrap method, such trapped results should be dismissed.}
We observe that the uncertainty of the learned parameters increases with the noise strength, $r$.
The estimation of the true values by the neural network 
delivers highly accurate results for the synthetic data up to a noise strength of approximately $r \sim 0.1$.
Beyond this noise strength, the learned parameters begin to deviate from the true values. 
Despite this divergence at higher noise strengths, it is evident that the learning process is functioning effectively.

In Fig.~\ref{fig:4c}, we show the dependence of the loss value on the strength of the noise.
Here, the loss function we employed is Eq.~\eqref{eq:MSE_loss}.
If the learning is successful without falling into overfitting, the neural network will extract average features from the training data fluctuated by noise.
In this case, the loss is expected to be roughly proportional to the square of the noise strength.
Indeed, the loss value can be closely approximated by a quadratic function as shown in the figure, signifying that overfitting is avoided.

\begin{figure*}[t]
    \begin{tabular}{ccc}
      \begin{minipage}[t]{0.33\hsize}
        \centering
        \includegraphics[width=.95\columnwidth]{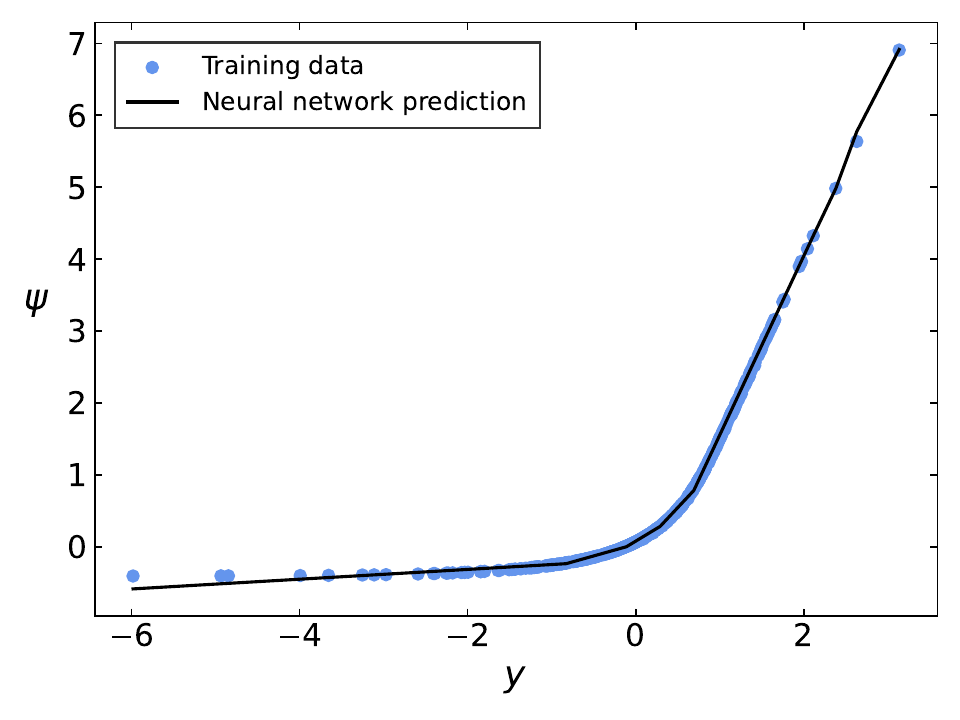}
        \subcaption{$r=0$}
        \label{fig:3a}
      \end{minipage} &
      \begin{minipage}[t]{0.33\hsize}
        \centering
        \includegraphics[width=.95\columnwidth]{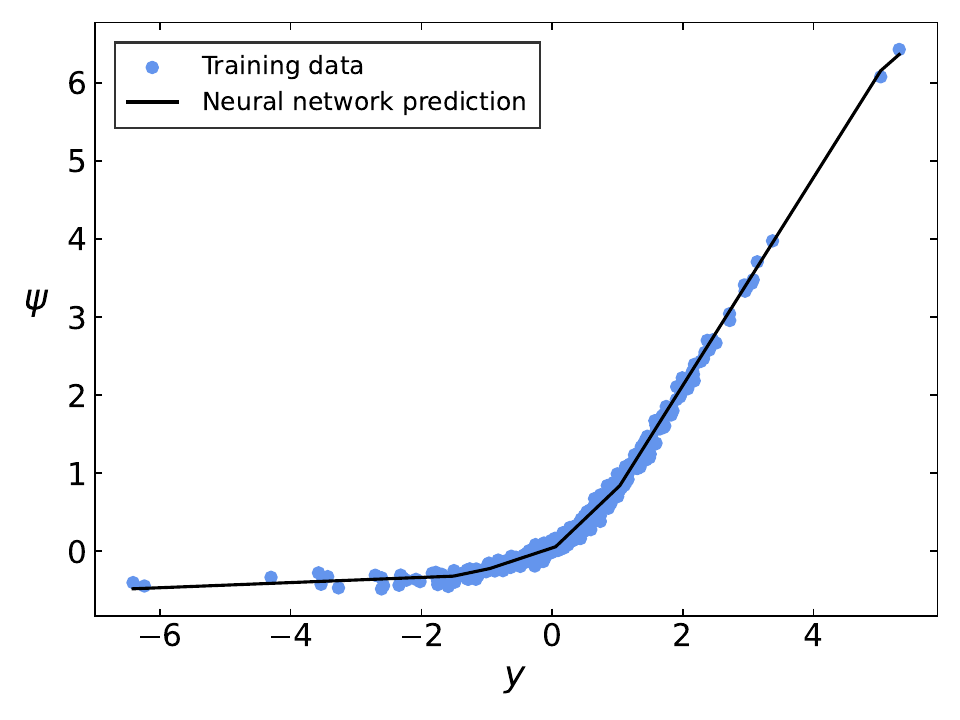}
        \subcaption{$r=0.05$}
        \label{fig:3b}
      \end{minipage}&
      \begin{minipage}[t]{0.33\hsize}
        \centering
        \includegraphics[width=.95\columnwidth]{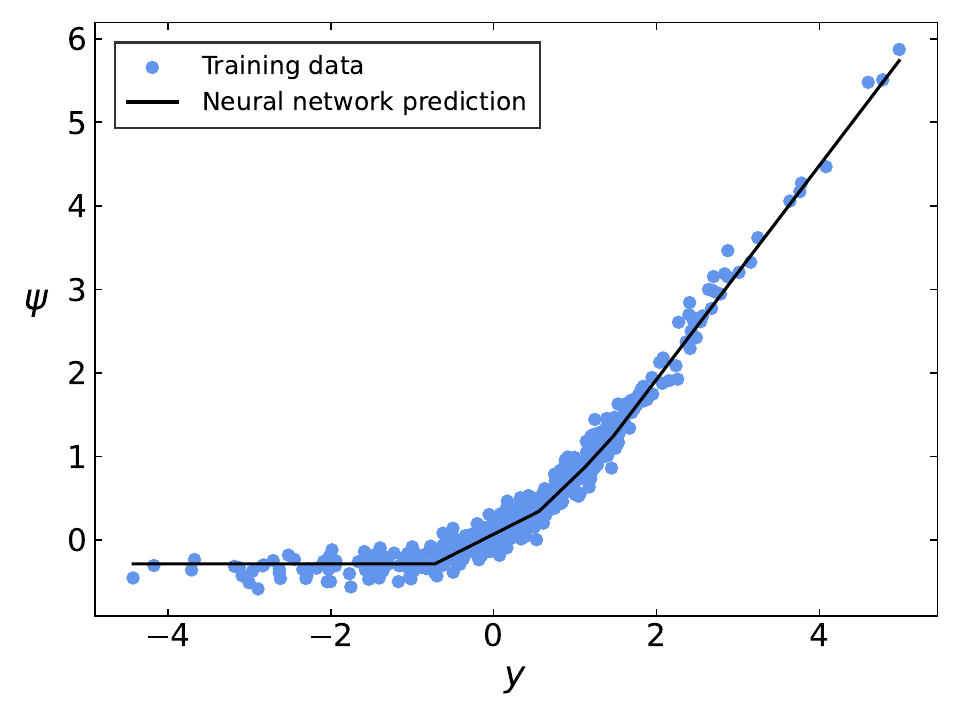}
        \subcaption{$r=0.1$}
        \label{fig:3c}
      \end{minipage} \\
      
      \begin{minipage}[t]{0.33\hsize}
        \centering
        \includegraphics[width=.95\columnwidth]{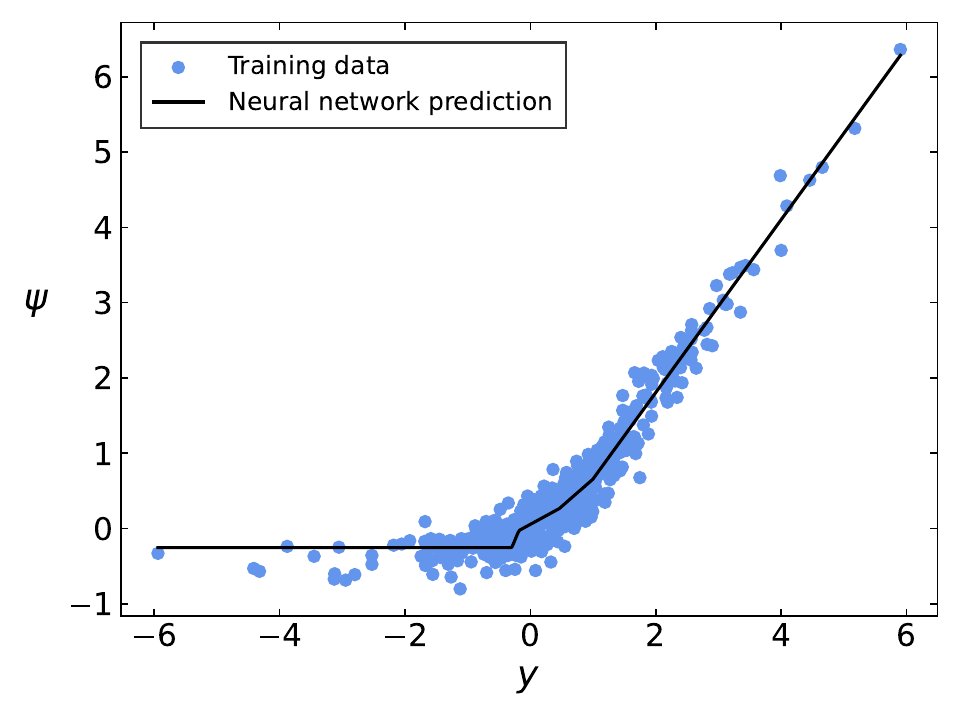}
        \subcaption{$r=0.15$}
        \label{fig:3d}
      \end{minipage} &
      \begin{minipage}[t]{0.33\hsize}
        \centering
        \includegraphics[width=.95\columnwidth]{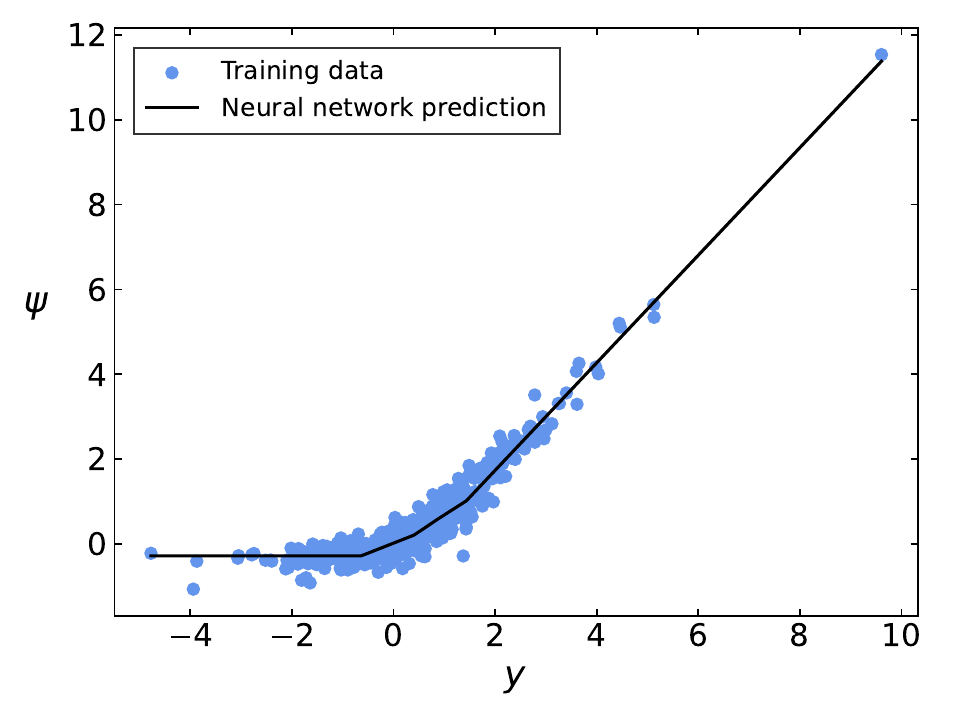}
        \subcaption{$r=0.2$}
        \label{fig:3e}
      \end{minipage}&
      \begin{minipage}[t]{0.33\hsize}
        \centering
        \includegraphics[width=.95\columnwidth]{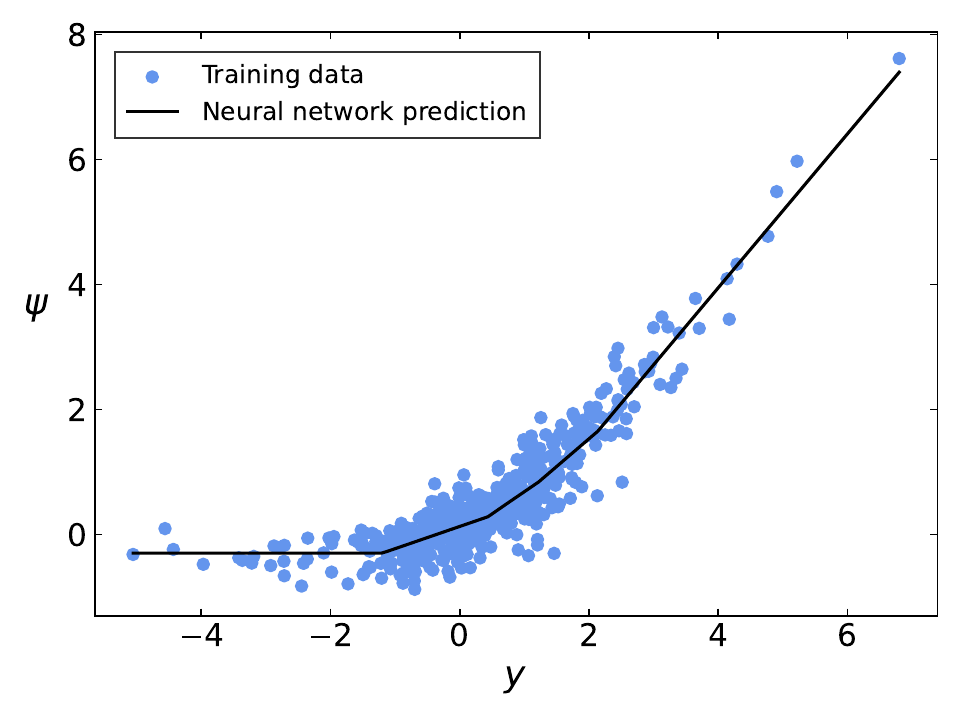}
        \subcaption{$r=0.25$}
        \label{fig:3f}
      \end{minipage} 
    \end{tabular}
    \caption{Synthetic training data ${\psi^{(i)}}$ and neural network predictions $\psi_{\rm NN}^{(i)}$ (without the regularization terms, $\lambda = 0$) as a function of $y^{(i)} = \hat{{\bm w}}^{(0)}\cdot {\bm x}^{(i)}$, where $\hat{{\bm w}}^{(0)}$ is given as \eqref{eq:w_hat}. Here, $r$ is the noise strength introduced in the synthetic data $\{({\bm x}^{(i)},\psi^{(i)})\}$.}
    \label{fig:3}
\end{figure*}

\begin{table}[t]
  \centering
  \begin{tabular}{ccc}
    $r$ & $p_1$ & $p_2$ \vspace{.5mm}\\ \hline
    $0$ & $-0.0025\pm 0.0008$ & $-0.5015\pm 0.0011$ \vspace{.5mm}\\ 
    \vspace{.5mm}
    $0.05$ & $-0.0033\pm 0.0034$ & $-0.5039\pm 0.0044$ \vspace{.5mm}\\ 
    \vspace{.5mm}
    $0.1$ & $0.0054\pm 0.0043$ & $-0.5198\pm 0.0079$ \vspace{.5mm}\\ 
    \vspace{.5mm}
    $0.15$ & $-0.007\pm 0.010$ & $-0.539\pm 0.015$ \vspace{.5mm}\\ 
    \vspace{.5mm}
    $0.2$ & $-0.0178\pm 0.0078$ & $-0.531\pm 0.020$ \vspace{.5mm}\\ 
    \vspace{.5mm}
    $0.25$ & $-0.005\pm 0.020$ & $-0.475\pm 0.025$ \vspace{.5mm}\\ 
    \hline
    \hline
    true & $0$ & $-0.5$    
  \end{tabular}
  \caption{The dependence of the normalized parameters $p_1$ and $p_2$ on the noise strength $r$, with the regularization term $\lambda = 0$.
  For each value of $r$, the reference values of the parameters are obtained in a single training with all data. The uncertainties are estimated by the bootstrap method.
  }
  \label{tb:1}
\end{table}

\begin{figure*}[t]
    \begin{tabular}{ccc}
      \begin{minipage}[t]{0.33\hsize}
        \centering
        \includegraphics[width=.95\columnwidth]{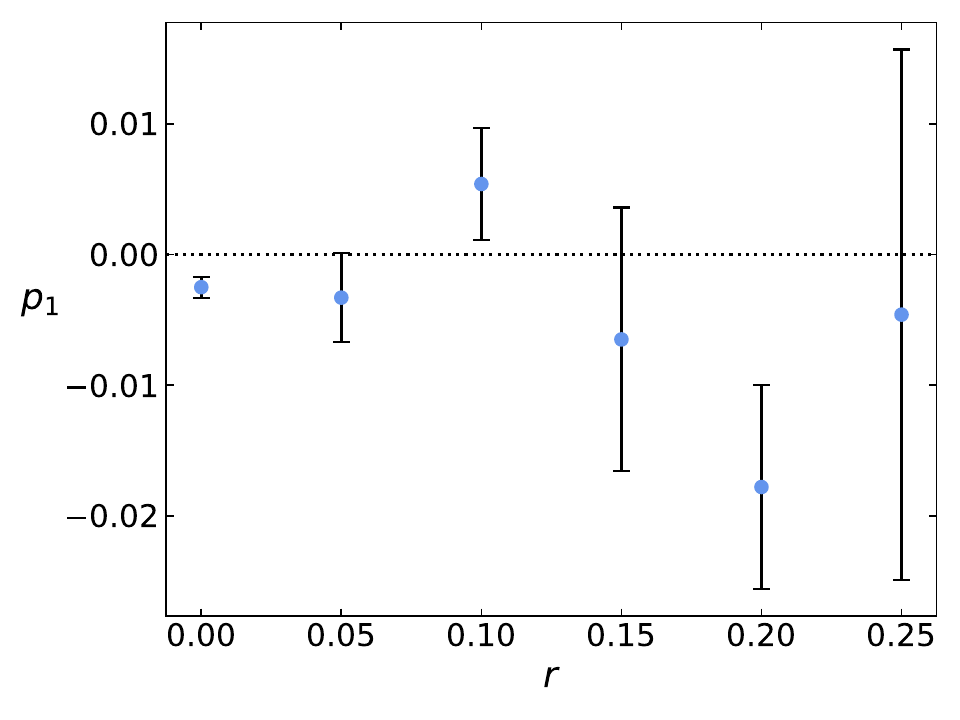}
        \subcaption{Parameter $p_1$ as a function of $r$. The true value is $p_1=0$.}
        \label{fig:4a}
      \end{minipage} &
      \begin{minipage}[t]{0.33\hsize}
        \centering
        \includegraphics[width=.95\columnwidth]{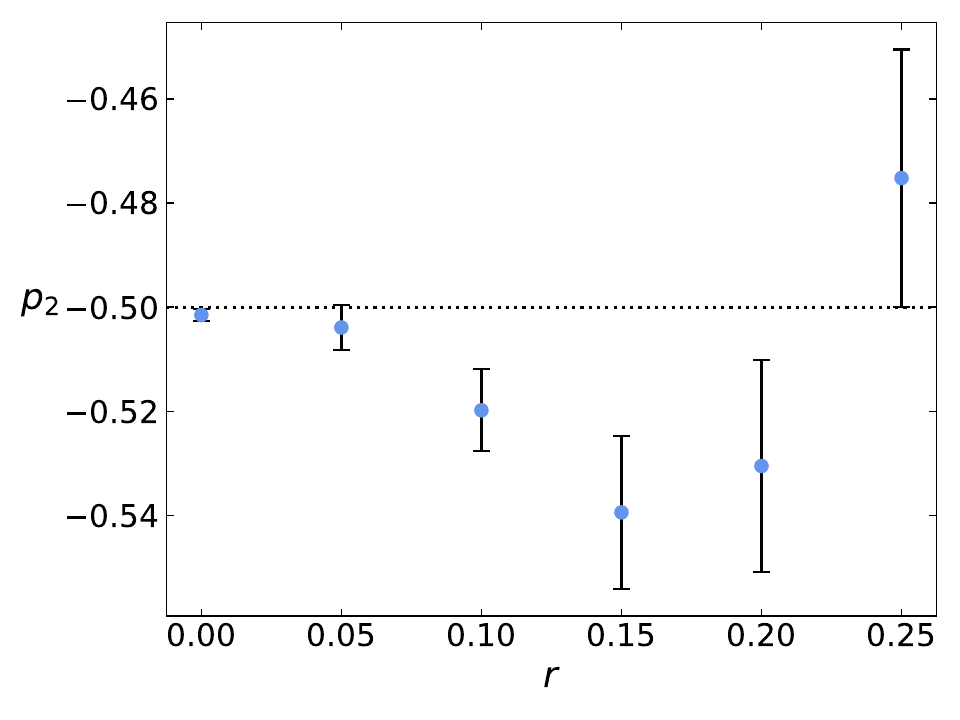}
        \subcaption{Parameter $p_2$ as a function of $r$. The true value is $p_2=-1/2$.}
        \label{fig:4b}
      \end{minipage}&
      \begin{minipage}[t]{0.33\hsize}
        \centering
        \includegraphics[width=.95\columnwidth]{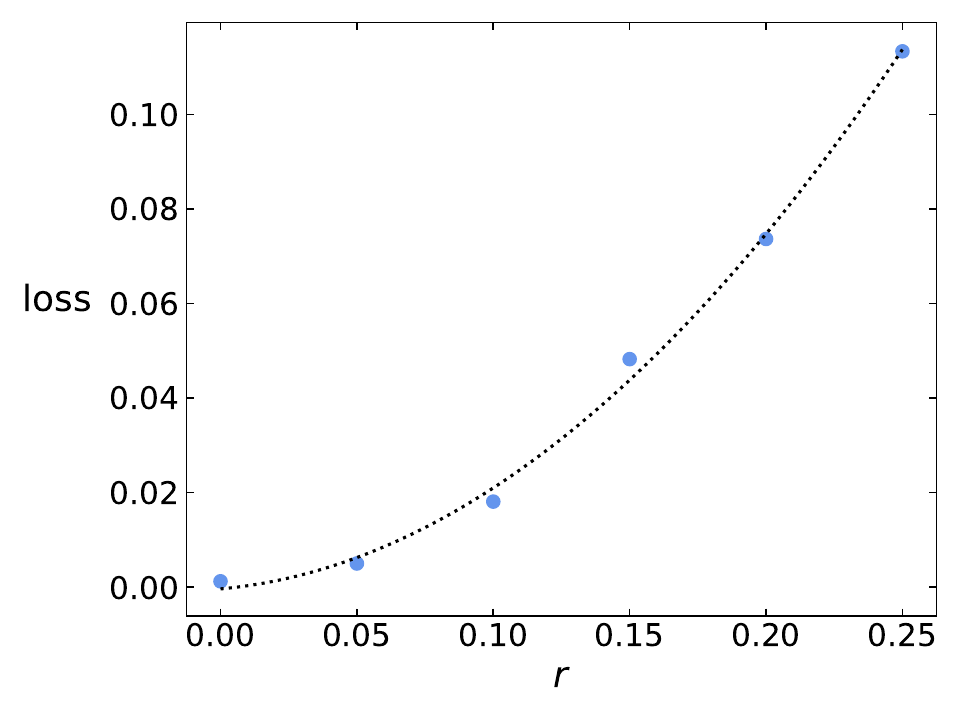}
        \subcaption{The loss as a function of $r$. The dotted line is a fit by a quadratic function (${\rm loss} = 1.62\,r^2 + 0.0517\,r -0.0004$).}
        \label{fig:4c}
      \end{minipage} 
    \end{tabular}
    \caption{The trained parameters after the normalization procedure and the values of the loss as a function of the noise strength in the synthetic data without the regularization term $\lambda = 0$. In Fig.~\ref{fig:4a} and \ref{fig:4b}, the error bars are estimated by the bootstrap method.}
    \label{fig:4}
\end{figure*}

Next, we turn on the regularization terms \eqref{eq:regularization}.
Figures~\ref{fig:5} and \ref{fig:7} show the results of the training for $\lambda=0.5$ and $1$, respectively.
The value of $\lambda$ was chosen so that the constraints set by the regularization terms \eqref{eq:regularization} work properly.
Table \ref{tb:2} and Fig.~\ref{fig:6a}, \ref{fig:6b} show the learned parameters after normalization, for $\lambda=0.5$.
Table \ref{tb:3} and Fig.~\ref{fig:8a}, \ref{fig:8b} are for $\lambda=1$.
Obviously, the training with regularization terms works as well as that without a regularization term ($\lambda=0$).
Thus, by adding regularization terms to the loss function, certain constraints can be imposed on the training parameters without spoiling the accuracy of the training.

\begin{figure*}[t]
    \begin{tabular}{ccc}
      \begin{minipage}[t]{0.33\hsize}
        \centering
        \includegraphics[width=.95\columnwidth]{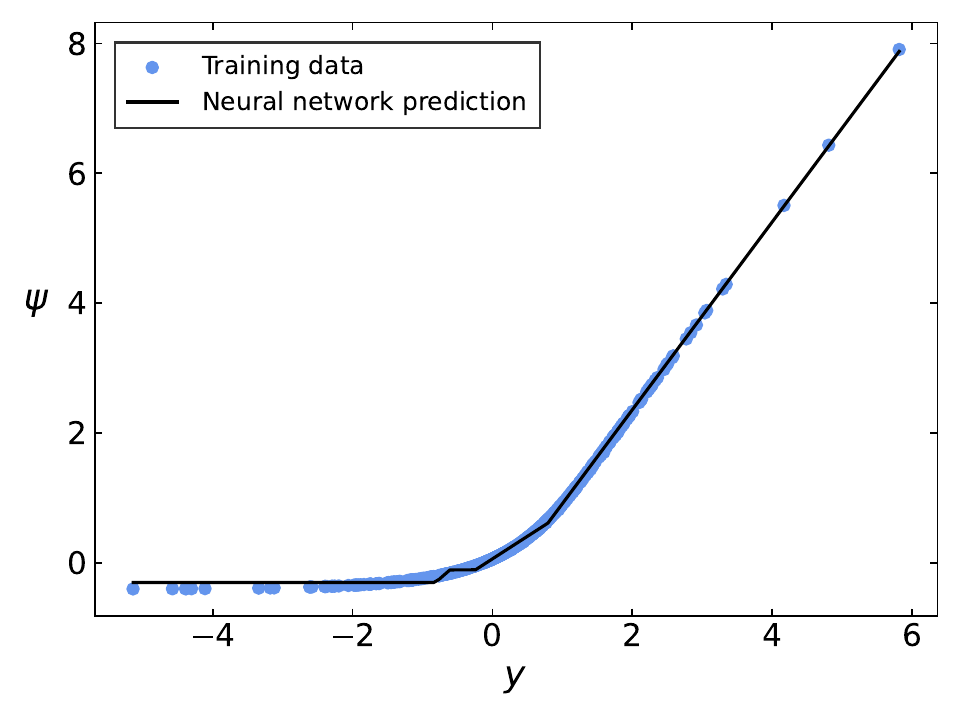}
        \subcaption{$r=0$}
        \label{fig:5a}
      \end{minipage} &
      \begin{minipage}[t]{0.33\hsize}
        \centering
        \includegraphics[width=.95\columnwidth]{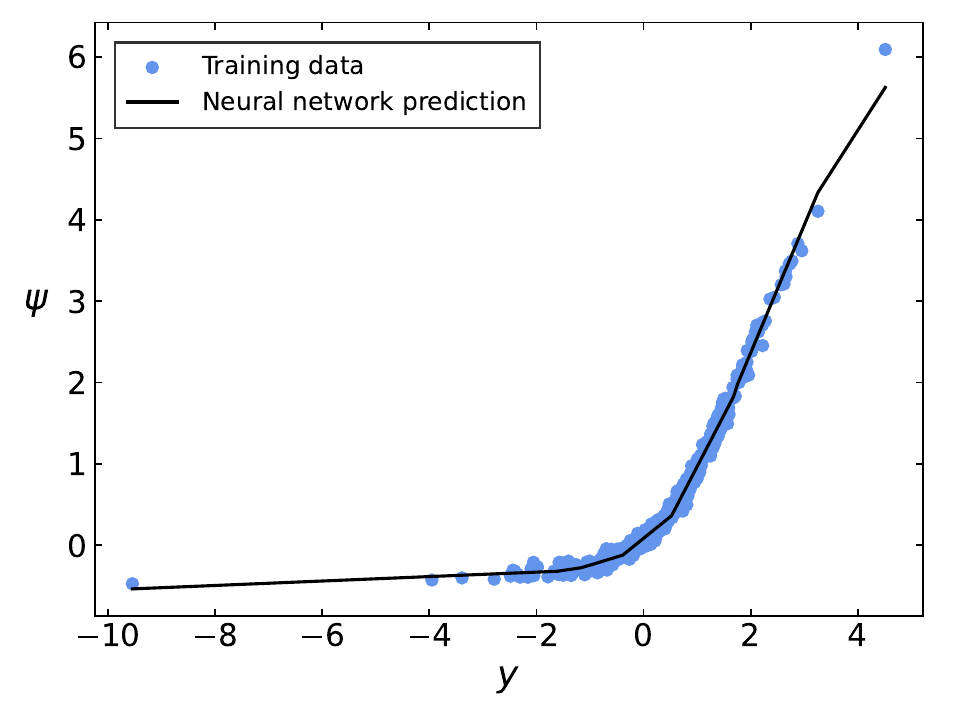}
        \subcaption{$r=0.05$}
        \label{fig:5b}
      \end{minipage}&
      \begin{minipage}[t]{0.33\hsize}
        \centering
        \includegraphics[width=.95\columnwidth]{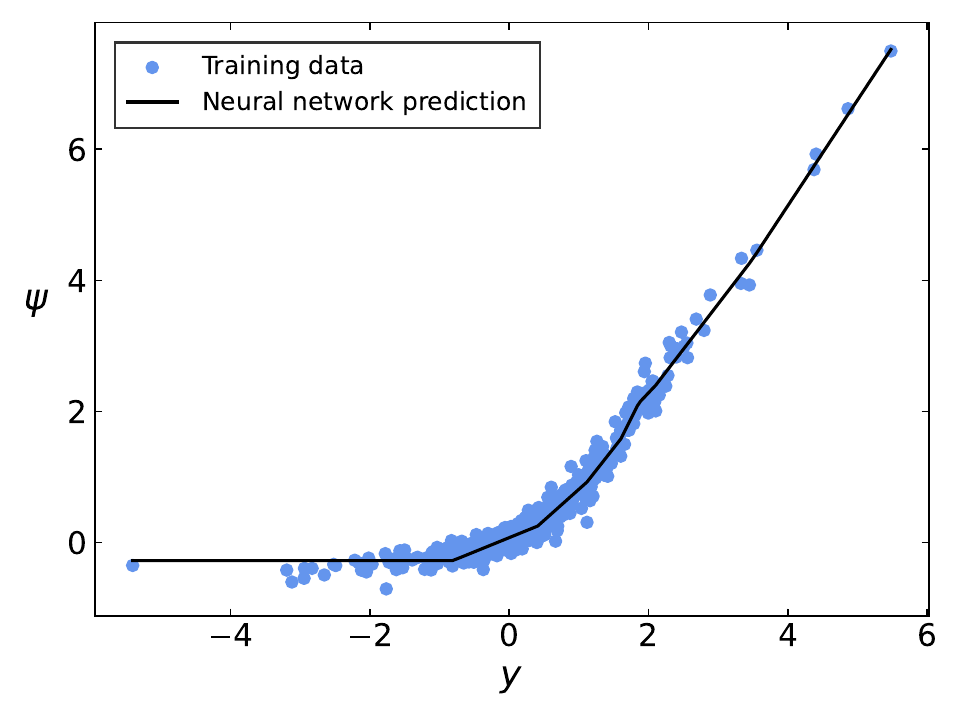}
        \subcaption{$r=0.1$}
        \label{fig:5c}
      \end{minipage} \\
      
      \begin{minipage}[t]{0.33\hsize}
        \centering
        \includegraphics[width=.95\columnwidth]{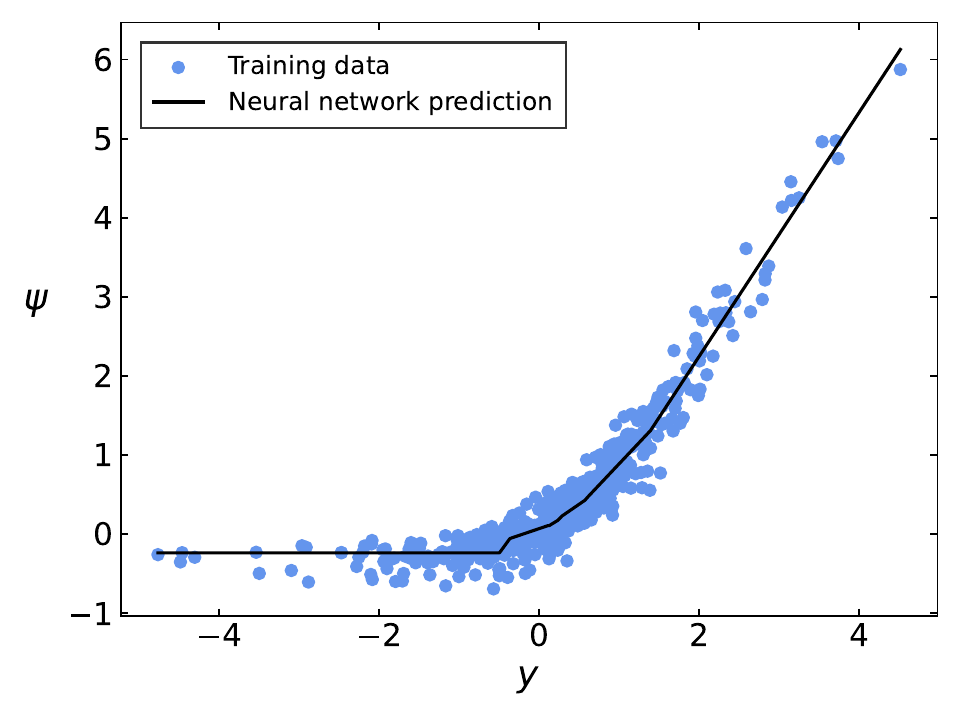}
        \subcaption{$r=0.15$}
        \label{fig:5d}
      \end{minipage} &
      \begin{minipage}[t]{0.33\hsize}
        \centering
        \includegraphics[width=.95\columnwidth]{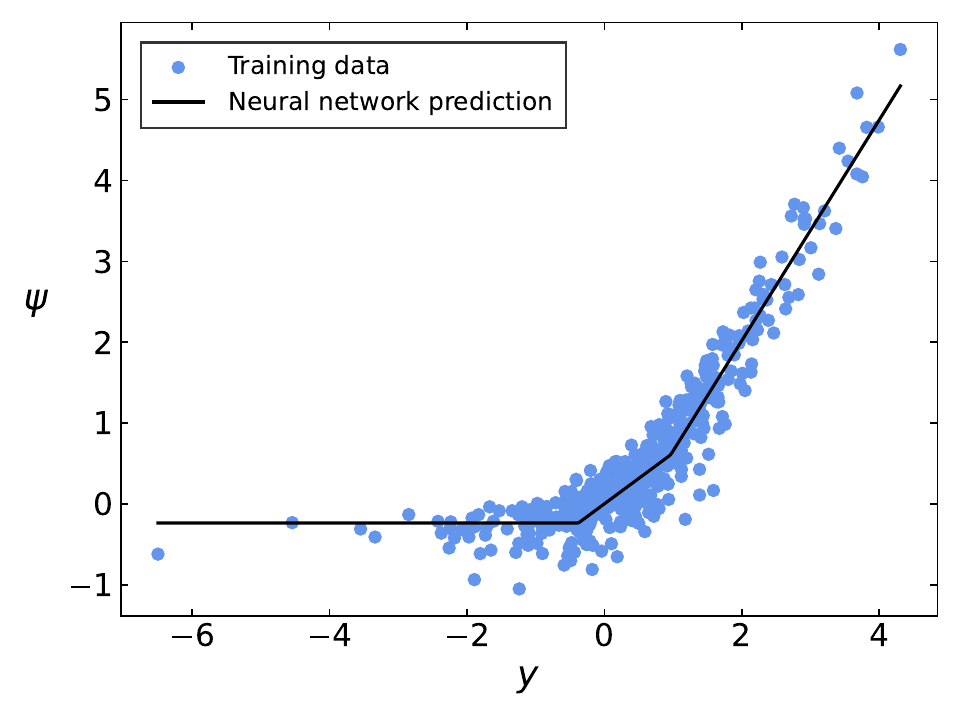}
        \subcaption{$r=0.2$}
        \label{fig:5e}
      \end{minipage}&
      \begin{minipage}[t]{0.33\hsize}
        \centering
        \includegraphics[width=.95\columnwidth]{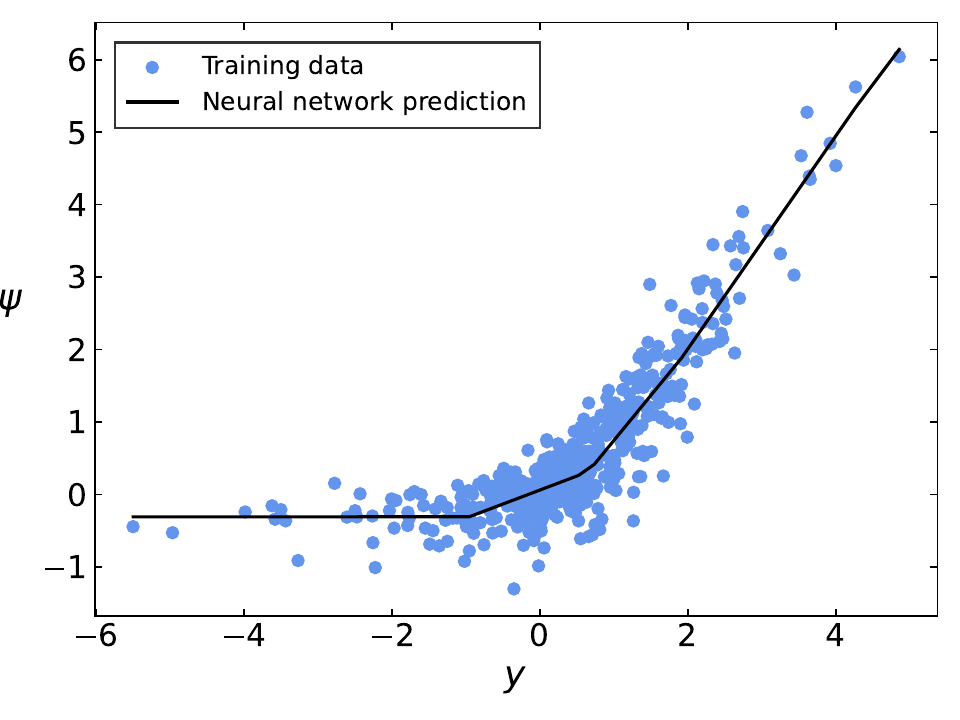}
        \subcaption{$r=0.25$}
        \label{fig:5f}
      \end{minipage} 
    \end{tabular}
    \caption{Synthetic training data ${\psi^{(i)}}$ and neural network predictions $\psi_{\rm NN}^{(i)}$ (with the regularization terms, $\lambda = 0.5$) as a function of $y^{(i)} = \hat{{\bm w}}^{(0)}\cdot {\bm x}^{(i)}$. Here, $r$ is the noise strength introduced in the synthetic data $\{({\bm x}^{(i)},\psi^{(i)})\}$.}
    \label{fig:5}
\end{figure*}

\begin{table}[t]
  \centering
  \begin{tabular}{ccc}
    $r$ & $p_1$ & $p_2$ \vspace{.5mm}\\ \hline
    $0$ & $-0.0026\pm 0.0007$ & $-0.5020\pm 0.0008$ \vspace{.5mm}\\ 
    \vspace{.5mm}
    $0.05$ & $0.0055\pm 0.0027$ & $-0.4888\pm 0.0047$ \vspace{.5mm}\\ 
    \vspace{.5mm}
    $0.1$ & $-0.0035\pm 0.0044$ & $-0.5190\pm 0.0078$ \vspace{.5mm}\\ 
    \vspace{.5mm}
    $0.15$ & $0.0012\pm 0.0084$ & $-0.504\pm 0.016$ \vspace{.5mm}\\ 
    \vspace{.5mm}
    $0.2$ & $-0.040\pm 0.011$ & $-0.536\pm 0.019$ \vspace{.5mm}\\ 
    \vspace{.5mm}
    $0.25$ & $-0.035\pm 0.017$ & $-0.532\pm 0.020$ \vspace{.5mm}\\ 
    \hline
    \hline
    true & $0$ & $-0.5$  
  \end{tabular}
  \caption{The dependence of the normalized parameters $p_1$ and $p_2$ on the noise strength $r$, with the regularization term $\lambda = 0.5$. For each value of $r$, the reference values of the parameters are obtained in a single training with all data. The uncertainties are estimated by the bootstrap method.
  }
  \label{tb:2}
\end{table}

\begin{figure*}[t]
    \begin{tabular}{cc}
      \begin{minipage}[t]{0.33\hsize}
        \centering
        \includegraphics[width=.95\columnwidth]{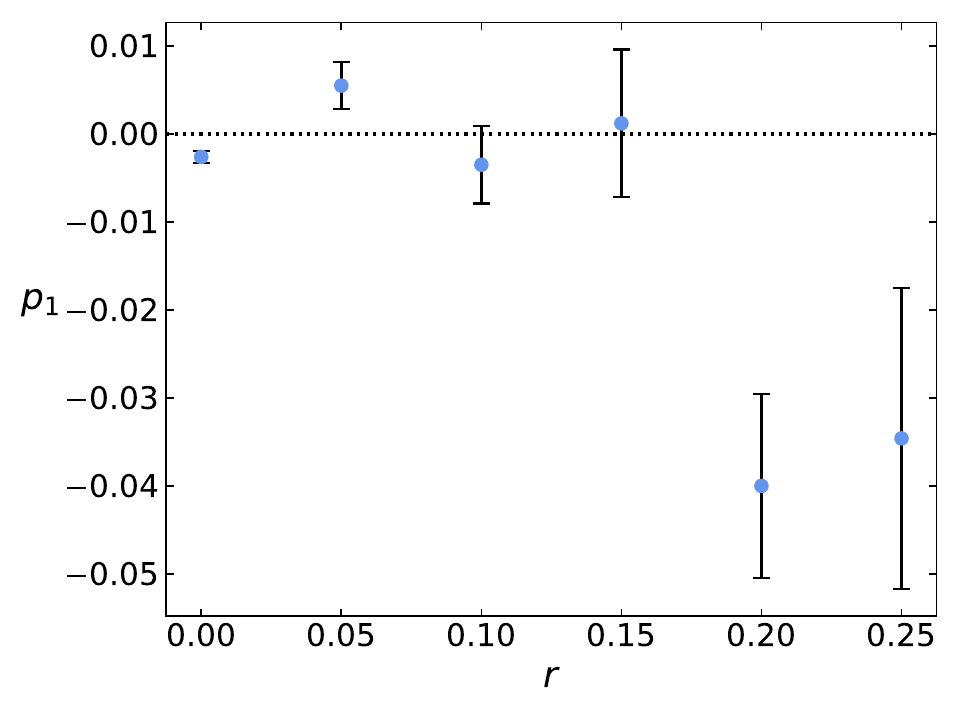}
        \subcaption{Parameter $p_1$ as a function of $r$. The true value is $p_1=0$.}
        \label{fig:6a}
      \end{minipage} &
      \begin{minipage}[t]{0.33\hsize}
        \centering
        \includegraphics[width=.95\columnwidth]{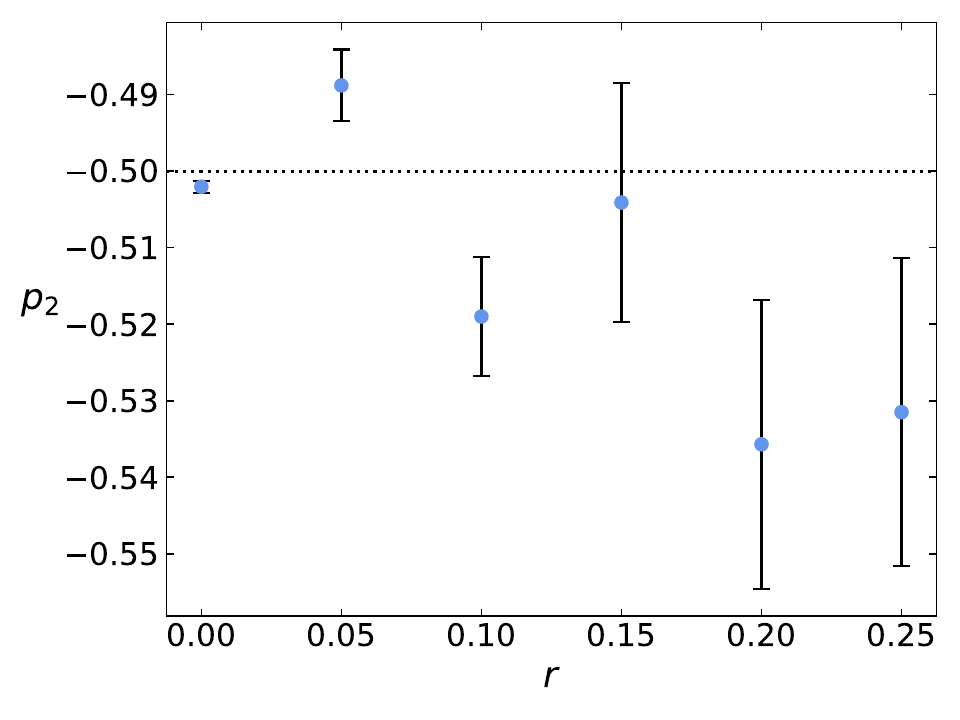}
        \subcaption{Parameter $p_2$ as a function of $r$. The true value is $p_2=-1/2$.}
        \label{fig:6b}
      \end{minipage}
    \end{tabular}
    \caption{The trained parameters after the normalization procedure and the values of the loss as a function of the noise strength in the synthetic data with the regularization term $\lambda = 0.5$. In Fig.~\ref{fig:6a} and \ref{fig:6b}, the error bars are estimated by the bootstrap method.}
    \label{fig:6}
\end{figure*}

\begin{figure*}[t]
    \begin{tabular}{ccc}
      \begin{minipage}[t]{0.33\hsize}
        \centering
        \includegraphics[width=.95\columnwidth]{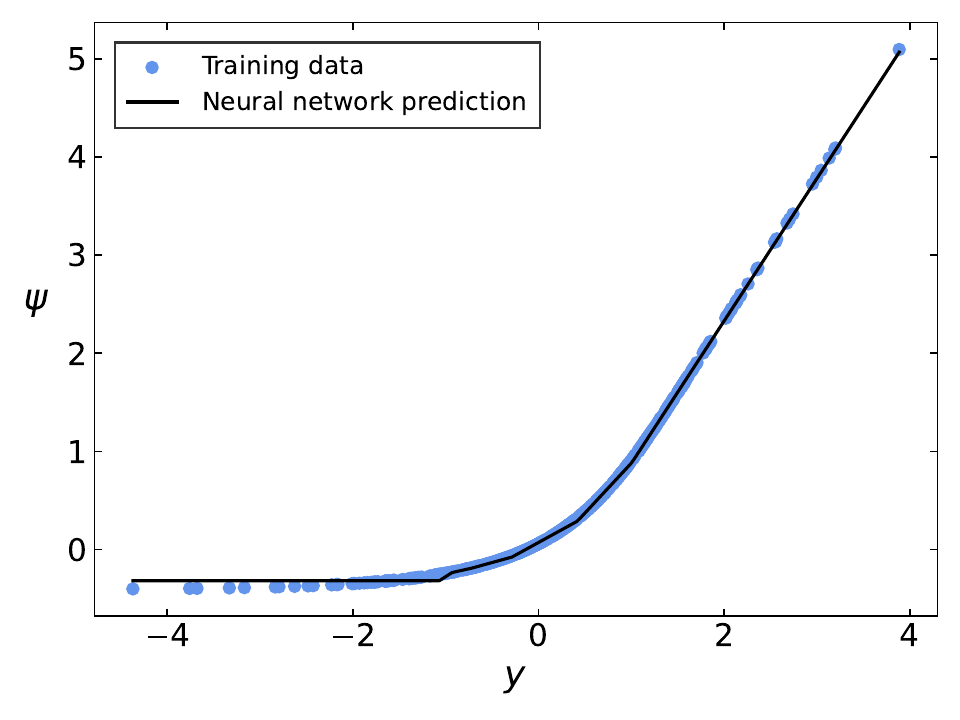}
        \subcaption{$r=0$}
        \label{fig:7a}
      \end{minipage} &
      \begin{minipage}[t]{0.33\hsize}
        \centering
        \includegraphics[width=.95\columnwidth]{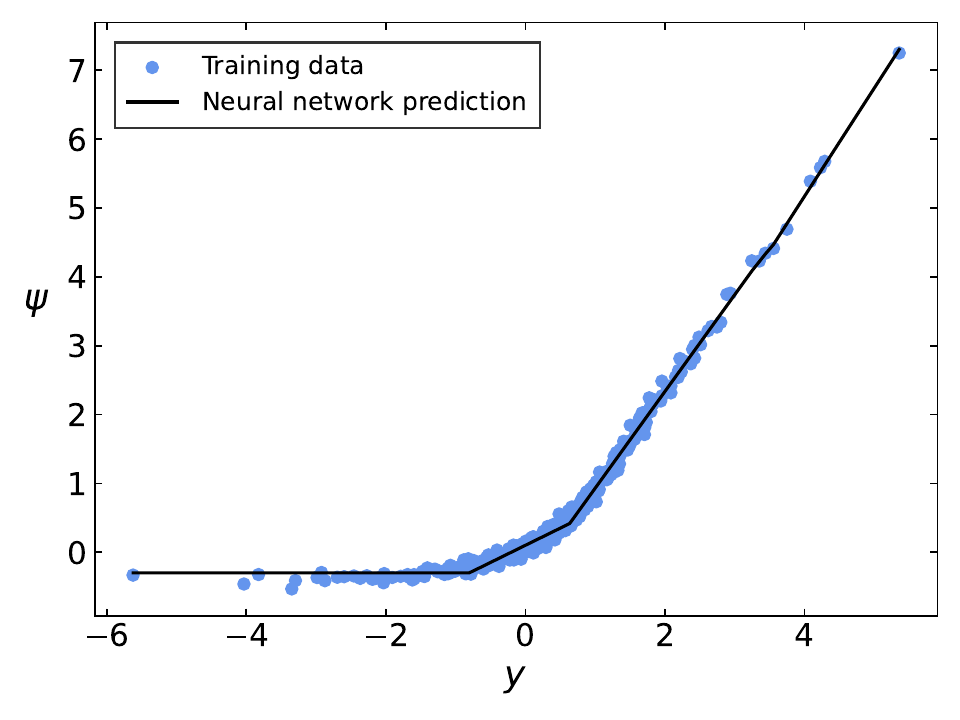}
        \subcaption{$r=0.05$}
        \label{fig:7b}
      \end{minipage}&
      \begin{minipage}[t]{0.33\hsize}
        \centering
        \includegraphics[width=.95\columnwidth]{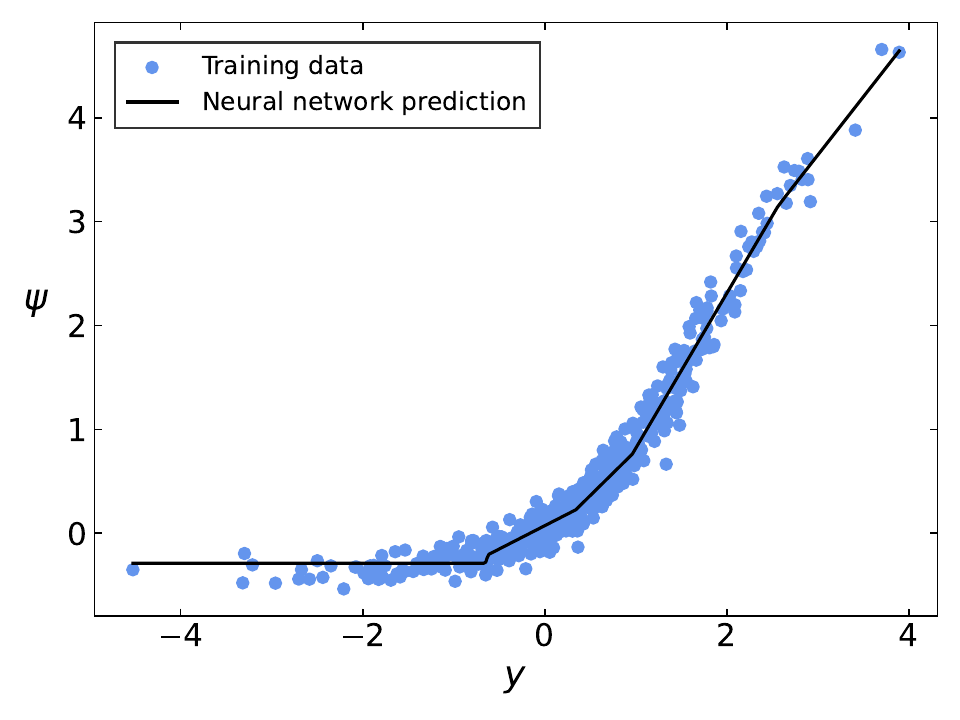}
        \subcaption{$r=0.1$}
        \label{fig:7c}
      \end{minipage} \\
      
      \begin{minipage}[t]{0.33\hsize}
        \centering
        \includegraphics[width=.95\columnwidth]{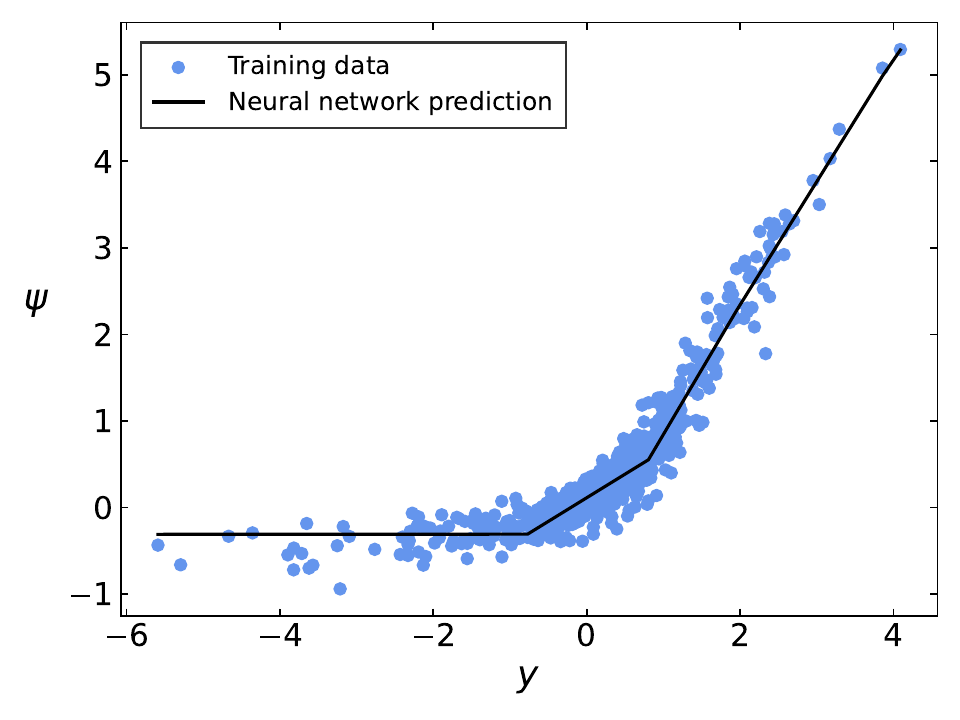}
        \subcaption{$r=0.15$}
        \label{fig:7d}
      \end{minipage} &
      \begin{minipage}[t]{0.33\hsize}
        \centering
        \includegraphics[width=.95\columnwidth]{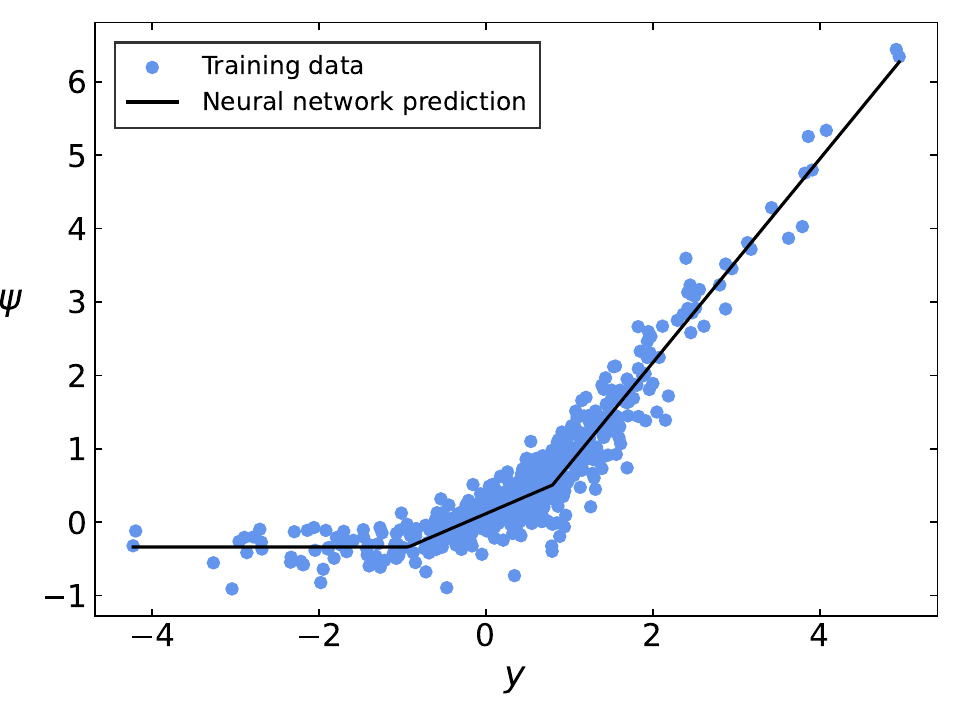}
        \subcaption{$r=0.2$}
        \label{fig:7e}
      \end{minipage}&
      \begin{minipage}[t]{0.33\hsize}
        \centering
        \includegraphics[width=.95\columnwidth]{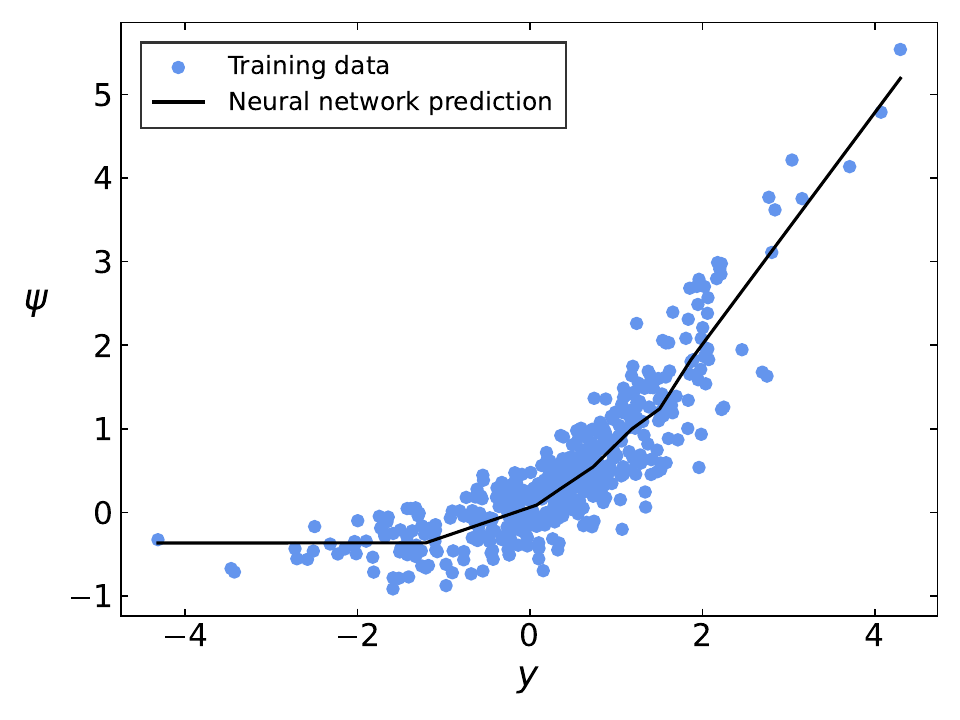}
        \subcaption{$r=0.25$}
        \label{fig:7f}
      \end{minipage} 
    \end{tabular}
    \caption{Synthetic training data ${\psi^{(i)}}$ and neural network predictions $\psi_{\rm NN}^{(i)}$ (with the regularization terms, $\lambda = 1$) as a function of $y^{(i)} = \hat{{\bm w}}^{(0)}\cdot {\bm x}^{(i)}$. Here, $r$ is the noise strength introduced in the synthetic data $\{({\bm x}^{(i)},\psi^{(i)})\}$.}
    \label{fig:7}
\end{figure*}

\begin{table}[t]
  \centering
  \begin{tabular}{ccc}
    $r$ & $p_1$ & $p_2$ \vspace{.5mm}\\ \hline
    $0$ & $-0.0011\pm 0.0004$ & $-0.5013\pm 0.0007$ \vspace{.5mm}\\ 
    \vspace{.5mm}
    $0.05$ & $-0.0015\pm 0.0033$ & $-0.4945\pm 0.0047$ \vspace{.5mm}\\ 
    \vspace{.5mm}
    $0.1$ & $-0.0058\pm 0.0056$ & $-0.5056\pm 0.0070$ \vspace{.5mm}\\ 
    \vspace{.5mm}
    $0.15$ & $0.0131\pm 0.0078$ & $-0.479\pm 0.014$ \vspace{.5mm}\\ 
    \vspace{.5mm}
    $0.2$ & $-0.005\pm 0.012$ & $-0.522\pm 0.016$ \vspace{.5mm}\\ 
    \vspace{.5mm}
    $0.25$ & $-0.013\pm 0.018$ & $-0.490\pm 0.032$ \vspace{.5mm}\\ 
    \hline
    \hline
    true & $0$ & $-0.5$   
  \end{tabular}
  \caption{The dependence of the normalized parameters $p_1$ and $p_2$ on the noise strength $r$ with the regularization term $\lambda = 1$.
  For each value of $r$, the reference values of the parameters are obtained in a single training with all data. The uncertainties are estimated by the bootstrap method.
  }
  \label{tb:3}
\end{table}

\begin{figure*}[t]
    \begin{tabular}{cc}
      \begin{minipage}[t]{0.33\hsize}
        \centering
        \includegraphics[width=.95\columnwidth]{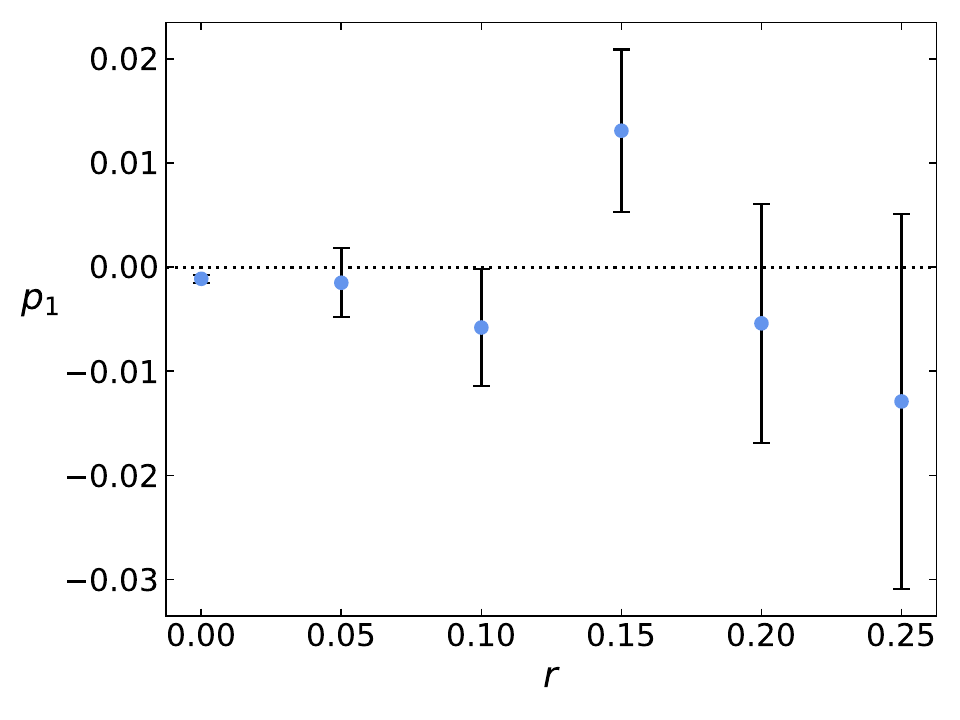}
        \subcaption{Parameter $p_1$ as a function of $r$. The true value is $p_1=0$.}
        \label{fig:8a}
      \end{minipage} &
      \begin{minipage}[t]{0.33\hsize}
        \centering
        \includegraphics[width=.95\columnwidth]{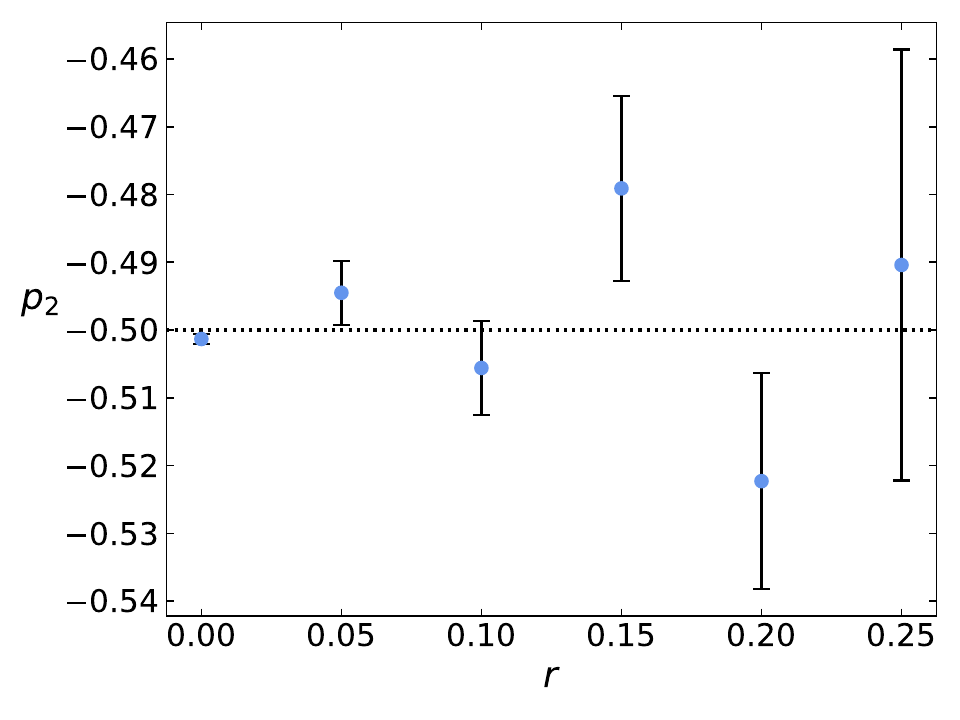}
        \subcaption{Parameter $p_2$ as a function of $r$. The true value is $p_2=-1/2$.}
        \label{fig:8b}
      \end{minipage}
    \end{tabular}
    \caption{The trained parameters after the normalization procedure and the values of the loss as a function of the noise strength in the synthetic data with the regularization term $\lambda = 1$. In Fig.~\ref{fig:8a} and \ref{fig:8b}, the error bars are estimated by the bootstrap method.}
    \label{fig:8}
\end{figure*}

\subsection{Training with experimental data}

Now we turn to training with actual experimental data. The data are taken from Ref.~\cite{Maruoka_2023} with other data plots using the different PDMS surface and different spheres. The size of the data set is $N_{\rm data} = 127$.
The architecture of the neural network is the same as that of the previous section.
To begin with, we performed training without regularization terms, which resulted in the following estimates of the parameters:
\begin{equation}
    {\bm w}^{(0)} = 
    \begin{pmatrix}
    0.751 & -0.245 & -0.250 & -0.006
    \end{pmatrix}^\top\,.
\end{equation}
This combination of parameters is almost proportional to the already known combination \eqref{eq:1dim_LHS_combination} which appears in the LHS of Eq.~\eqref{eq:1dim_parametrization}.
In other words, this result is trivial and the neural network has not learned non-trivial self-similarity.
To avoid finding trivial relations, we need to introduce regularization terms.
We show the results of training with regularization terms, 
$\lambda=0.5$ and $1$, in Fig.~\ref{fig:12}.
The plots are scattered compared to the case of the synthetic data, which may be due to the errors in the experimental data.
The neural network learns non-trivial self-similarity, where the normalized parameters are found to be
\begin{align}
    (p_1,p_2) &= (0.44\pm 0.24, -0.46\pm 0.14)~~(\lambda =0.5) \,,
\label{eq:ex_lambda_one_half} \\
    (p_1,p_2) &= (0.43\pm 0.10, -0.463\pm 0.034)~~(\lambda =1) \,,
\end{align}
while the expectations by the theoretical model are $(p_1,p_2)=(0,-0.5)$.
Here the reference values are obtained in a single training with all experimental data and the statistical uncertainties are estimated by the bootstrap method with $N_{\rm bs}=50$.
The statistical uncertainties of the $\lambda=1$ case are smaller than those of the $\lambda=0.5$ case.
This may be due to the fact that as the magnitude of the regularization term increases, the parameter values become more localized.
Therefore, the regularization term should not be too large to estimate the error.
In Eq.~\eqref{eq:ex_lambda_one_half}, the theoretical value \eqref{eq:true-value} is included in the $2\sigma$ range.
Compared to the case of training on synthetic data, the results 
show greater variability in each trial and deviate from true value. The estimation of experimental data depends on the quality and quantity of data. 
Experimental data generally include unexpected errors driven by uncontrolled factors.
The data set includes the data points of different surface treatment\footnote{The surface of data in Ref.~\cite{Maruoka_2023} was dusted by chalks, which was found to be most effective treatment to eliminate the adhesive effect. However, the treatment of other data was coated with the grease WD-40, which also had the effect of eliminating the adhesive effect though it was not as effective as dusting with chalks. Their data collapse plotted using Eq.~\eqref{eq:phi-truth} was not clear. See the supplemental materials of Ref.~\cite{Maruoka_2023}.}. It is possible that such a different treatment may be reflected in the deviation.
The present results may change depending on the architecture of the neural network. 

Nevertheless, the estimation of exponents using the experimental data supplies a nontrivial insight into the self-similar structure of the problem.

\begin{figure*}[t]
  \begin{minipage}[b]{0.45\linewidth}
    \centering
    \includegraphics[width=7cm]{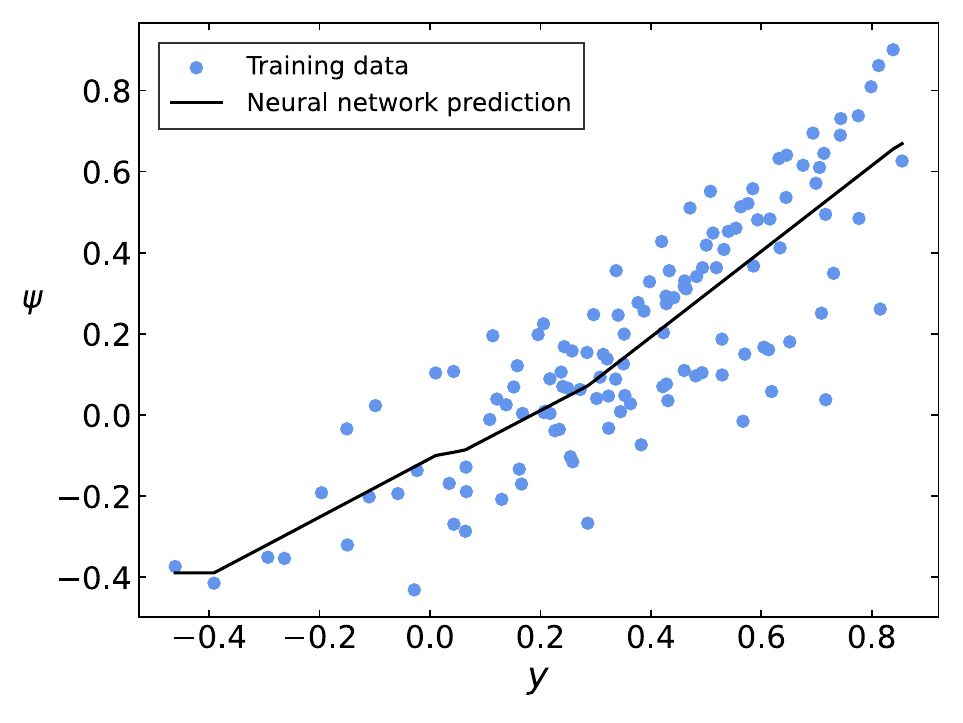}
    \subcaption{$\lambda=0.5$}
  \end{minipage}
  \hspace{5mm}
  \begin{minipage}[b]{0.45\linewidth}
    \centering
    \includegraphics[width=7cm]{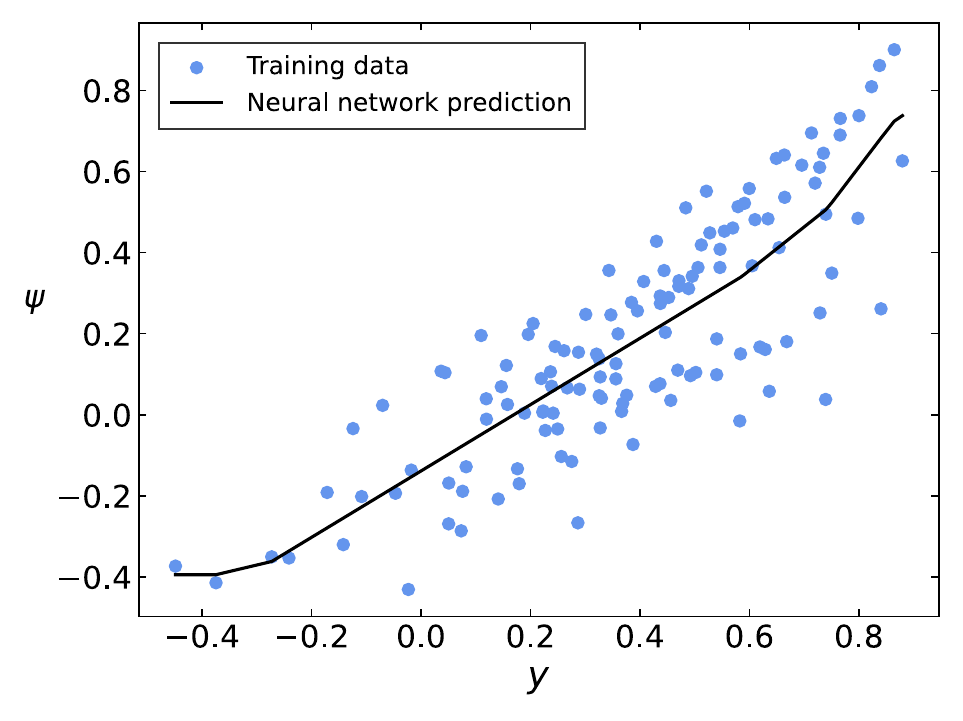}
    \subcaption{$\lambda=1$}
  \end{minipage}
  \caption{The results of training on the experimental data with regularization terms.}
  \label{fig:12}
\end{figure*}

\section{ Extension to the cases where the scaling function has multiple arguments
}\label{sec:example-two-combinations}

The present method can be extended to the situations where scaling functions have multiple arguments.
In this section, we show that our method works for the cases where the scaling function has two arguments by using synthetic data.
We here explore two examples. One is a hypothetical model, and the other is based on a physical model.

\subsection{ Synthetic system }\label{sec:two-var-ex1}

Let $\{\Psi,\pi_0,\pi_1,\pi_2,\pi_3\}$ be the set of dimensionless parameters in the system, 
and suppose that $\Psi$ is written in terms of other parameters as
\begin{equation}
    \Psi = \Phi(Z_1,Z_2),
    \label{two_var}
\end{equation}
where the two arguments are 
\begin{equation}
Z_1 = \pi_0 \pi_1 \pi_2,\quad Z_2 =\pi_0 \pi_3^{\frac{1}{2}} .
\end{equation}
and the function $\Phi$ is taken to be
\begin{equation}
 \Phi(Z_1,Z_2) \coloneqq (1+Z_1+Z_1^2) e^{Z_2}.
  \label{eq:two_var_func}
\end{equation}
We randomly generated $\pi_0,\pi_1,\pi_2,\pi_3$ from [0,1], 
and computed $\Psi$ according to Eq.~\eqref{two_var}.
We add noises to the generated data in the same
way as in Sec.~\ref{sec:training-synthetic}.

Suppose that we do not know the combination
of the argument of the RHS of Eq.~\eqref{two_var}.
We would like to determine this combination
based on the synthetic data.
Let us parametrize the argument as
\begin{equation}
    \Psi=\Phi (\pi_0^{p^{(0)}_0}\pi_1^{p^{(0)}_1}\pi_2^{p^{(0)}_2}\pi_3^{p^{(0)}_3} ,\pi_0^{p^{(1)}_0}\pi_1^{p^{(1)}_1}\pi_2^{p^{(1)}_2}\pi_3^{p^{(1)}_3} ).
    \label{solvable}
\end{equation}
Introducing the log of the parameters,
\begin{equation}
\begin{split}
 &\psi  \coloneqq \ln \Psi, \quad 
 x_0\coloneqq \ln \pi_0,\quad 
 x_1\coloneqq \ln \pi_1, 
 \\
 &x_2\coloneqq \ln \pi_2, \quad 
 x_3\coloneqq \ln \pi_3,    
\end{split}
\end{equation}
Eq.~\eqref{solvable} can be written as
\begin{equation}
\begin{split}
 \psi &=\ln \Phi \left( \exp \left[\sum_{i=0}^3 p^{(0)}_i x_i\right], 
    \exp \left[\sum_{i=0}^3 p^{(1)}_i x_i\right] \right)
\\    
& \eqqcolon \phi \left( \sum_{i=0}^{3}p^{(0)}_i x_i, \sum_{i=0}^{3}p^{(1)}_i x_i \right),    
\end{split}
\end{equation}
where we have defined
\begin{equation}
    \phi(x,y) \coloneqq \ln(1+e^x+e^{2x}) + e^y.
\end{equation}

The true values of the vectors $\bm p^{(0)},\bm p^{(1)}$ are given by
\begin{equation}\label{eq:art_two_var_ans}
    {\bm p_{\rm true}^{(0)} }= 
    \begin{pmatrix}
    1 & 1 & 1 & 0    
    \end{pmatrix}^\top, 
    \quad 
    {\bm p_{\rm true}^{(1)} }=
    \begin{pmatrix}
    1 & 0 & 0 & \frac{1}{2}
    \end{pmatrix}^\top.
\end{equation}

We estimated their true values from the data by using
neural networks as explained in Sec.~\ref{sec:method}.
The structure of the neural network we used is
4-2-100-100-100-100-1.
We generated 1,000 data points.
The number of epochs was set 20,000 and 
we used the mean squared error \eqref{eq:MSE_loss} as the loss function.

To make a meaningful comparison, we have to fix the ambiguities of exponents discussed in Sec.~\ref{sec:method}.
The raw values of the exponents obtained from the trained neural network are in general linear combinations of $\bm p^{(0)}$ and $\bm p^{(1)}$.
We fix the ambiguities as follows. 
Let us denote the raw  parameters by $\bm w^{(0,0)}$ and $\bm w^{(0,1)}$\footnote{
Even if the training is perfect, these parameters are linear combinations of $\bm p^{(0)}$ and $\bm p^{(1)}$, 
$\bm w^{(0,0)} = c_0 \bm p^{(0)} + c_1\bm p^{(1)}, \bm w^{(0,1)}= c_2 \bm p^{(0)} + c_3 \bm p^{(1)}$.}.
We first add a vector $\propto \bm w^{(0,1)}$ to $\bm w^{(0,0)}$
in such a way that the third components of the total vector is zero. 
Then, we further scale the obtained vector so that its zeroth component becomes 1, and the resulting vector is denoted by $\bm p^{(0)}$.
Similarly, we add a vector $\propto \bm p^{(0)}$ to $\bm w^{(0,1)}$
so that the first component of the total vector is zero,
and then we scale it to make the zeroth component to unity. 
The operation described above is summarized as follows: 
\begin{align}
 &\bm a = \bm w^{(0,0)} - \frac{w^{(0,0)}_3}{w^{(0,1)}_3}\bm w^{(0,1)}, \ \ \bm p^{(0)} = \frac{1}{a_0} \bm a \label{op:1} \\
 &\bm b = \bm w^{(0,1)} - \frac{w^{(0,1)}_1}{p^{(0)}_1} \bm p^{(0)}, \ \ \bm p^{(1)} = \frac{1}{b_0} \bm b \label{op:2}.
\end{align}

In Table~\ref{TABLE:II} and Fig.~\ref{fig:p_and_q_plot_two}, 
we show the results of estimated exponents 
for different values of noise strength $r$, 
alongside their the ground-truth values. 
The uncertainties associated with these estimates were determined using the bootstrap method, employing a resampling count $N_{\rm bs} =50$.
We find that the estimated values are consistent with true values 
across all examined noise strengths ($r =0, 0.05, 0.1$). 
We also observe that the uncertainties are increasing function of the noise strength.
In Fig.~\ref{fig:two_var_loss}, we depict the behavior of the loss function post-training as a function of the noise strength $r$.
Similarly to the case when a single combination is estimated, the quadratic rise of the loss in $r$ signifies the absence of overfitting.

In Fig.~\ref{fig:two var artificial sym result}, 
we plot $\psi$ of the data (black points) and the neural network (blue surface) as a function of the two combinations of the parameters 
estimated in the training. 
This visualization indicates a successful data collapse,
affirming the accuracy of the parameter estimation.

\begin{table*}[t]
  \centering
  \begin{minipage}{1.0\textwidth}
      \begin{tabular}{ccccc}
        $r$ & $p^{(0)}_1$ & $p^{(0)}_2$ & $p^{(1)}_2$ & $p^{(1)}_3$\vspace{.5mm}\\ \hline
        $0$ & $0.980  \pm 0.017$ & $0.982\pm 0.023$ & $-0.0001 \pm 0.0029$ & $0.5055\pm 0.0044$ \vspace{.5mm}\\ 
        \vspace{.5mm}
        $0.05$ & $0.937 \pm 0.067$ & $0.953\pm 0.053$ & $0.0107 \pm 0.0067$ & $0.528\pm 0.027$\vspace{.5mm}\\ 
        \vspace{.5mm}
        $0.1$ & $0.97  \pm 0.11$ & $1.01 \pm 0.12$ & $-0.014 \pm 0.030$ & $0.565 \pm 0.092$\vspace{.5mm}\\ 
        \hline
        \hline
        true & $1$ & $1$ & $0$ & $0.5$\vspace{.5mm}
      \end{tabular}
  \end{minipage}
  \caption{Normalized parameters $\bm p^{(0)},\bm p^{(1)}$ 
  for different values of noise strength $r$. }
  \label{TABLE:II}
\end{table*}

\begin{figure}
\begin{minipage}{\textwidth}
    \includegraphics[clip,trim=0cm 0cm 0cm 0cm,width=\columnwidth]{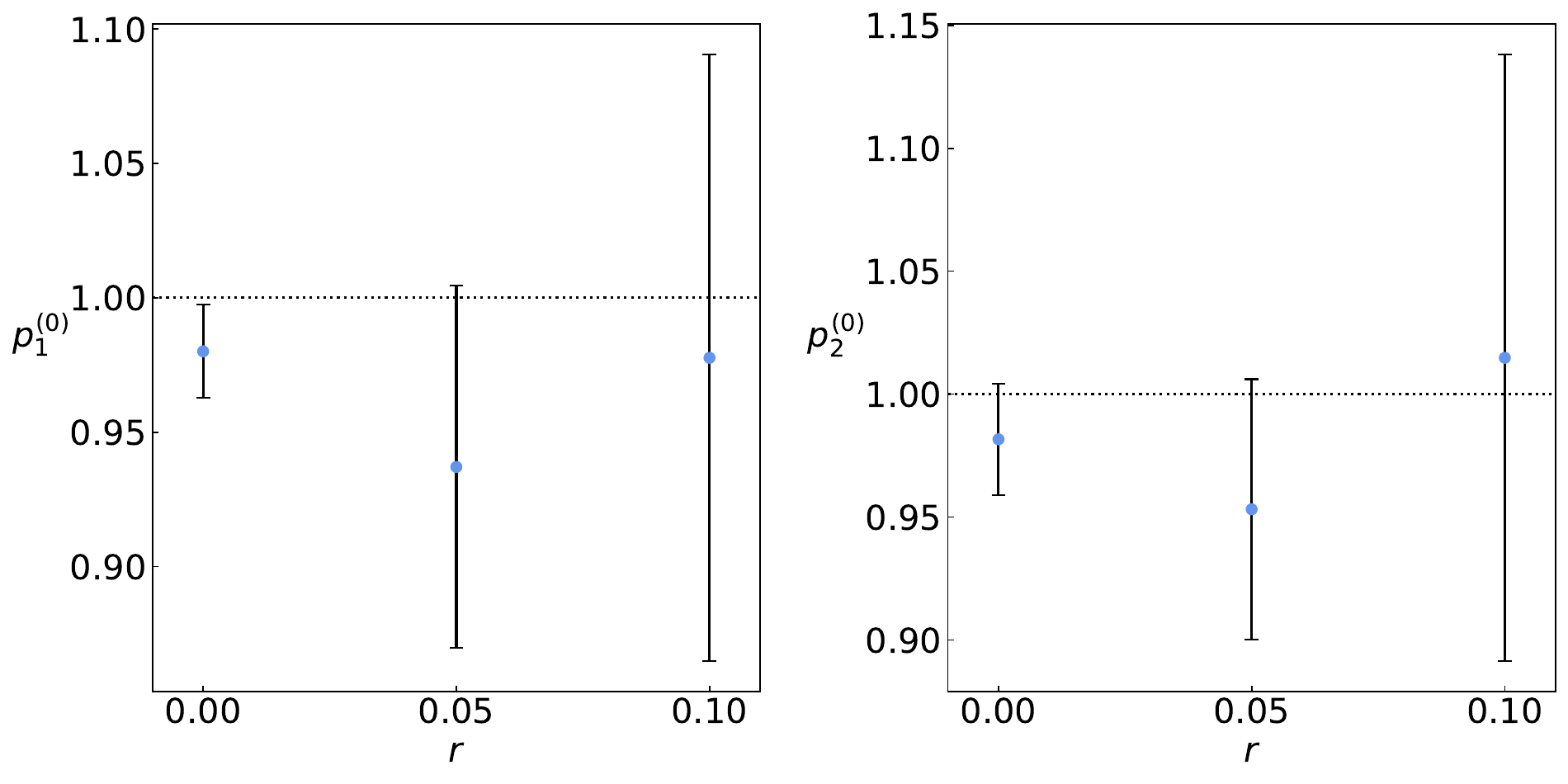}
    \centering
    \subcaption{The noise strength $r$ dependency of parameters $\bm p^{(0)}$}
    \label{fig:p_plot_two_var}
\end{minipage}\\
\begin{minipage}{\textwidth}
    \includegraphics[clip,trim=0cm 0cm 0cm -1cm,width=\columnwidth]{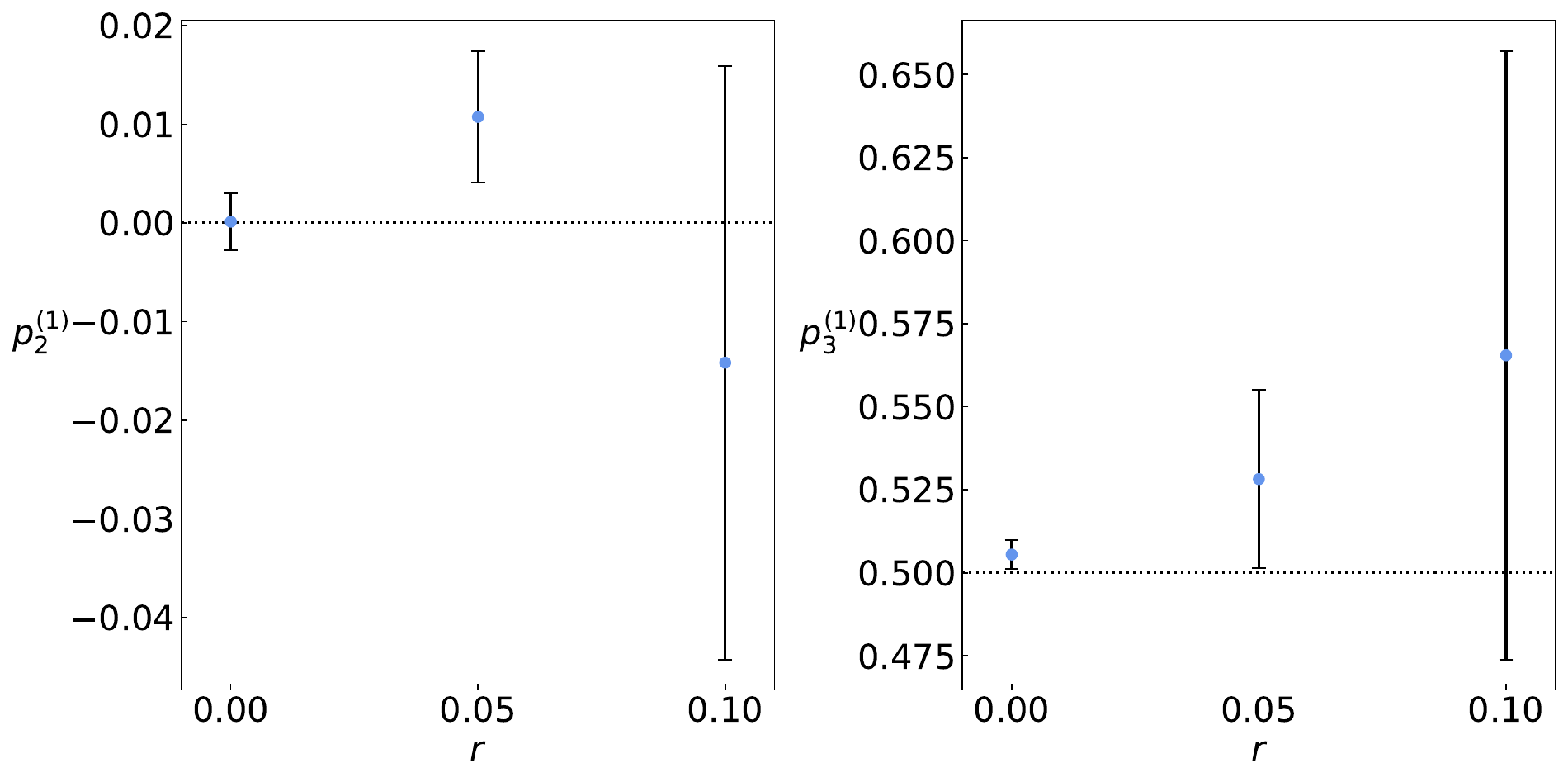}
    \centering
    \subcaption{The noise strength $r$ dependency of parameters $\bm p^{(1)}$}
    \label{fig:q_plot_two_var}
\end{minipage}
    \centering
    \caption{The noise strength $r$ dependency for each parameters}
    \label{fig:p_and_q_plot_two}
\end{figure}

\begin{figure}
    \centering
    \includegraphics[clip,trim=0cm 14.8cm 20cm 0cm,width=\columnwidth]{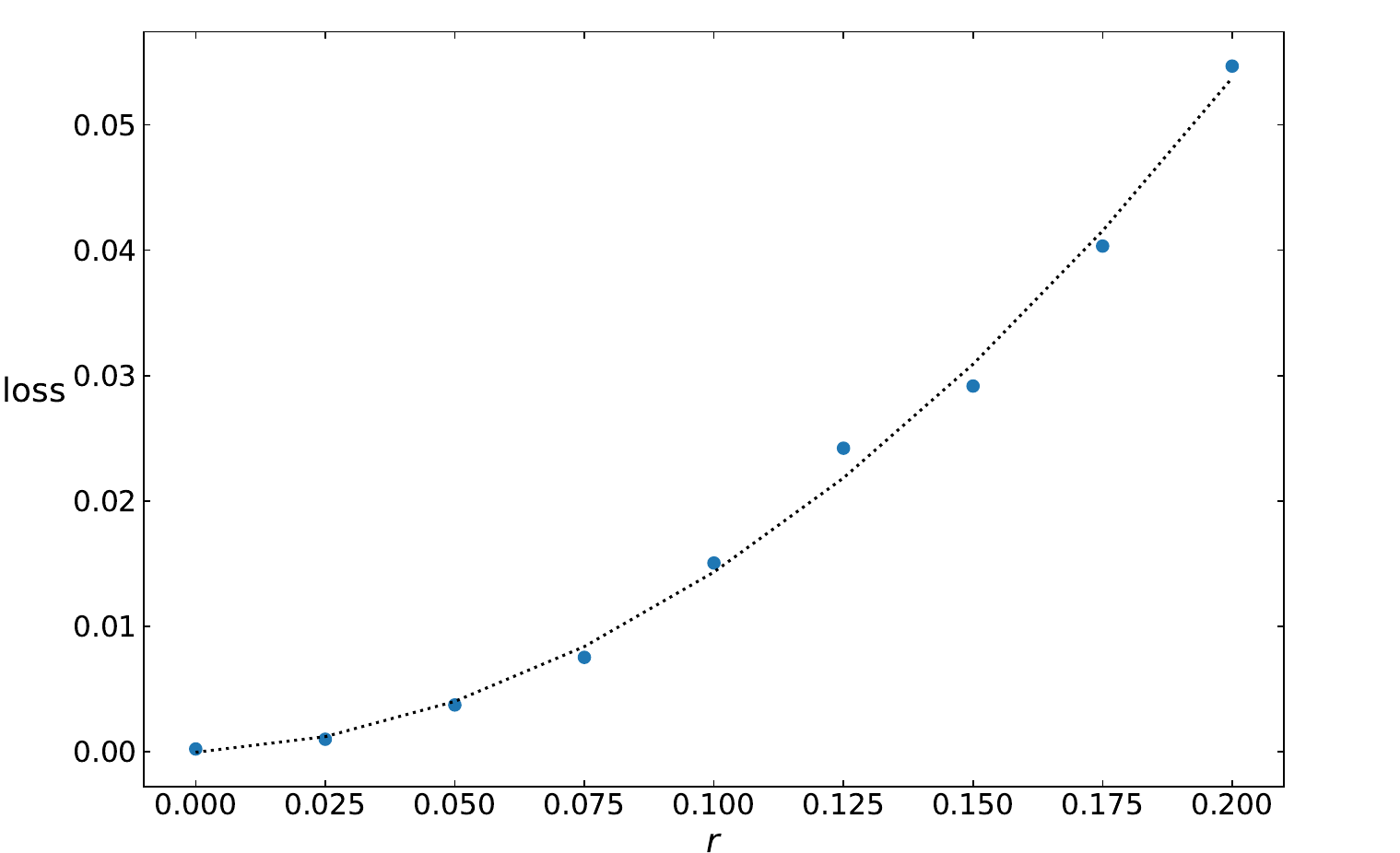}
    \caption{Loss function as a function the noise strength~$r$. The blue points are the values of the loss and the black dotted line represents a quadratic fit (${\rm loss} = 1.25r^2 + 0.0186r - 0.0000542$).}
    \label{fig:two_var_loss}
\end{figure}

\begin{figure*}[t]
\begin{tabular}{cc}
\begin{minipage}{0.325\hsize}
\centering
\includegraphics[clip,trim=4cm 5cm 4.3cm 7cm,width=\columnwidth]{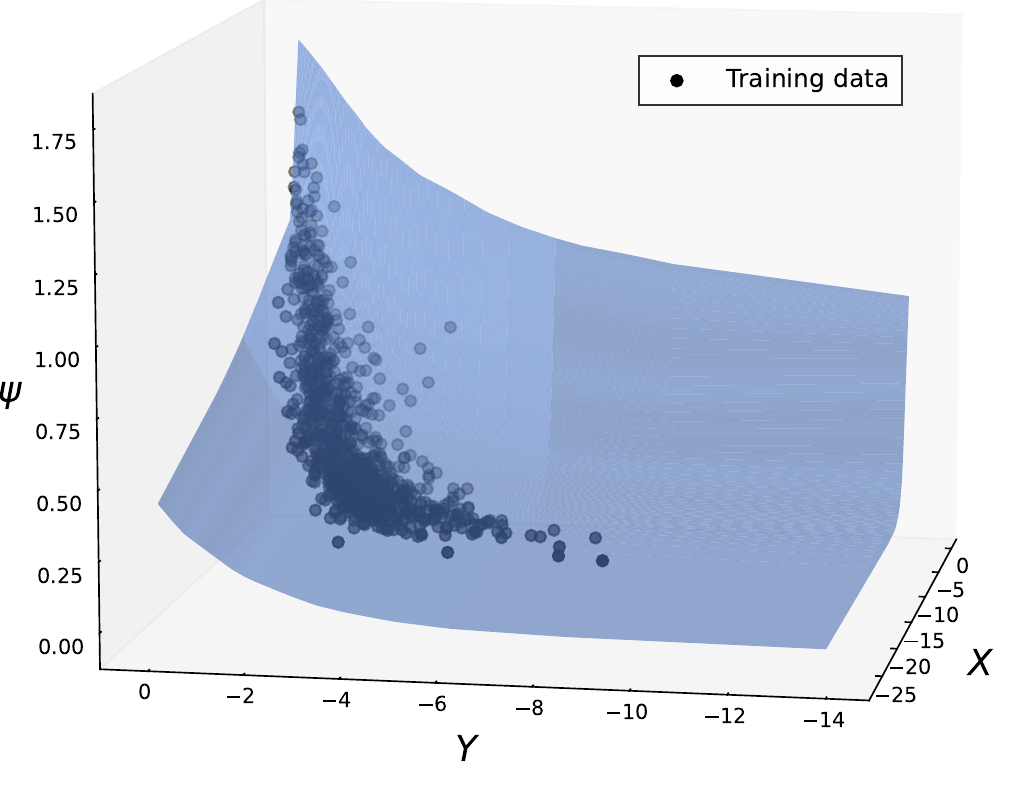}
\subcaption{$r=0$}
\label{fig:0 errror}
\end{minipage}
\begin{minipage}{0.325\hsize}
\includegraphics[clip,trim=4cm 5cm 4.3cm 7cm,width=\columnwidth]{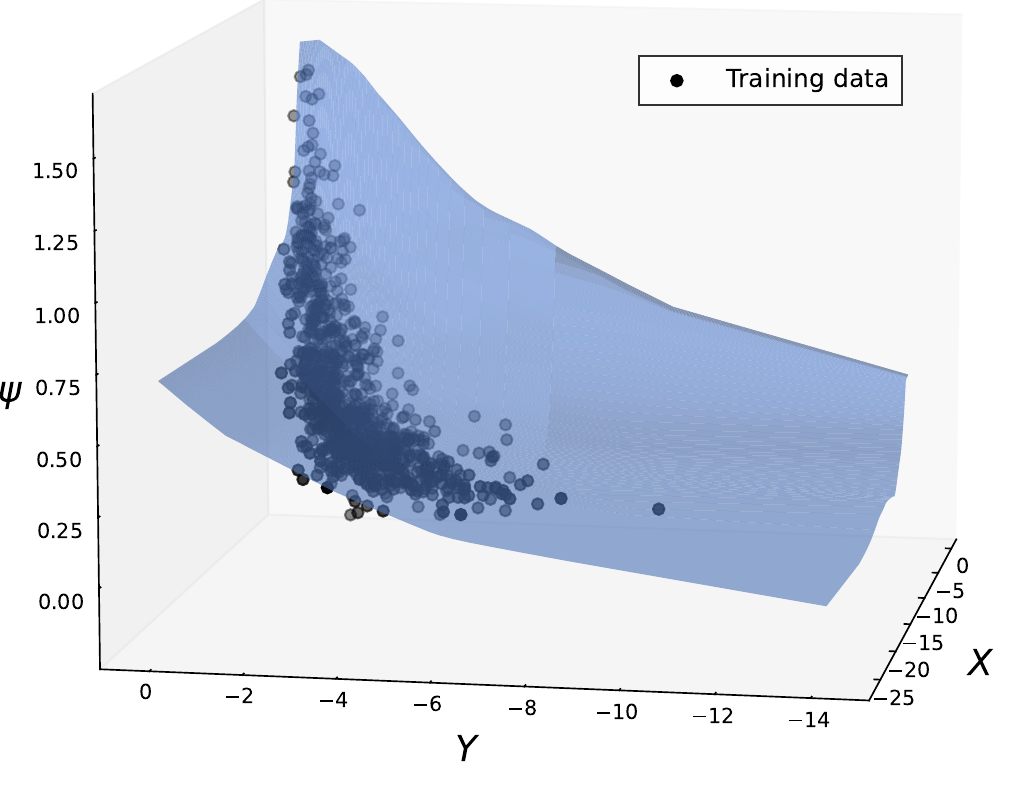}
\subcaption{$r=0.05$}
\label{fig:5 error}
\end{minipage}
\begin{minipage}{0.325\hsize}
\includegraphics[clip,trim=4cm 5cm 4.3cm 7cm,width=\columnwidth]{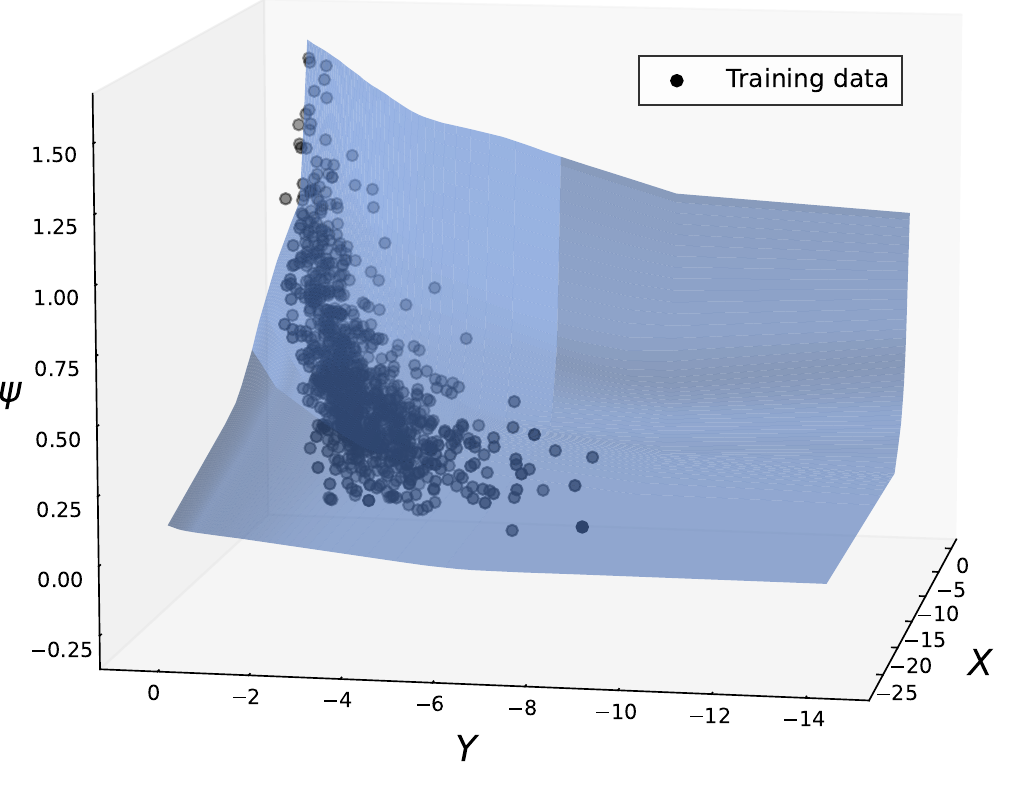}
\subcaption{$r=0.1$}
\label{fig:10 error}
\end{minipage}
\end{tabular}
\caption{The black points plot the training data $\psi^{(i)}$ and the blue surface represents neural network predictions $\psi_{\rm NN}^{(i)}$ as a function of $X^{(i)} = \bm p^{(0)} \cdot \bm x^{(i)}$ and $Y^{(i)} =\bm p^{(1)} \cdot \bm x^{(i)}$.}
\label{fig:two var artificial sym result}
\end{figure*}

\subsection{ Dynamical impact of a solid sphere onto a Zener viscoelastic board }\label{sec:example-two-combinations_Zener}

\subsubsection{Setting}

In the following, we examine our methods using the synthetic data 
corresponding to the dynamical impact of a solid sphere onto a Zener viscoelastic board.
The experimental configuration is the same as Sec.~\ref{sec:example-1}, but the Zener element is used for the foundation of viscoelasticity instead of the Maxwell element\footnote{The dynamical impacts of Zener viscoelastic materials were considered in the supplemental material of Ref.~\cite{Maruoka_2023}}.
The Zener model can be depicted as a dashpot $\mu$ and an elastic spring $E$ connected in series, which is in turn connected in parallel with another elastic spring $E_K$.

The physical parameters of the system are as follows,
\begin{equation}
\bm z = \left(\delta_m, h, R, \rho, \mu, E, E_K, v_i \right)^\top.
\end{equation}
We introduce the dimensionless parameters as 
\begin{equation}
\begin{split}
\Pi &\coloneq \frac{\delta_m}{R}, \quad 
\kappa \coloneq \frac{h}{R}, \quad 
\theta \coloneq \frac{\mu}{E^{1/2} \rho^{1/2} R}, 
\\
\nu &\coloneq \frac{E}{E_K}, \quad 
\eta \coloneq \frac{\rho v_i^2}{E}.    
\end{split}
  \label{eq:Zener-pi}
\end{equation}
In terms of the dimensionless parameters, 
the physically significant relation in this system is written 
as~\cite{Maruoka_2023} 
\begin{equation}
 \Psi = \Phi(Z, \nu), 
  \label{eq:two_var_eq}
\end{equation}
where
\begin{equation}
  \Psi = \frac{\Pi^3}{\eta \kappa},
  \quad 
   Z = \frac{\Pi}{\theta \eta^{1/2}},
\end{equation}
and the function $\Phi(Z, \nu)$ is defined as
\begin{equation}
 \Phi (Z, \nu) \coloneq \frac{2Z}{\frac{Z}{\nu}  + 3( 1 - e^{-Z} )}.
\end{equation}

Suppose we perform the rescaling of units~\eqref{eq:unit-rescale}.
The dimension function $\phi$ acting on physical parameters $\bm z$ is given as
\begin{equation}
\phi = \left( L, L, L, \frac{M}{L^3}, \frac{M}{LT}, \frac{M}{LT^2}, \frac{M}{LT^2}, \frac{L}{T} \right)^\top.
\end{equation}
An element $\varphi \in G^{(\rm s)}_{\rm second}$ acting on the dimensionless parameters $\bm \pi =( \Pi, \kappa, \eta, \theta, \nu )^{\top}$ is
\begin{equation}
\varphi = \left( A, A^3 B^{-1}, B, A B^{-1/2}, 1 \right)^\top.
\end{equation}
One can easily check that the transformation $\varphi$ leaves Eq.~\eqref{eq:Zener-pi} invariant. The CGM number $\Psi$ is fixed similarity parameter as in the previous section.

\subsubsection{Training}

We examined our neural network method using synthetically generated data according to Eq.~\eqref{eq:two_var_eq}.
We prepared two training data sets. 
In the first data set, 
$\Pi^{(i)}, \eta^{(i)}, \theta^{(i)}, \nu^{(i)}$ are prepared in a similar manner to Sec.~\ref{sec:two-var-ex1} and $\Psi^{(i)}$ is generated by using those data set $\Pi^{(i)}, \eta^{(i)}, \theta^{(i)}, \nu^{(i)}$ through Eq.~\eqref{eq:two_var_eq}.
The second set is generated in such a way that it is closer to actual experimental data.
In actual experiments, 
there can be certain parameters difficult to change 
for technical reasons. 
For example, the impact speed can be readily altered by modifying the drop height of the impactor,
while adjusting the viscous coefficient and the elastic modulus is challenging as these parameters are dependent on the selected viscoelastic board, and the options for such boards are limited.
The second set of training data is prepared 
taking into account this point. 

We introduce the following parameters,
\begin{equation}
\begin{split}
&\psi \coloneq \ln \Psi, \quad 
x_0 = \ln \Pi, \quad
x_1 = \ln \eta, \quad \\
&x_2 = \ln \theta,\quad
x_3 = \ln \nu.
\end{split}
\end{equation}
$\kappa$ can be omitted from the input during training without loss of generality, as per Sec.~\ref{sec:method-nn} .
We parametrize Eq.~\eqref{eq:two_var_eq} as
\begin{equation}
\begin{split}
 \psi &=\ln \Phi \left( \exp \left[\sum_{i=0}^3 p^{(0)}_i x_i\right], 
    \exp \left[\sum_{i=0}^3 p^{(1)}_i x_i\right] \right)
\\    
& \eqqcolon \phi \left( \sum_{i=0}^{3}p^{(0)}_i x_i, \sum_{i=0}^{3}p^{(1)}_i x_i \right),    
\end{split}
\label{eq:psi-phi-zener}
\end{equation}
where we have defined
\begin{equation}    
\phi(a , b) \coloneqq \ln \Phi (e^a, e^b). 
\end{equation}
Hence, the true values of the vectors $\bm p^{(0)}, \bm p^{(1)}$ are given by
\begin{equation}\label{two_var_ans}
    {\bm p_{\rm true}^{(0)} }= 
    \begin{pmatrix}
    1 & -0.5 & -1 & 0    
    \end{pmatrix}^\top, 
    \quad 
    {\bm p_{\rm true}^{(1)} }=
    \begin{pmatrix}
    0 & 0 & 0 & 1
    \end{pmatrix}^\top.
\end{equation}

We estimate these true values by optimizing a neural network as explained in Sec.~\ref{sec:method}.
Namely, we introduce a function corresponding to Eq.~\eqref{eq:psi-phi-zener} using a neural network. 
The structure of the neural network we used is 4-2-100-100-100-1. 
The number of prepared data points is $N_{\rm data} = 250$.
The number of epochs is 20,000 and we used the mean squared error as the loss function.

After optimizing the neural network, we fix the ambiguity 
of pow-law exponents in the following way.
Let $\bm w^{(0,0)}, \bm w^{(0,1)}$ be the 
raw parameter values of the first layer of the optimized neural network.
We compute $\bm p^{(0)}$ and $\bm p^{(1)}$, which are to be compared with the ground-true values, as 
\begin{equation}
\begin{split}
& \bm  a = \bm w^{(0,0)} - \frac{w^{(0,0)}_3}{w^{(0,1)}_3}\bm w^{(0,1)}, \quad
\bm p^{(0)} = \frac{1}{a_0} \bm a, \\
& \bm b = \bm w^{(0,1)} - \frac{w^{(0,1)}_0}{p^{(0)}_0} \bm p^{(0)},\quad
\bm p^{(1)} = \frac{1}{b_3} \bm b.
\end{split}
\end{equation}

Table~\ref{table:data two_var} presents the calculated estimations of the power-law exponents for both data sets.
These estimated values are within two standard deviations, resulting in a successful data collapse for both the random data and the data close to the experiment. 
Figures~\ref{fig:synthetic data two var surface} and \ref{fig:experimental data two var surface} visually illustrate the outcome of the training with the random data and the data under the experimental condition, where the successful data collapse can be clearly seen.

\begin{table*}[t]
  \centering
  \begin{tabular}{ccccc}
    kind of data & $p^{(0)}_1$ & $p^{(0)}_2$ & $p^{(1)}_1$ & $p^{(1)}_2$ \vspace{.5mm}\\ \hline
    synthetic & $-0.499\pm 0.011 $ & $-1.015\pm 0.021 $ & $ 0.0003 \pm 0.0061 $ & $ 0.0059 \pm 0.0067 $\vspace{.5mm}\\ 
    closer to the experimental condition & $  -0.519 \pm 0.018 $ & $ -1.089 \pm 0.086 $ &$ 0.047 \pm 0.053 $ & $  0.228 \pm 0.25 $\vspace{.5mm}\\ 
    \hline\hline
    true &$-0.5$ & $-1$ & $0$ & $0$ 
  \end{tabular}
  \caption{ Power-law exponents obtained by trained neural networks. }
  \label{table:data two_var}
\end{table*}

\begin{figure*}[t]
\begin{tabular}{cc}
\begin{minipage}{0.45\hsize}
\centering
\includegraphics[clip,trim=4cm 5cm 4.3cm 7cm,width=\columnwidth]{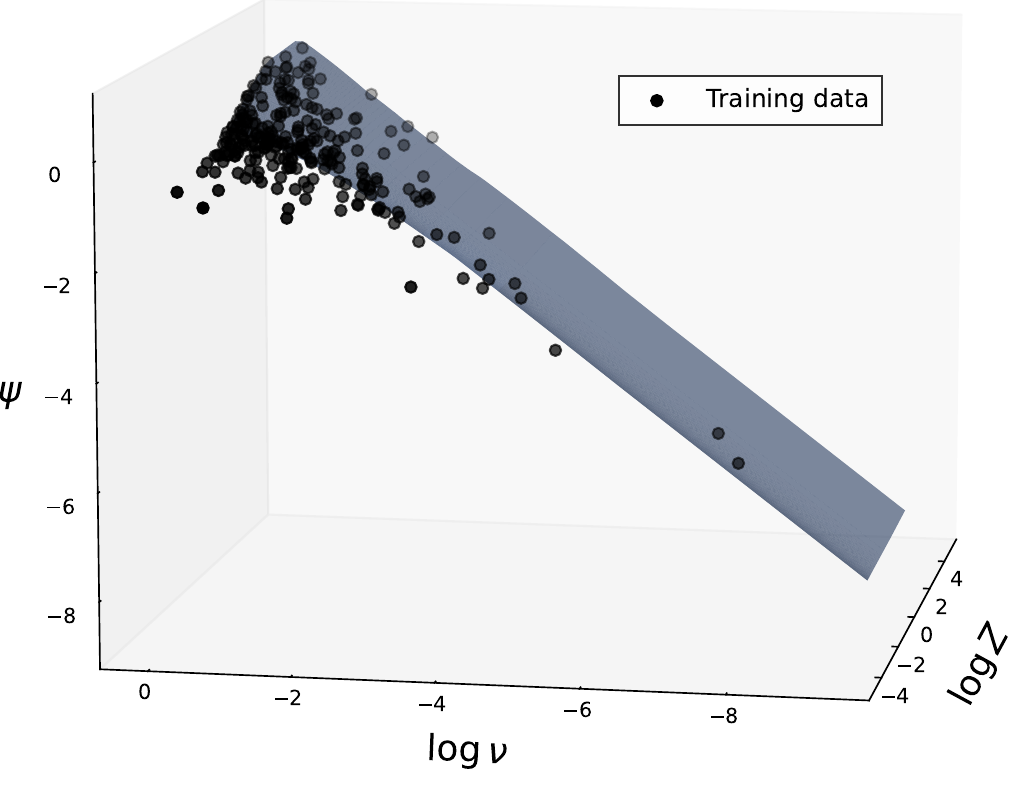}
\subcaption{Training with the first data set.}
\label{fig:synthetic data two var surface}
\end{minipage}
\begin{minipage}{0.45\hsize}
\includegraphics[clip,trim=4cm 5cm 4.3cm 7cm,width=\columnwidth]{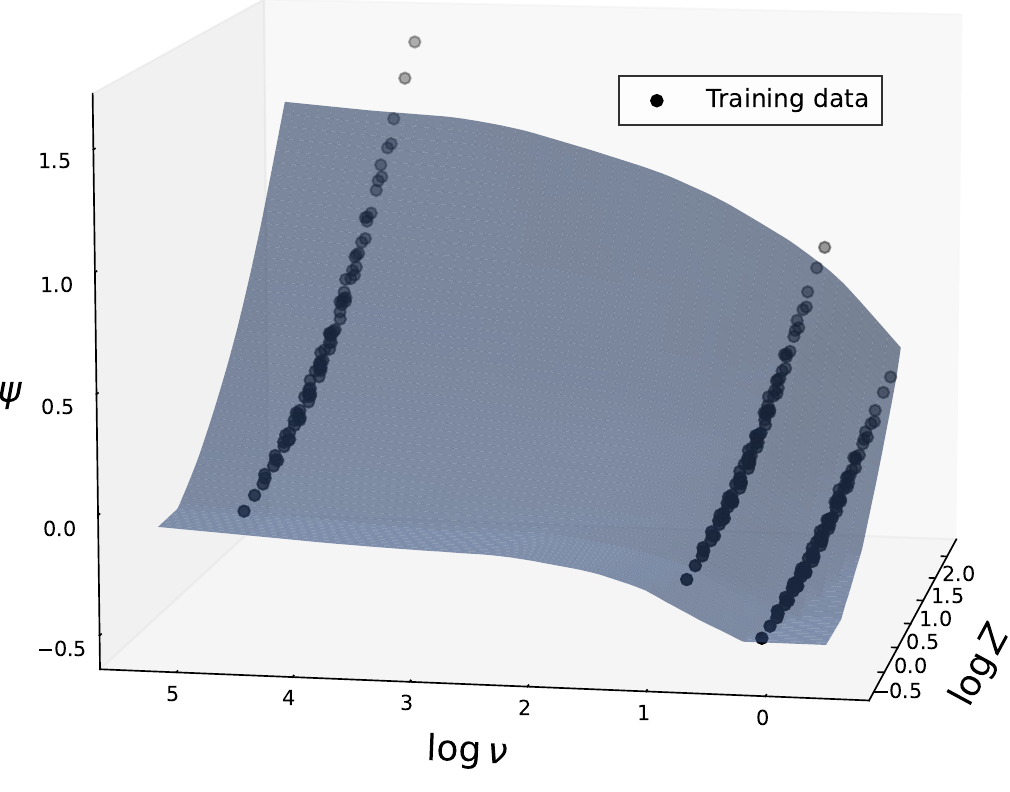}
\subcaption{Training with the second data set, which is prepared to respect experimental constraints.}
\label{fig:experimental data two var surface}
\end{minipage}
\end{tabular}
\caption{The result of training on with 
the first (left) and the second (right) data sets: The black dots represent the training data and the blue surface are the output of the neural network.}
\label{fig:Zener surface}
\end{figure*}

\section{ Conclusion and discussion }\label{sec:conclusion}

In this paper, we have developed a novel approach for finding self-similarity and the associated data collapse
in a data-driven way.
The information of self-similarity is encoded in 
the structures of monomials in the arguments 
of functions characterizing physical relations.
In the proposed method, we implement such dependencies in a parametrized way using neural networks. 
By performing the training of the neural networks
using a given data set, 
one can extract the exponents characterizing the self-similar solutions of the system. 

We have tested the method using synthetic and experimental data.\footnote{
\changed{
Additionally, we tested our method on data without self-similarity of the second kind.
We revealed several indicators suggesting the absence of self-similarity in the data.
See Appendix~\ref{sec:appendix_b} for details.
}
}
In the case where the scaling function is assumed to have single argument, we have shown that our method works effectively using numerically generated data.
We also applied our method to actual experimental data and gained non-trivial insights into self-similarity structure of the system.
The method can be naturally generalized 
to the situations where the scaling function has multiple arguments. 
We confirmed that our method also works well in the case of two arguments using numerically generated data.

In summarizing the outcomes of our test training with both synthetic and experimental data, we can assert that our approach has effectively achieved its aim of instituting the traditional method for investigating the similarity of the second kind.
Traditionally, the identification of data collapses 
has been pursued through a process of trial and error, aside from dimensional analysis.
However, our approach, which utilizes neural networks, paves the way for a new direction.

Here are some tips for effective estimations;
\begin{itemize}
    \item[T1] Ensure the quality and quantity of the dataset.
    \item[T2] The data should cover a wide range of parameters.
    \item[T3] Estimation of similarity parameters should be performed with the minimum numbers of physical parameters.
\end{itemize}

The accuracy of the estimation greatly depends on 
both the quality and the quantity of data. 
A dataset that is too small can lead to overfitting.
The minimum number of data points required for a reliable estimation can depend on various properties such as the setup being considered and the errors in the data. 
From the results of our training on synthetic data, it is clear that our method will work once sufficient data is given.
In the training with actual experimental data discussed in Sec.~\ref{sec:example-1}, the number of data points was $N_{\rm data} = 127$.
We found non-trivial self-similar structure, which seems to be consistent with the previous study within the range of possible uncertainties of the training results.
The quality of data, which is crucial for successful training, is influenced by two main factors.
First is the extent to which the dataset is free from errors. We investigated how noise affects our estimations by introducing synthetic noise into our data. Our analyses show that while noise enlarges uncertainties in the estimations, they remain robust even in the presence of significant noise (up to a relative noise strength $r = 0.25$). 
The other factor is the range of parameters in the data set. 
To facilitate accurate estimations, 
the data should be collected from a wide range of physical parameters.
The range of possible physical parameters is often limited 
especially in actual experiments.
In the example of dynamical impacts of viscoelastic board, 
while the impact velocity can be easily differentiated by changing the height to drop the impactor, 
changing the viscous coefficient and the elastic modulus 
is difficult because these parameters are derived from 
a chosen viscoelastic board and there are not many available choices of boards. 
Such a restriction induces particular correlations in the data,
and we may have to prevent neural networks to find
such correlations by introducing appropriate regularization terms. 

Selecting the appropriate physical parameters is also crucial for accurate estimations.
Every problem has a basic set of parameters that must be considered. 
The recipe of Barenblatt is effective for the screening process\footnote{See Ref.~\cite{barenblatt_1996} (pp.159-160) and Ref.~\cite{barenblatt_2003} (pp.91-93).}.
When a parameter is not involved in the problem at all, the dimensionless function $\Phi$ converges to a finite limit and their intermediate asymptotics is obtained. 
Achieving precise estimations with our neural network approach is facilitated by reducing the number of similarity parameters we need to consider. 

\changed{
In certain cases, however, the data itself may exhibit hysteresis or bivalued behavior~\cite{Hatano,Saitoh},
that can introduce instability into the estimation process. When such behavior is present, the accuracy of the estimation can be significantly improved by restricting the analysis to data points within the single-valued region, as discussed in 
Appendix~\ref{sec:appendix_a}.
}

Finally, let us summarize the limitations of our approach;
\begin{itemize}
\item[L1] At least one similarity parameter should be given a priori as a fixed parameter.
\item[L2] Number of similarity parameters of the second \changed{class} must be assumed.
\item[L3] The approach is only applicable to the similarity of the second kind in which the power-law exponents appearing in the combination of parameters are constants. 
\end{itemize}

L1 is not a critical weakness as it can be obtained naturally in the framework of the crossover of scaling laws. 
Even if the problem is complicated, simplification through idealization can reveal a scaling law within a certain range of scale. This simplification allows us to gradually approach the original problem from its idealized version.

Regarding the second limitation, our strategy requires an initial assumption about the number of combinations 
of the similarity parameters of the second \changed{class}. 
This necessitates a process of trial and error. 
For instance, if assuming two invariants proves insufficient, we may need to explore the possibility of additional invariants. 
\changed{
We find that the estimation becomes visibly unsatisfactory when we have a false assumption about the number of invariants (see Appendix \ref{sec:appendix_b}).
}

The third limitation is more fundamental. Our method yields specific numerical values for power exponents, which may not suit problems where these exponents are nontrivial functions of dimensionless parameters, which corresponds to self-similar solutions of the second kind, Type B\footnote{E.g., the self-similar solution of the second kind for the porous medium equation considering the fissures in the rock in Chapter~3 of Ref.~\cite{barenblatt_2003}.
}. 
If it is possible to determine the parameter range where the power exponents are viewed as unchanging, the present method continues to be of value.
We are of the opinion that this situation can potentially be addressed by adjusting the exponents using supplementary neural networks. The practicality of this approach remains to be evaluated.

Finally, we discuss the broader significance of the present work. Self-similarity is a fundamental concept in physics. As emphasized in Barenblatt's recipe, self-similarity has to be ideally pursued in a data-driven way without introducing biases. Our research demonstrates that modern technologies using neural networks offer a novel method to achieve this. We believe that an approach using neural networks can be a powerful tool for uncovering physical laws across a variety of physical problem. \\

The source codes and the data for this article are openly available from Ref.~\cite{repository}.

\begin{acknowledgments}
The work of R.~W. was supported by Grant-in-Aid for JSPS Fellows No. JP22KJ1940.
The work of Y. ~H. was supported in part by JSPS KAKENHI Grant Numbers JP22H05111, JP22H05118, JP24K23186 and in part by JST, PRESTO Grant Number JPMJPR24K8. The work of H.~M was supported by Grants-in-Aid of MEXT, Japan for Scientific Research, Grant No. JP21H01006.
We wish to thank J. R{\o}nning, S. Lakhal, M. M. Bandi, 
M. Tezuka, T. Yamaguchi, K. Hashimoto, and J. Stout for helpful discussions.
\end{acknowledgments}

\appendix

\section{The estimation for the data including bivalued behavior }\label{sec:appendix_a}

\changed{
In some cases, hysteresis appears in the data collapse relations~\cite{Hatano,Saitoh}, resulting in bivalued behavior. It is therefore important to assess whether our method remains effective when applied to such bivalued data.
}

\changed{
We synthesized data including an artificial hysteresis possessing the same similarity relations 
as those in Eqs.~(\ref{eq:phi-truth}) and~(\ref{eq:dimless-vars_second}), and data collapse is achieved by plotting the similarity parameters, $\Psi = \frac{\Pi^3}{\kappa \eta}$ and $Z = \frac{\Pi}{\theta \eta^{1/2}}$. 
We generated data numerically by using 
relations $\Psi = \Phi_1 \left(Z\right)$ and $\Psi = \Phi_2 \left(Z\right)$, where
\begin{eqnarray}\label{eq:setup_bivalued}
\Phi_1\left( Z\right) &=& \frac{2}{3}\frac{Z}{1-\exp\left(-Z \right)},  \\
\Phi_2\left( Z\right) &=& \frac{2}{3}Z.
\end{eqnarray}
We synthesized 100 number of data points. We examined our approach using this synthetic data.
}

\changed{
Our results suggest that the method can still be effective for data exhibiting bivalued behavior (see Fig.~\ref{fig:hysteresis1}). 
The estimated values were $(p_1,p_2) = (0.1 \pm 4.2, -0.6 \pm 1.5)$ while true values are $(p_1,p_2) = (0.0 , -0.5 )$.
Although there was a slight deviation from the true values, the data collapse was generally effective. However, the high variance in the estimates indicated some instability in the results. 
By selectively analyzing only the data points within the single-valued region (where $Z > 2$), we improved the estimation, yielding $(p_1, p_2) = (0.00 \pm 0.18, -0.50 \pm 0.09)$ (Fig.~\ref{fig:hysteresis2}).
}

\begin{figure*}[t]
\begin{tabular}{cc}
\begin{minipage}{0.5\hsize}
\centering
\includegraphics[width=\columnwidth]{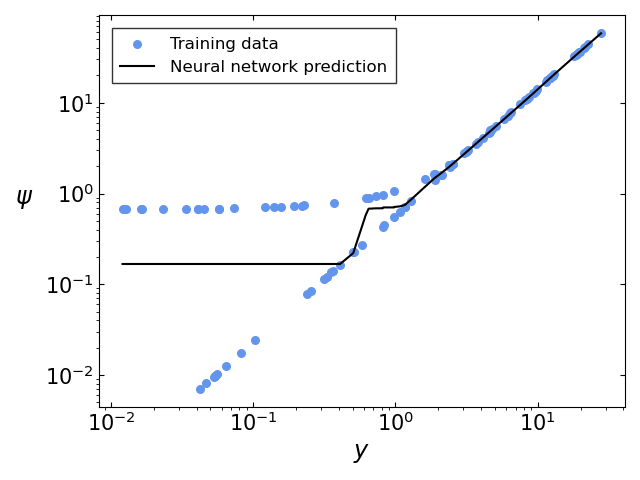}
\subcaption{
Estimation results for data with bivalued behavior using our approach. The estimated values were $(p_1,p_2) = (0.1 \pm 4.2, -0.6 \pm 1.5)$ while true values are $(p_1,p_2) = (0.0, -0.5)$.}
\label{fig:hysteresis1}
\end{minipage}
\begin{minipage}{0.5\hsize}
\includegraphics[width=\columnwidth]{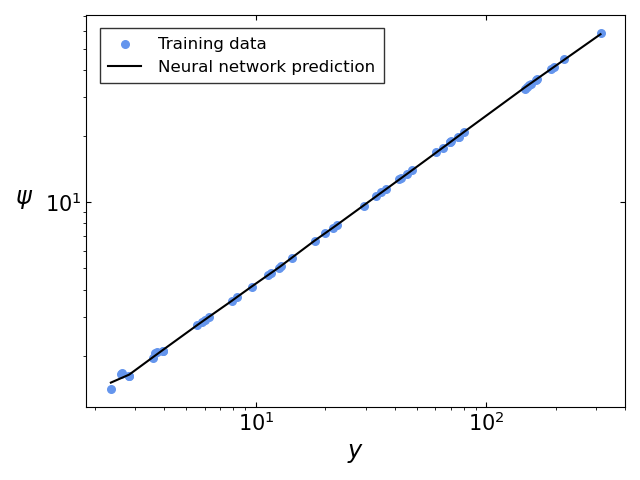}
\subcaption{Estimation results using data points in the single-valued region ($Z>2$). The estimated values were $(p_1,p_2) = (0.00 \pm 0.18, -0.50 \pm 0.09)$ while true values are $(p_1,p_2) = (0.0, -0.5)$.}
\label{fig:hysteresis2}
\end{minipage}
\end{tabular}
\caption{Estimation results for data involving the bivalued behavior where $\Psi = \Pi^3/ \kappa \eta$ and $Z_{NN} = \Pi \kappa^{p_1}\eta^{p2} \theta^{-1}$. The circle dots indicate the data plots of $\Psi$, the solid line indicate the data plots of $\Psi$ obtained from a neural network as $e^{\psi_{NN}}$. }
\label{fig:test_SS_reg_ex}
\end{figure*}

\changed{
These results indicate that selective screening of data points can enhance estimation accuracy, depending on the data characteristics. While using the entire dataset can reveal overall behavior, targeted selection may be necessary for more accurate results in cases involving hysteresis.
}

\section{
Testing the neural network method on data without self-similarity of the second kind}\label{sec:appendix_b}

\changed{
In this appendix, we test our method of detecting self-similarity via neural networks on data that lack self-similarity of the second kind.
}

\changed{
We generated data according to the following relation,
\begin{equation}\label{eq:verify_ss_sincos}
\Phi(x_0,x_1) = 1+ \sin^2 x_0 + \cos^2 x_1.
\end{equation}
As $x_0$ and $x_1$ are considered as dimensionless, $\Phi$ does not exhibit self-similarity of the second kind.
Under this condition, we applied our neural network method, assuming the presence of self-similarity of the second kind, and analyzed its behavior.
The dataset is generated by sampling $x_0$ and $x_1$ from a uniform distribution over $[0,3]$
and computing $\Phi$ via Eq.~\eqref{eq:verify_ss_sincos}. 
We created 1,000 training data points for this analysis.
The neural network architecture consists of layers with 2-1-10-10-1 nodes, respectively. For estimating statistical error via the bootstrap method, we used 100 resampled datasets, allowing replacement.
}

\changed{
Our findings reveal several indicators pointing to the absence of self-similarity in the data.
First, Fig.~\ref{fig:test_loss_hist} shows that the loss values do not decrease significantly during training, indicating poor model convergence.
Second, as can be seen in Fig.~\ref{fig:test_loss_vs_param},
even in cases where the loss values do decrease, 
the second parameter 
is either nearly zero or highly variable, 
meaning that one of the two parameters effectively vanishes.
This results indicates that the best fit is achieved when 
the network output depends on either either $x_0$ or $x_1$.
}

\changed{
Figure~\ref{fig:test_collapse} provides examples of the ``data collapse'' plot, displaying both the training data and the neural network predictions based on the estimated parameters. These plots show poor data collapse, further suggesting the absence of self-similarity in the system.
}

\changed{
In summary, our analysis demonstrates that when applied to a system lacking self-similarity, our neural network method fails to achieve stable learning outcomes. The behavior of the loss function, the variability of key parameters, and the poor data collapse in the plots collectively serve as strong indicators of the absence of self-similarity. 
These results suggest that the distributions of loss values and parameter estimates can act as effective criteria for identifying systems without self-similarity, providing users with diagnostic tools to evaluate the suitability of this method for various datasets.
}

\begin{figure*}[t]
\begin{tabular}{cc}
\begin{minipage}{0.5\hsize}
\includegraphics[clip,trim=0cm 0.5cm 0cm 0cm,width=\columnwidth]{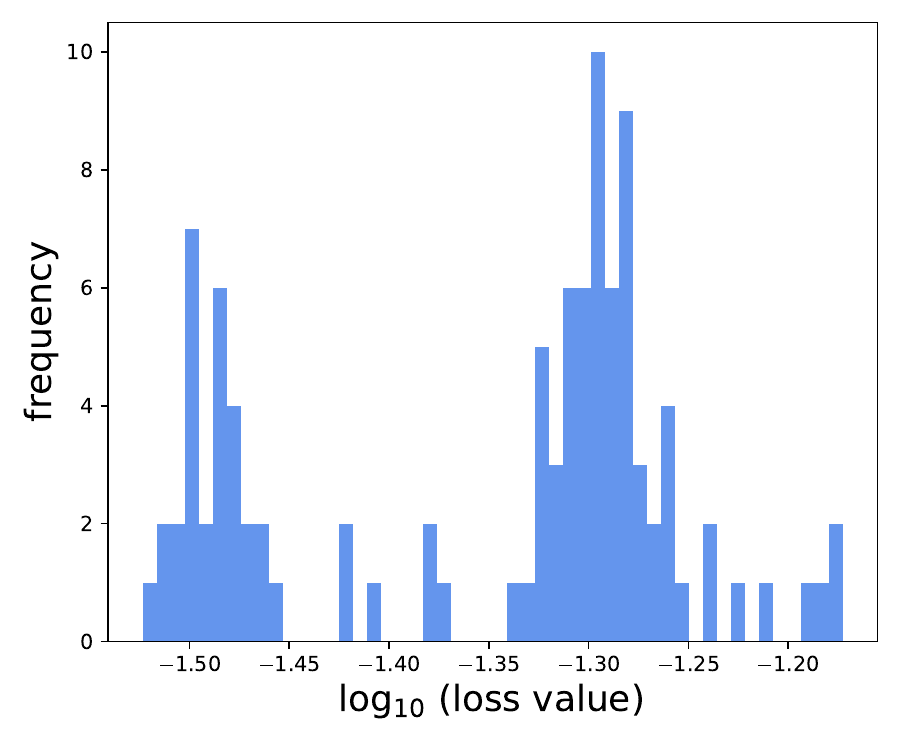}
\subcaption{Histogram of the logarithm of loss values.}
\label{fig:test_loss_hist}
\end{minipage}
\begin{minipage}{0.5\hsize}
\centering
\includegraphics[clip,trim=0cm 0cm 0cm 0cm,width=\columnwidth]{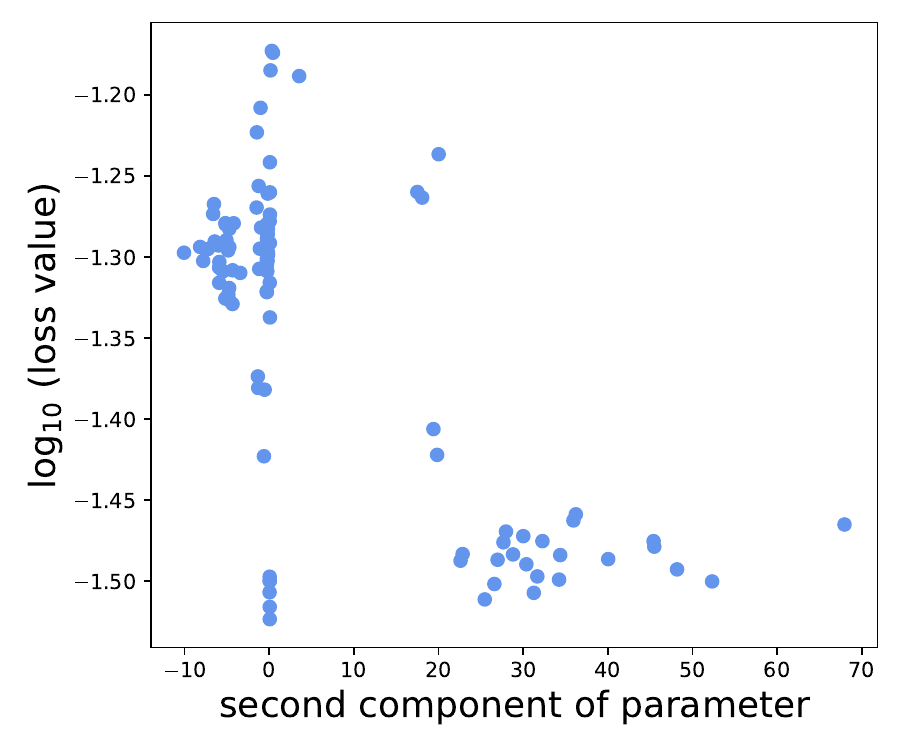}
\subcaption{
Scatter plot of the second parameter and the logarithm of loss values.}
\label{fig:test_loss_vs_param}
\end{minipage}
\end{tabular}
\caption{The distribution of loss values and the relation between the loss value and second parameter on the analysis assuming the existence of self-similarity of the second kind to the data without the self-similarity of the second kind.}
\label{fig:test_loss_param}
\end{figure*}

\begin{figure*}[t]
\begin{tabular}{cc}
\begin{minipage}{0.5\hsize}
\centering
\includegraphics[width=\columnwidth]{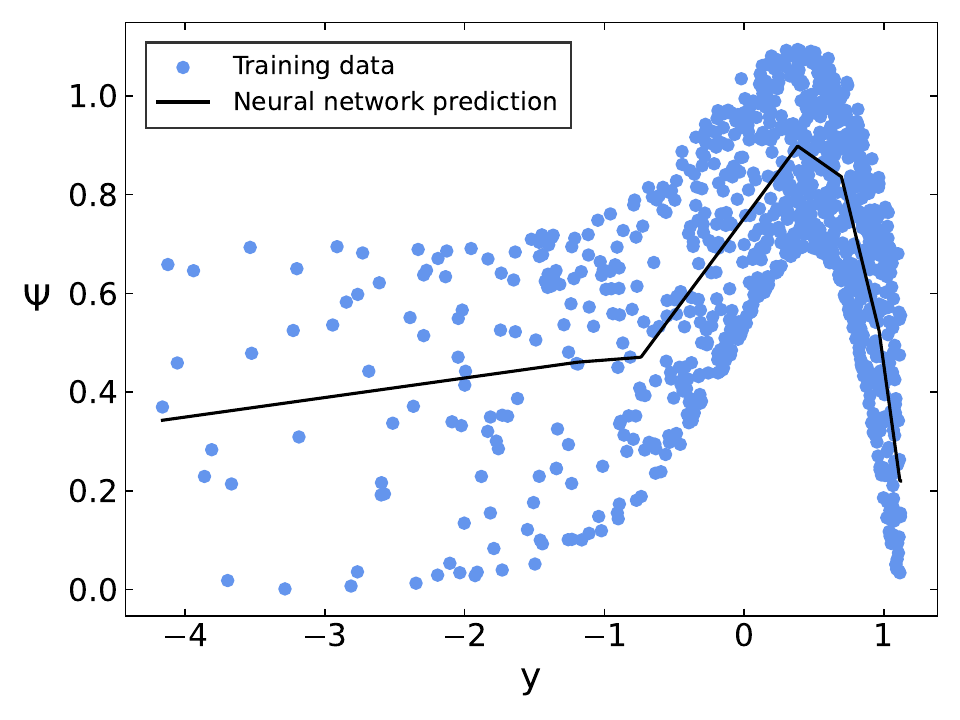}
\subcaption{p =[1,0.03042903]}
\label{fig:test_collapse_0}
\end{minipage}
\begin{minipage}{0.5\hsize}
\includegraphics[width=\columnwidth]{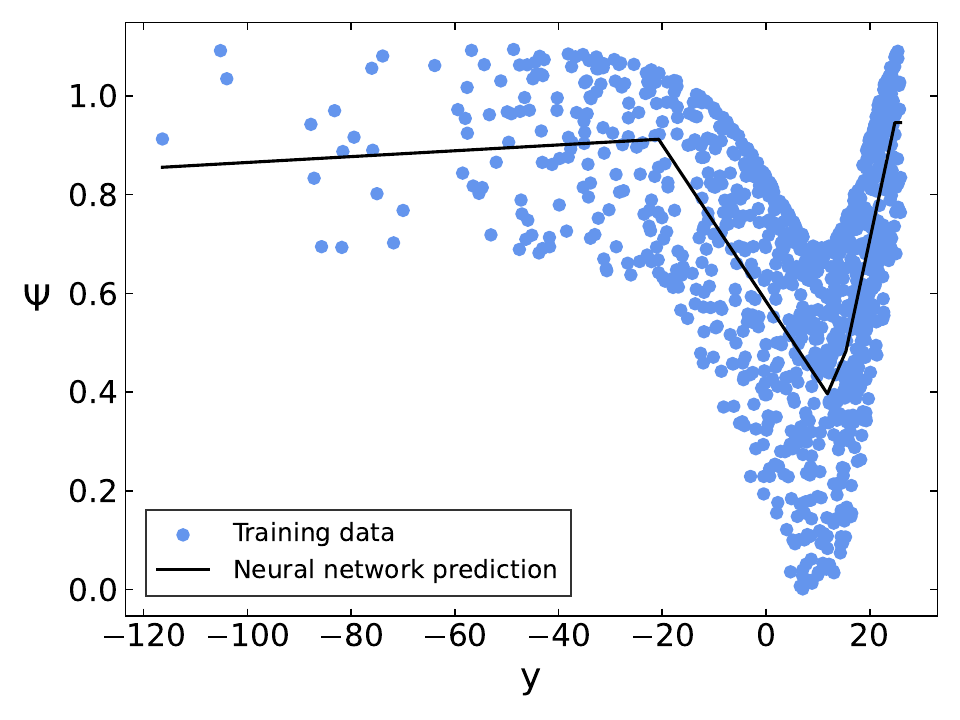}
\subcaption{p=[1, 22.832462]}
\label{fig:test_collapse_30}
\end{minipage}
\end{tabular}
\caption{The examples of training results on the analysis assuming the existence of self-similarity of the second kind to the data without the self-similarity of the second kind, where $y= p \cdot \ln x$ and $\Psi = \ln \Phi$.}
\label{fig:test_collapse}
\end{figure*}

\bibliography{refs}

\begin{thebibliography}{37}%
\makeatletter
\providecommand \@ifxundefined [1]{%
 \@ifx{#1\undefined}
}%
\providecommand \@ifnum [1]{%
 \ifnum #1\expandafter \@firstoftwo
 \else \expandafter \@secondoftwo
 \fi
}%
\providecommand \@ifx [1]{%
 \ifx #1\expandafter \@firstoftwo
 \else \expandafter \@secondoftwo
 \fi
}%
\providecommand \natexlab [1]{#1}%
\providecommand \enquote  [1]{``#1''}%
\providecommand \bibnamefont  [1]{#1}%
\providecommand \bibfnamefont [1]{#1}%
\providecommand \citenamefont [1]{#1}%
\providecommand \href@noop [0]{\@secondoftwo}%
\providecommand \href [0]{\begingroup \@sanitize@url \@href}%
\providecommand \@href[1]{\@@startlink{#1}\@@href}%
\providecommand \@@href[1]{\endgroup#1\@@endlink}%
\providecommand \@sanitize@url [0]{\catcode `\\12\catcode `\$12\catcode `\&12\catcode `\#12\catcode `\^12\catcode `\_12\catcode `\%12\relax}%
\providecommand \@@startlink[1]{}%
\providecommand \@@endlink[0]{}%
\providecommand \url  [0]{\begingroup\@sanitize@url \@url }%
\providecommand \@url [1]{\endgroup\@href {#1}{\urlprefix }}%
\providecommand \urlprefix  [0]{URL }%
\providecommand \Eprint [0]{\href }%
\providecommand \doibase [0]{https://doi.org/}%
\providecommand \selectlanguage [0]{\@gobble}%
\providecommand \bibinfo  [0]{\@secondoftwo}%
\providecommand \bibfield  [0]{\@secondoftwo}%
\providecommand \translation [1]{[#1]}%
\providecommand \BibitemOpen [0]{}%
\providecommand \bibitemStop [0]{}%
\providecommand \bibitemNoStop [0]{.\EOS\space}%
\providecommand \EOS [0]{\spacefactor3000\relax}%
\providecommand \BibitemShut  [1]{\csname bibitem#1\endcsname}%
\let\auto@bib@innerbib\@empty
\bibitem [{\citenamefont {Goldenfeld}(2018)}]{goldenfeld2018lectures}%
  \BibitemOpen
  \bibfield  {author} {\bibinfo {author} {\bibfnamefont {N.}~\bibnamefont {Goldenfeld}},\ }\href {https://doi.org/https://doi.org/10.1201/9780429493492} {\emph {\bibinfo {title} {Lectures on phase transitions and the renormalization group}}}\ (\bibinfo  {publisher} {CRC Press},\ \bibinfo {year} {2018})\BibitemShut {NoStop}%
\bibitem [{\citenamefont {Barenblatt}(1996)}]{barenblatt_1996}%
  \BibitemOpen
  \bibfield  {author} {\bibinfo {author} {\bibfnamefont {G.~I.}\ \bibnamefont {Barenblatt}},\ }\href {https://doi.org/10.1017/CBO9781107050242} {\emph {\bibinfo {title} {Scaling, Self-Similarity, and Intermediate Asymptotics}}},\ Cambridge Texts in Applied Mathematics\ (\bibinfo  {publisher} {Cambridge University Press},\ \bibinfo {year} {1996})\BibitemShut {NoStop}%
\bibitem [{\citenamefont {Barenblatt}(2003)}]{barenblatt_2003}%
  \BibitemOpen
  \bibfield  {author} {\bibinfo {author} {\bibfnamefont {G.~I.}\ \bibnamefont {Barenblatt}},\ }\href {https://doi.org/10.1017/CBO9780511814921} {\emph {\bibinfo {title} {Scaling}}},\ Cambridge Texts in Applied Mathematics\ (\bibinfo  {publisher} {Cambridge University Press},\ \bibinfo {year} {2003})\BibitemShut {NoStop}%
\bibitem [{\citenamefont {de~Gennes}(1979)}]{deGennesScaling}%
  \BibitemOpen
  \bibfield  {author} {\bibinfo {author} {\bibfnamefont {P.-G.}\ \bibnamefont {de~Gennes}},\ }\href {https://doi.org/10.1002/actp.1981.010320517} {\emph {\bibinfo {title} {Scaling concepts in polymer physics}}}\ (\bibinfo  {publisher} {Cornell University Press, New York},\ \bibinfo {year} {1979})\BibitemShut {NoStop}%
\bibitem [{\citenamefont {B\"aumchen}\ \emph {et~al.}(2013)\citenamefont {B\"aumchen}, \citenamefont {Benzaquen}, \citenamefont {Salez}, \citenamefont {McGraw}, \citenamefont {Backholm}, \citenamefont {Fowler}, \citenamefont {Rapha\"el},\ and\ \citenamefont {Dalnoki-Veress}}]{Baumchen}%
  \BibitemOpen
  \bibfield  {author} {\bibinfo {author} {\bibfnamefont {O.}~\bibnamefont {B\"aumchen}}, \bibinfo {author} {\bibfnamefont {M.}~\bibnamefont {Benzaquen}}, \bibinfo {author} {\bibfnamefont {T.}~\bibnamefont {Salez}}, \bibinfo {author} {\bibfnamefont {J.~D.}\ \bibnamefont {McGraw}}, \bibinfo {author} {\bibfnamefont {M.}~\bibnamefont {Backholm}}, \bibinfo {author} {\bibfnamefont {P.}~\bibnamefont {Fowler}}, \bibinfo {author} {\bibfnamefont {E.}~\bibnamefont {Rapha\"el}},\ and\ \bibinfo {author} {\bibfnamefont {K.}~\bibnamefont {Dalnoki-Veress}},\ }\bibfield  {title} {\bibinfo {title} {Relaxation and intermediate asymptotics of a rectangular trench in a viscous film},\ }\href {https://doi.org/10.1103/PhysRevE.88.035001} {\bibfield  {journal} {\bibinfo  {journal} {Phys. Rev. E}\ }\textbf {\bibinfo {volume} {88}},\ \bibinfo {pages} {035001} (\bibinfo {year} {2013})}\BibitemShut {NoStop}%
\bibitem [{\citenamefont {Hatano}(2008)}]{Hatano}%
  \BibitemOpen
  \bibfield  {author} {\bibinfo {author} {\bibfnamefont {T.}~\bibnamefont {Hatano}},\ }\bibfield  {title} {\bibinfo {title} {Scaling properties of granular rheology near the jamming transition},\ }\href {https://doi.org/10.1143/JPSJ.77.123002} {\bibfield  {journal} {\bibinfo  {journal} {Journal of the Physical Society of Japan}\ }\textbf {\bibinfo {volume} {77}},\ \bibinfo {pages} {123002} (\bibinfo {year} {2008})}\BibitemShut {NoStop}%
\bibitem [{\citenamefont {Saitoh}\ and\ \citenamefont {Kawasaki}(2020)}]{Saitoh}%
  \BibitemOpen
  \bibfield  {author} {\bibinfo {author} {\bibfnamefont {K.}~\bibnamefont {Saitoh}}\ and\ \bibinfo {author} {\bibfnamefont {T.}~\bibnamefont {Kawasaki}},\ }\bibfield  {title} {\bibinfo {title} {Critical scaling of diffusion coefficients and size of rigid clusters of soft athermal particles under shear},\ }\href {https://doi.org/10.3389/fphy.2020.00099} {\bibfield  {journal} {\bibinfo  {journal} {Frontiers in Physics}\ }\textbf {\bibinfo {volume} {8}},\ \bibinfo {pages} {99} (\bibinfo {year} {2020})}\BibitemShut {NoStop}%
\bibitem [{\citenamefont {Barenblatt}(2014)}]{Barenblatt_2014}%
  \BibitemOpen
  \bibfield  {author} {\bibinfo {author} {\bibfnamefont {G.~I.}\ \bibnamefont {Barenblatt}},\ }\href {https://doi.org/10.1017/CBO9781139030014} {\emph {\bibinfo {title} {Flow, Deformation and Fracture}}},\ Cambridge Texts in Applied Mathematics\ (\bibinfo  {publisher} {Cambridge University Press},\ \bibinfo {year} {2014})\BibitemShut {NoStop}%
\bibitem [{\citenamefont {Mandelbrot}(1983)}]{Mandelbrot}%
  \BibitemOpen
  \bibfield  {author} {\bibinfo {author} {\bibfnamefont {B.}~\bibnamefont {Mandelbrot}},\ }\href@noop {} {\emph {\bibinfo {title} {The Fractal Geometry of Nature}}}\ (\bibinfo  {publisher} {Macmillan, New York},\ \bibinfo {year} {1983})\BibitemShut {NoStop}%
\bibitem [{\citenamefont {Stanley}(1999)}]{Stanley_1999}%
  \BibitemOpen
  \bibfield  {author} {\bibinfo {author} {\bibfnamefont {H.~E.}\ \bibnamefont {Stanley}},\ }\bibfield  {title} {\bibinfo {title} {Scaling, universality, and renormalization: Three pillars of modern critical phenomena},\ }\href {https://doi.org/10.1103/RevModPhys.71.S358} {\bibfield  {journal} {\bibinfo  {journal} {Reviews of modern physics}\ }\textbf {\bibinfo {volume} {71}},\ \bibinfo {pages} {S358} (\bibinfo {year} {1999})}\BibitemShut {NoStop}%
\bibitem [{\citenamefont {Cabella}\ \emph {et~al.}(2011)\citenamefont {Cabella}, \citenamefont {Martinez},\ and\ \citenamefont {Ribeiro}}]{Cabella}%
  \BibitemOpen
  \bibfield  {author} {\bibinfo {author} {\bibfnamefont {B.~C.~T.}\ \bibnamefont {Cabella}}, \bibinfo {author} {\bibfnamefont {A.~S.}\ \bibnamefont {Martinez}},\ and\ \bibinfo {author} {\bibfnamefont {F.}~\bibnamefont {Ribeiro}},\ }\bibfield  {title} {\bibinfo {title} {Data collapse, scaling functions, and analytical solutions of generalized growth models},\ }\href {https://doi.org/10.1103/PhysRevE.83.061902} {\bibfield  {journal} {\bibinfo  {journal} {Phys. Rev. E}\ }\textbf {\bibinfo {volume} {83}},\ \bibinfo {pages} {061902} (\bibinfo {year} {2011})}\BibitemShut {NoStop}%
\bibitem [{\citenamefont {Kimchi}\ \emph {et~al.}(2018)\citenamefont {Kimchi}, \citenamefont {Sheckelton}, \citenamefont {McQueen},\ and\ \citenamefont {Lee}}]{Kimchi}%
  \BibitemOpen
  \bibfield  {author} {\bibinfo {author} {\bibfnamefont {I.}~\bibnamefont {Kimchi}}, \bibinfo {author} {\bibfnamefont {J.~P.}\ \bibnamefont {Sheckelton}}, \bibinfo {author} {\bibfnamefont {T.~M.}\ \bibnamefont {McQueen}},\ and\ \bibinfo {author} {\bibfnamefont {P.~A.}\ \bibnamefont {Lee}},\ }\bibfield  {title} {\bibinfo {title} {Scaling and data collapse from local moments in frustrated disordered quantum spin systems},\ }\href {https://doi.org/https://doi.org/10.1038/s41467-018-06800-2} {\bibfield  {journal} {\bibinfo  {journal} {Nature communications}\ }\textbf {\bibinfo {volume} {9}},\ \bibinfo {pages} {4367} (\bibinfo {year} {2018})}\BibitemShut {NoStop}%
\bibitem [{\citenamefont {Yokota}\ and\ \citenamefont {Okumura}(2011)}]{Yokota}%
  \BibitemOpen
  \bibfield  {author} {\bibinfo {author} {\bibfnamefont {M.}~\bibnamefont {Yokota}}\ and\ \bibinfo {author} {\bibfnamefont {K.}~\bibnamefont {Okumura}},\ }\bibfield  {title} {\bibinfo {title} {Dimensional crossover in the coalescence dynamics of viscous drops confined in between two plates},\ }\href {https://doi.org/10.1073/pnas.1017112108} {\bibfield  {journal} {\bibinfo  {journal} {Proceedings of the National Academy of Sciences}\ }\textbf {\bibinfo {volume} {108}},\ \bibinfo {pages} {6395} (\bibinfo {year} {2011})}\BibitemShut {NoStop}%
\bibitem [{\citenamefont {Murano}\ and\ \citenamefont {Okumura}(2020)}]{Okumura2020}%
  \BibitemOpen
  \bibfield  {author} {\bibinfo {author} {\bibfnamefont {M.}~\bibnamefont {Murano}}\ and\ \bibinfo {author} {\bibfnamefont {K.}~\bibnamefont {Okumura}},\ }\bibfield  {title} {\bibinfo {title} {Rising bubble in a cell with a high aspect ratio cross-section filled with a viscous fluid and its connection to viscous fingering},\ }\href {https://doi.org/10.1103/PhysRevResearch.2.013188} {\bibfield  {journal} {\bibinfo  {journal} {Phys. Rev. Res.}\ }\textbf {\bibinfo {volume} {2}},\ \bibinfo {pages} {013188} (\bibinfo {year} {2020})}\BibitemShut {NoStop}%
\bibitem [{\citenamefont {Goldenfeld}\ \emph {et~al.}(1989)\citenamefont {Goldenfeld}, \citenamefont {Martin},\ and\ \citenamefont {Oono}}]{goldenfeld1989intermediate}%
  \BibitemOpen
  \bibfield  {author} {\bibinfo {author} {\bibfnamefont {N.}~\bibnamefont {Goldenfeld}}, \bibinfo {author} {\bibfnamefont {O.}~\bibnamefont {Martin}},\ and\ \bibinfo {author} {\bibfnamefont {Y.}~\bibnamefont {Oono}},\ }\bibfield  {title} {\bibinfo {title} {Intermediate asymptotics and renormalization group theory},\ }\href {https://doi.org/https://doi.org/10.1007/BF01060993} {\bibfield  {journal} {\bibinfo  {journal} {Journal of scientific computing}\ }\textbf {\bibinfo {volume} {4}},\ \bibinfo {pages} {355} (\bibinfo {year} {1989})}\BibitemShut {NoStop}%
\bibitem [{\citenamefont {Goodfellow}\ \emph {et~al.}(2016)\citenamefont {Goodfellow}, \citenamefont {Bengio},\ and\ \citenamefont {Courville}}]{Goodfellow-et-al-2016}%
  \BibitemOpen
  \bibfield  {author} {\bibinfo {author} {\bibfnamefont {I.}~\bibnamefont {Goodfellow}}, \bibinfo {author} {\bibfnamefont {Y.}~\bibnamefont {Bengio}},\ and\ \bibinfo {author} {\bibfnamefont {A.}~\bibnamefont {Courville}},\ }\href@noop {} {\emph {\bibinfo {title} {Deep Learning}}}\ (\bibinfo  {publisher} {MIT Press},\ \bibinfo {year} {2016})\ \bibinfo {note} {\url{http://www.deeplearningbook.org}}\BibitemShut {NoStop}%
\bibitem [{\citenamefont {Bishop}\ and\ \citenamefont {Bishop}(2024)}]{bishop2024deep}%
  \BibitemOpen
  \bibfield  {author} {\bibinfo {author} {\bibfnamefont {C.~M.}\ \bibnamefont {Bishop}}\ and\ \bibinfo {author} {\bibfnamefont {H.}~\bibnamefont {Bishop}},\ }\href {https://doi.org/https://doi.org/10.1007/978-3-031-45468-4} {\emph {\bibinfo {title} {Deep learning: foundations and concepts}}}\ (\bibinfo  {publisher} {Springer},\ \bibinfo {year} {2024})\BibitemShut {NoStop}%
\bibitem [{\citenamefont {Cybenko}(1989)}]{cybenko1989approximation}%
  \BibitemOpen
  \bibfield  {author} {\bibinfo {author} {\bibfnamefont {G.}~\bibnamefont {Cybenko}},\ }\bibfield  {title} {\bibinfo {title} {Approximation by superpositions of a sigmoidal function},\ }\href {https://doi.org/https://doi.org/10.1007/BF02551274} {\bibfield  {journal} {\bibinfo  {journal} {Mathematics of control, signals and systems}\ }\textbf {\bibinfo {volume} {2}},\ \bibinfo {pages} {303} (\bibinfo {year} {1989})}\BibitemShut {NoStop}%
\bibitem [{\citenamefont {Hornik}\ \emph {et~al.}(1989)\citenamefont {Hornik}, \citenamefont {Stinchcombe},\ and\ \citenamefont {White}}]{HORNIK1989359}%
  \BibitemOpen
  \bibfield  {author} {\bibinfo {author} {\bibfnamefont {K.}~\bibnamefont {Hornik}}, \bibinfo {author} {\bibfnamefont {M.}~\bibnamefont {Stinchcombe}},\ and\ \bibinfo {author} {\bibfnamefont {H.}~\bibnamefont {White}},\ }\bibfield  {title} {\bibinfo {title} {Multilayer feedforward networks are universal approximators},\ }\href {https://doi.org/https://doi.org/10.1016/0893-6080(89)90020-8} {\bibfield  {journal} {\bibinfo  {journal} {Neural Networks}\ }\textbf {\bibinfo {volume} {2}},\ \bibinfo {pages} {359} (\bibinfo {year} {1989})}\BibitemShut {NoStop}%
\bibitem [{\citenamefont {Wang}\ \emph {et~al.}(2023)\citenamefont {Wang}, \citenamefont {Fu}, \citenamefont {Du}, \citenamefont {Gao}, \citenamefont {Huang}, \citenamefont {Liu}, \citenamefont {Chandak}, \citenamefont {Liu}, \citenamefont {Van~Katwyk}, \citenamefont {Deac} \emph {et~al.}}]{wang2023scientific}%
  \BibitemOpen
  \bibfield  {author} {\bibinfo {author} {\bibfnamefont {H.}~\bibnamefont {Wang}}, \bibinfo {author} {\bibfnamefont {T.}~\bibnamefont {Fu}}, \bibinfo {author} {\bibfnamefont {Y.}~\bibnamefont {Du}}, \bibinfo {author} {\bibfnamefont {W.}~\bibnamefont {Gao}}, \bibinfo {author} {\bibfnamefont {K.}~\bibnamefont {Huang}}, \bibinfo {author} {\bibfnamefont {Z.}~\bibnamefont {Liu}}, \bibinfo {author} {\bibfnamefont {P.}~\bibnamefont {Chandak}}, \bibinfo {author} {\bibfnamefont {S.}~\bibnamefont {Liu}}, \bibinfo {author} {\bibfnamefont {P.}~\bibnamefont {Van~Katwyk}}, \bibinfo {author} {\bibfnamefont {A.}~\bibnamefont {Deac}}, \emph {et~al.},\ }\bibfield  {title} {\bibinfo {title} {Scientific discovery in the age of artificial intelligence},\ }\href {https://doi.org/10.1038/s41586-023-06221-2} {\bibfield  {journal} {\bibinfo  {journal} {Nature}\ }\textbf {\bibinfo {volume} {620}},\ \bibinfo {pages} {47} (\bibinfo {year} {2023})}\BibitemShut {NoStop}%
\bibitem [{\citenamefont {Carleo}\ and\ \citenamefont {Troyer}(2017)}]{doi:10.1126/science.aag2302}%
  \BibitemOpen
  \bibfield  {author} {\bibinfo {author} {\bibfnamefont {G.}~\bibnamefont {Carleo}}\ and\ \bibinfo {author} {\bibfnamefont {M.}~\bibnamefont {Troyer}},\ }\bibfield  {title} {\bibinfo {title} {Solving the quantum many-body problem with artificial neural networks},\ }\href {https://doi.org/10.1126/science.aag2302} {\bibfield  {journal} {\bibinfo  {journal} {Science}\ }\textbf {\bibinfo {volume} {355}},\ \bibinfo {pages} {602} (\bibinfo {year} {2017})}\BibitemShut {NoStop}%
\bibitem [{\citenamefont {Raissi}\ \emph {et~al.}(2019)\citenamefont {Raissi}, \citenamefont {Perdikaris},\ and\ \citenamefont {Karniadakis}}]{RAISSI2019686}%
  \BibitemOpen
  \bibfield  {author} {\bibinfo {author} {\bibfnamefont {M.}~\bibnamefont {Raissi}}, \bibinfo {author} {\bibfnamefont {P.}~\bibnamefont {Perdikaris}},\ and\ \bibinfo {author} {\bibfnamefont {G.}~\bibnamefont {Karniadakis}},\ }\bibfield  {title} {\bibinfo {title} {Physics-informed neural networks: A deep learning framework for solving forward and inverse problems involving nonlinear partial differential equations},\ }\href {https://doi.org/https://doi.org/10.1016/j.jcp.2018.10.045} {\bibfield  {journal} {\bibinfo  {journal} {Journal of Computational Physics}\ }\textbf {\bibinfo {volume} {378}},\ \bibinfo {pages} {686} (\bibinfo {year} {2019})}\BibitemShut {NoStop}%
\bibitem [{\citenamefont {Tanaka}\ and\ \citenamefont {Tomiya}(2017)}]{doi:10.7566/JPSJ.86.063001}%
  \BibitemOpen
  \bibfield  {author} {\bibinfo {author} {\bibfnamefont {A.}~\bibnamefont {Tanaka}}\ and\ \bibinfo {author} {\bibfnamefont {A.}~\bibnamefont {Tomiya}},\ }\bibfield  {title} {\bibinfo {title} {Detection of phase transition via convolutional neural networks},\ }\href {https://doi.org/10.7566/JPSJ.86.063001} {\bibfield  {journal} {\bibinfo  {journal} {Journal of the Physical Society of Japan}\ }\textbf {\bibinfo {volume} {86}},\ \bibinfo {pages} {063001} (\bibinfo {year} {2017})}\BibitemShut {NoStop}%
\bibitem [{rep()}]{repository}%
  \BibitemOpen
  \href@noop {} {\bibinfo {title} {{GitHub repository for python codes for this paper}}},\ \bibinfo {howpublished} {\url{https://github.com/RyotaWatanabe7/Self-similarity-finder}}\BibitemShut {NoStop}%
\bibitem [{\citenamefont {Bhattacharjee}\ and\ \citenamefont {Seno}(2001)}]{Somendra_M_Bhattacharjee_2001}%
  \BibitemOpen
  \bibfield  {author} {\bibinfo {author} {\bibfnamefont {S.~M.}\ \bibnamefont {Bhattacharjee}}\ and\ \bibinfo {author} {\bibfnamefont {F.}~\bibnamefont {Seno}},\ }\bibfield  {title} {\bibinfo {title} {A measure of data collapse for scaling},\ }\href {https://doi.org/10.1088/0305-4470/34/33/302} {\bibfield  {journal} {\bibinfo  {journal} {Journal of Physics A: Mathematical and General}\ }\textbf {\bibinfo {volume} {34}},\ \bibinfo {pages} {6375} (\bibinfo {year} {2001})}\BibitemShut {NoStop}%
\bibitem [{\citenamefont {Maruoka}(2023)}]{Maruoka_2023}%
  \BibitemOpen
  \bibfield  {author} {\bibinfo {author} {\bibfnamefont {H.}~\bibnamefont {Maruoka}},\ }\bibfield  {title} {\bibinfo {title} {A framework for crossover of scaling law as a self-similar solution: dynamical impact of viscoelastic board},\ }\href {https://doi.org/10.1140/epje/s10189-023-00292-9} {\bibfield  {journal} {\bibinfo  {journal} {The European Physical Journal E}\ }\textbf {\bibinfo {volume} {46}},\ \bibinfo {pages} {35} (\bibinfo {year} {2023})}\BibitemShut {NoStop}%
\bibitem [{\citenamefont {Wang}\ and\ \citenamefont {Yu}(2023)}]{wang2023physicsguided}%
  \BibitemOpen
  \bibfield  {author} {\bibinfo {author} {\bibfnamefont {R.}~\bibnamefont {Wang}}\ and\ \bibinfo {author} {\bibfnamefont {R.}~\bibnamefont {Yu}},\ }\href@noop {} {\bibinfo {title} {Physics-guided deep learning for dynamical systems: A survey}} (\bibinfo {year} {2023}),\ \Eprint {https://arxiv.org/abs/2107.01272} {arXiv:2107.01272 [cs.LG]} \BibitemShut {NoStop}%
\bibitem [{\citenamefont {Otto}\ \emph {et~al.}(2023)\citenamefont {Otto}, \citenamefont {Zolman}, \citenamefont {Kutz},\ and\ \citenamefont {Brunton}}]{otto2023unified}%
  \BibitemOpen
  \bibfield  {author} {\bibinfo {author} {\bibfnamefont {S.~E.}\ \bibnamefont {Otto}}, \bibinfo {author} {\bibfnamefont {N.}~\bibnamefont {Zolman}}, \bibinfo {author} {\bibfnamefont {J.~N.}\ \bibnamefont {Kutz}},\ and\ \bibinfo {author} {\bibfnamefont {S.~L.}\ \bibnamefont {Brunton}},\ }\href@noop {} {\bibinfo {title} {A unified framework to enforce, discover, and promote symmetry in machine learning}} (\bibinfo {year} {2023}),\ \Eprint {https://arxiv.org/abs/2311.00212} {arXiv:2311.00212 [cs.LG]} \BibitemShut {NoStop}%
\bibitem [{\citenamefont {Barenboim}\ \emph {et~al.}(2021)\citenamefont {Barenboim}, \citenamefont {Hirn},\ and\ \citenamefont {Sanz}}]{10.21468/SciPostPhys.11.1.014}%
  \BibitemOpen
  \bibfield  {author} {\bibinfo {author} {\bibfnamefont {G.}~\bibnamefont {Barenboim}}, \bibinfo {author} {\bibfnamefont {J.}~\bibnamefont {Hirn}},\ and\ \bibinfo {author} {\bibfnamefont {V.}~\bibnamefont {Sanz}},\ }\bibfield  {title} {\bibinfo {title} {{Symmetry meets AI}},\ }\href {https://doi.org/10.21468/SciPostPhys.11.1.014} {\bibfield  {journal} {\bibinfo  {journal} {SciPost Phys.}\ }\textbf {\bibinfo {volume} {11}},\ \bibinfo {pages} {014} (\bibinfo {year} {2021})}\BibitemShut {NoStop}%
\bibitem [{\citenamefont {Krippendorf}\ and\ \citenamefont {Syvaeri}(2020)}]{Krippendorf_2021}%
  \BibitemOpen
  \bibfield  {author} {\bibinfo {author} {\bibfnamefont {S.}~\bibnamefont {Krippendorf}}\ and\ \bibinfo {author} {\bibfnamefont {M.}~\bibnamefont {Syvaeri}},\ }\bibfield  {title} {\bibinfo {title} {Detecting symmetries with neural networks},\ }\href {https://doi.org/10.1088/2632-2153/abbd2d} {\bibfield  {journal} {\bibinfo  {journal} {Machine Learning: Science and Technology}\ }\textbf {\bibinfo {volume} {2}},\ \bibinfo {pages} {015010} (\bibinfo {year} {2020})}\BibitemShut {NoStop}%
\bibitem [{\citenamefont {Yang}\ \emph {et~al.}(2023)\citenamefont {Yang}, \citenamefont {Walters}, \citenamefont {Dehmamy},\ and\ \citenamefont {Yu}}]{pmlr-v202-yang23n}%
  \BibitemOpen
  \bibfield  {author} {\bibinfo {author} {\bibfnamefont {J.}~\bibnamefont {Yang}}, \bibinfo {author} {\bibfnamefont {R.}~\bibnamefont {Walters}}, \bibinfo {author} {\bibfnamefont {N.}~\bibnamefont {Dehmamy}},\ and\ \bibinfo {author} {\bibfnamefont {R.}~\bibnamefont {Yu}},\ }\bibfield  {title} {\bibinfo {title} {Generative adversarial symmetry discovery},\ }in\ \href {https://proceedings.mlr.press/v202/yang23n.html} {\emph {\bibinfo {booktitle} {Proceedings of the 40th International Conference on Machine Learning}}},\ \bibinfo {series} {Proceedings of Machine Learning Research}, Vol.\ \bibinfo {volume} {202},\ \bibinfo {editor} {edited by\ \bibinfo {editor} {\bibfnamefont {A.}~\bibnamefont {Krause}}, \bibinfo {editor} {\bibfnamefont {E.}~\bibnamefont {Brunskill}}, \bibinfo {editor} {\bibfnamefont {K.}~\bibnamefont {Cho}}, \bibinfo {editor} {\bibfnamefont {B.}~\bibnamefont {Engelhardt}}, \bibinfo {editor} {\bibfnamefont {S.}~\bibnamefont {Sabato}},\ and\ \bibinfo {editor} {\bibfnamefont {J.}~\bibnamefont
  {Scarlett}}}\ (\bibinfo  {publisher} {PMLR},\ \bibinfo {year} {2023})\ pp.\ \bibinfo {pages} {39488--39508}\BibitemShut {NoStop}%
\bibitem [{\citenamefont {Desai}\ \emph {et~al.}(2022)\citenamefont {Desai}, \citenamefont {Nachman},\ and\ \citenamefont {Thaler}}]{PhysRevD.105.096031}%
  \BibitemOpen
  \bibfield  {author} {\bibinfo {author} {\bibfnamefont {K.}~\bibnamefont {Desai}}, \bibinfo {author} {\bibfnamefont {B.}~\bibnamefont {Nachman}},\ and\ \bibinfo {author} {\bibfnamefont {J.}~\bibnamefont {Thaler}},\ }\bibfield  {title} {\bibinfo {title} {Symmetry discovery with deep learning},\ }\href {https://doi.org/10.1103/PhysRevD.105.096031} {\bibfield  {journal} {\bibinfo  {journal} {Phys. Rev. D}\ }\textbf {\bibinfo {volume} {105}},\ \bibinfo {pages} {096031} (\bibinfo {year} {2022})}\BibitemShut {NoStop}%
\bibitem [{\citenamefont {Olver}(1986)}]{MR836734}%
  \BibitemOpen
  \bibfield  {author} {\bibinfo {author} {\bibfnamefont {P.~J.}\ \bibnamefont {Olver}},\ }\href {https://doi.org/10.1007/978-1-4684-0274-2} {\emph {\bibinfo {title} {Applications of {L}ie groups to differential equations}}},\ \bibinfo {series} {Graduate Texts in Mathematics}, Vol.\ \bibinfo {volume} {107}\ (\bibinfo  {publisher} {Springer-Verlag},\ \bibinfo {address} {New York},\ \bibinfo {year} {1986})\ pp.\ \bibinfo {pages} {xxvi+497}\BibitemShut {NoStop}%
\bibitem [{\citenamefont {Johnson}(1985)}]{Johnson_1985}%
  \BibitemOpen
  \bibfield  {author} {\bibinfo {author} {\bibfnamefont {K.}~\bibnamefont {Johnson}},\ }\href {https://doi.org/10.1017/CBO9781139171731} {\emph {\bibinfo {title} {Contact Mechanics}}},\ Cambridge Texts in Applied Mathematics\ (\bibinfo  {publisher} {Cambridge University Press},\ \bibinfo {year} {1985})\BibitemShut {NoStop}%
\bibitem [{\citenamefont {Chastel}\ \emph {et~al.}(2016)\citenamefont {Chastel}, \citenamefont {Gondret},\ and\ \citenamefont {Mongruel}}]{Chastel_2016}%
  \BibitemOpen
  \bibfield  {author} {\bibinfo {author} {\bibfnamefont {T.}~\bibnamefont {Chastel}}, \bibinfo {author} {\bibfnamefont {P.}~\bibnamefont {Gondret}},\ and\ \bibinfo {author} {\bibfnamefont {A.}~\bibnamefont {Mongruel}},\ }\bibfield  {title} {\bibinfo {title} {Texture-driven elastohydrodynamic bouncing},\ }\href {https://doi.org/10.1017/jfm.2016.580} {\bibfield  {journal} {\bibinfo  {journal} {Journal of Fluid Mechanics}\ }\textbf {\bibinfo {volume} {805}},\ \bibinfo {pages} {577 } (\bibinfo {year} {2016})}\BibitemShut {NoStop}%
\bibitem [{\citenamefont {Chastel}\ and\ \citenamefont {Mongruel}(2019)}]{Chastel_2019}%
  \BibitemOpen
  \bibfield  {author} {\bibinfo {author} {\bibfnamefont {T.}~\bibnamefont {Chastel}}\ and\ \bibinfo {author} {\bibfnamefont {A.}~\bibnamefont {Mongruel}},\ }\bibfield  {title} {\bibinfo {title} {Sticking collision between a sphere and a textured wall in a viscous fluid},\ }\href {https://doi.org/10.1103/PhysRevFluids.4.014301} {\bibfield  {journal} {\bibinfo  {journal} {Physical Review Fluids}\ }\textbf {\bibinfo {volume} {4}},\ \bibinfo {pages} {014301} (\bibinfo {year} {2019})}\BibitemShut {NoStop}%
\bibitem [{\citenamefont {Mongruel}\ and\ \citenamefont {Gondret}(2020)}]{Mongruel_2020}%
  \BibitemOpen
  \bibfield  {author} {\bibinfo {author} {\bibfnamefont {A.}~\bibnamefont {Mongruel}}\ and\ \bibinfo {author} {\bibfnamefont {P.}~\bibnamefont {Gondret}},\ }\bibfield  {title} {\bibinfo {title} {Viscous dissipation in the collision between a sphere and a textured wall},\ }\href {https://doi.org/10.1017/jfm.2020.325} {\bibfield  {journal} {\bibinfo  {journal} {Journal of Fluid Mechanics}\ }\textbf {\bibinfo {volume} {896}},\ \bibinfo {pages} {A8} (\bibinfo {year} {2020})}\BibitemShut {NoStop}%
\end{thebibliography}%

\end{document}